\documentclass[
    10pt,
    showpacs,
    twocolumn,
    pra,
    superscriptaddress,
    notitlepage,
]{revtex4-2}

\usepackage[a4paper, scale=0.85]{geometry}
\usepackage{amsmath, amssymb, amsthm, mathrsfs, amsfonts, dsfont}
\usepackage{physics}
\usepackage{bm}
\usepackage{booktabs}

\usepackage[dvips]{graphicx}
\graphicspath{{./figure/}}
\usepackage{subfigure, epsfig}
\usepackage{tikz}
\usetikzlibrary{arrows, shapes, fadings, snakes, decorations.pathmorphing, patterns, calc, positioning}
\usepackage{array}
\usepackage{adjustbox}
\usepackage[normalem]{ulem}
\usepackage{float}

\usepackage{multirow}
\usepackage{diagbox}
\usepackage{algorithm}
\usepackage{algorithmicx}
\usepackage{algpseudocode}

\usepackage{qcircuit}

\newtheorem{prop}{Proposition}

\newcommand{\comments}[1]{}

\usepackage{caption}
\usepackage{ragged2e}

\newcommand{\orcidlink}[1]{\href{https://orcid.org/#1}{\includegraphics[width=8pt]{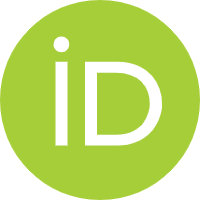}}}

\usepackage{hyperref}
\hypersetup{
   colorlinks=true,
   linkcolor=purple,
   urlcolor=purple,
   citecolor=purple
}

\makeatletter
\AtBeginDocument{

}
\makeatother

\begin{document}
\title{Circuit optimization of qubit IC-POVMs for shadow estimation}

\author{Zhou You\orcidlink{0000-0002-6140-2092}}
\affiliation{Key Laboratory for Information Science of Electromagnetic Waves (Ministry of Education), Fudan University, Shanghai 200433, China}

\author{Qing Liu\orcidlink{0000-0002-1576-1975}}
\affiliation{Key Laboratory for Information Science of Electromagnetic Waves (Ministry of Education), Fudan University, Shanghai 200433, China}

\author{You Zhou\orcidlink{0000-0003-0886-077X}}
\email{you\_zhou@fudan.edu.cn}
\affiliation{Key Laboratory for Information Science of Electromagnetic Waves (Ministry of Education), Fudan University, Shanghai 200433, China}

\begin{abstract}
Extracting information from quantum systems is crucial in quantum physics and information processing. Methods based on randomized measurements, like shadow estimation, show advantages in effectively achieving such tasks. However, randomized measurements require the application of random unitary evolution, which unavoidably necessitates frequent adjustments to the experimental setup or circuit parameters, posing challenges for practical implementations. To address these limitations, positive operator-valued measurements (POVMs) have been integrated to realize real-time single-setting shadow estimation. In this work, we advance the POVM-based shadow estimation by reducing the CNOT gate count for the implementation circuits of informationally complete POVMs (IC-POVMs), in particular, the symmetric IC-POVMs (SIC-POVMs), through the dimension dilation framework. We show that any single-qubit minimal IC-POVM can be implemented using at most 2 CNOT gates, while an SIC-POVM can be implemented with only 1 CNOT gate. In particular, we provide a concise form of the compilation circuit of any SIC-POVM along with an efficient algorithm for the determination of gate parameters. Moreover, we apply the optimized circuit compilation to shadow estimation, showcasing its noise-resilient performance and highlighting the flexibility in compiling various SIC-POVMs. Our work paves the way for the practical applications of qubit IC-POVMs on quantum platforms.
\end{abstract}
\maketitle

\section{Introduction}
Estimating the properties of quantum systems constitutes a fundamental task in various fields, ranging from quantum physics to quantum information processing.
However, as the degree of freedom, such as the number of qubits in quantum systems, increases, the exponentially large Hilbert space faces significant challenges in quantum measurement and benchmarking \cite{preskill2018quantum,bharti2022noisy}. Moreover, efficient methods are urgently needed to characterize the noise and performance of quantum computing platforms \cite{eisert2020quantum,kliesch2021theory}.

A recent significant advancement along this direction is shadow estimation \cite{huang2020predicting}, which enables the simultaneous estimation of many properties of the state $\rho$ using the shadow set $\{\hat{\rho}_i\}$ collected from randomized measurements \cite{elben2023randomized}. Here, each snapshot $\hat{\rho}_i$ is an unbiased estimator of $\rho$, which is constructed from the projective measurement result $b_i$ in the basis determined by the unitary $U_i$ sampled from some tomographic-(over)complete random unitary ensemble $\mathcal{U}$ \cite{hu2023classical,akhtar2023scalable,bertoni2024shallow,zhang2024minimal,imai2024collective}.
Based on this framework, there are various developments in the applications, from near-term quantum algorithms \cite{sack2022avoiding,boyd2022training} and quantum error mitigation \cite{Hu2022Logical,seif2023shadow,peng2024experimental,zhou2024hybrid,liu2024auxiliary}
to quantum correlation detection \cite{elben2020mixed,rath2021quantum,liu2022detecting} and quantum chaos diagnosis \cite{garcia2021quantum,mcginley2022quantifying}. Nevertheless, the routine of sampling from some random unitary ensemble introduces considerable complexity in practical realizations. On the one hand, such sampling requires considerable resources in randomness generation and essential compilations. On the other hand, it demands frequent changes of the experimental setups or circuit parameters. Moreover, each shot $b_i$ is highly correlated with the corresponding $U_i \in \mathcal{U}$, which could make the shadow estimation vulnerable to noises due to the frequent changes of basis, especially the gate-dependent noises \cite{wallman2018randomized,chen2021robust,brieger2023stability}.

\begin{figure}[!ht]
\centering
\includegraphics[width=0.5\textwidth]{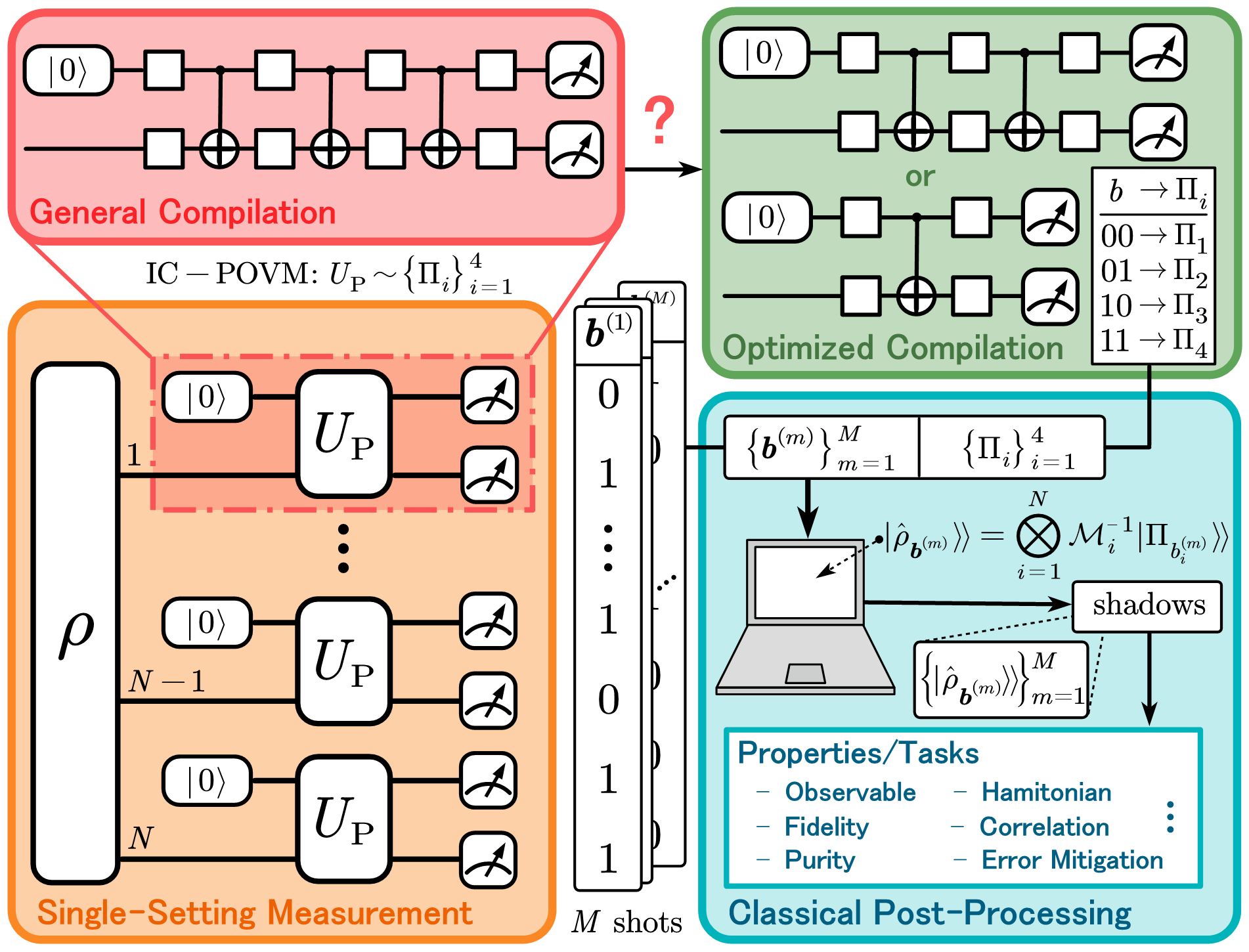}
\caption{\justifying{Framework for circuit optimization of qubit IC-POVM for shadow estimation. The bottom-left (orange) part illustrates the quantum measurement stage, where IC-POVM is implemented for each qubit using an auxiliary qubit and a suitable $U_{\mathrm{P}}$, as explained in Sec.~\ref{sec:preli}. The top-left (red) part shows the conventional compiled circuit structure of $U_{\mathrm{P}}$, which typically requires 3 CNOT gates. The top-right (green) part represents the main contribution of this work, as discussed in Sec.~\ref{sec:reduce}, which reduces the CNOT gate count to 2 or 1 while still implementing the same IC-POVM. The bottom-right (blue) part depicts the classical post-processing stage in the shadow estimation, which is investigated in Sec.~\ref{sec:appli} to demonstrate the advantages of the optimized circuit.}}
\label{fig:overall}
\end{figure}

General IC-POVMs beyond the projective measurements are introduced to overcome such limitations to realize real-time single-setting shadow estimation \cite{acharya2021shadow,nguyen2022optimizing,innocenti2023shadow,garcia2021learning,stricker2022experimental}. However, optimizing the implementation of IC-POVMs for shadow estimation remains a crucial yet unresolved challenge, especially considering the noises in quantum hardware.
Refs.~\cite{acharya2021shadow,nguyen2022optimizing,innocenti2023shadow} effectively introduce IC-POVMs into shadow estimation and provide performance analysis along with optimization strategies, but these works do not address the compilation of POVMs. Ref.~\cite{garcia2021learning} takes a step further by implementing POVMs in a (direct-tensor) dilation framework but relies on brute-force unitary decomposition for the compilation, which typically requires three CNOT gates for a single-qubit POVM. This presents a challenge for near-term devices, where CNOT gates are particularly susceptible to noise and thus exhibit higher error rates compared to single-qubit gates \cite{tannu2019not}. Consequently, it is imperative to lessen their usage. Moreover, the direct-sum dilation framework presented in Refs.~\cite{stricker2022experimental,fischer2022ancilla} offers an alternative approach by utilizing additional levels to implement qubit-POVM with qudit. However, this method could be limited to platforms that support qudit extensions, such as the ion-trap system \cite{stricker2022experimental}. Therefore, developing a general framework for efficiently compiling and optimizing POVMs on qubit-based quantum computers remains an open and significant problem. 

To address this challenge, we propose an efficient IC-POVM compilation strategy for qubit circuits and apply it to effectively implement and optimize shadow estimation, as shown in Fig.~\ref{fig:overall}. 
To analyze and reduce the CNOT count of the compilation, we extract the parameters of the unitary corresponding to some POVM as gates in the circuit form. This allows us to freely adjust the parameters that are irrelevant to the POVM. In particular, we demonstrate that any single-qubit minimal IC-POVM can be implemented with no more than 2 CNOT gates. For SIC-POVMs, we further reduce the CNOT count to 1 and provide both an explicit compilation circuit and an algorithm for determining the required parameters. Moreover, we apply the proposed 1-CNOT SIC-POVM circuits to shadow estimation tasks and demonstrate their noise-resilient performance and flexibility in compiling various SIC-POVMs. These findings show that our work can enhance the practicability of POVM-based shadow estimation, from noise reduction to compatibility with POVM-related algorithms \cite{garcia2021learning,nguyen2022optimizing}. 

\section{Preliminaries}\label{sec:preli}
In this section, we first introduce the basic concepts of POVMs, and then discuss their implementations in quantum circuits, and finally introduce the local equivalence class of 2-qubit gates, which helps determine the minimal CNOT count needed for gate compilation.

\subsection{Positive Operator-Valued Measurements}
Assume that $\left\{ \Pi _i \right\} _{i=1}^{n}$  defines a discrete POVM with $n$ elements on a $d$-dimensional Hilbert space $\mathcal{H}$, and $\mathcal{L} \left( \mathcal{H} \right)$ is the $d^2$-dimensional space of linear operators on $\mathcal{H}$, then this POVM satisfies the following conditions that
\begin{align}\label{eq:povm_conditions}
    \begin{aligned}
        \Pi _i\succeq 0, \ \sum\nolimits_{i=1}^n{\Pi _i=I},
    \end{aligned}
\end{align}
where $I$ is the identity matrix. The first part of these conditions ensures the positive semi-definiteness of the POVM elements, while the second part ensures completeness. Note that POVM elements do not necessarily span the entire $\mathcal{L} \left( \mathcal{H} \right)$, those that do are called IC-POVMs with
\begin{align}\label{eq:icpovm_conditions}
    \begin{aligned}
        \mathrm{rank}\left( \sum_{i=1}^n{|\Pi _i\rangle \!\rangle \langle \!\langle \Pi _i|} \right) =\mathrm{dim}\left( \mathcal{L} \left( \mathcal{H} \right) \right) .
    \end{aligned}
\end{align}
Here, the term $|A\rangle \!\rangle$ denotes the vectorization of the operator $A\in\mathcal{L} \left( \mathcal{H} \right)$. For two operators $A,B\in \mathcal{L} \left( \mathcal{H} \right)$, the Hilbert-Schmidt inner product is given by $\langle \!\langle A|B\rangle \!\rangle =\mathrm{Tr}\left( A^{\dagger}B \right)$. An IC-POVM that spans the entire $\mathcal{L} \left( \mathcal{H} \right)$ using the minimum number of elements is called a minimal IC-POVM \cite{zhu2011quantum}, where $n = \mathrm{dim}\left( \mathcal{L} \left( \mathcal{H} \right) \right) = d^2$. Among minimal IC-POVMs, SIC-POVM is an important class characterized by the high symmetry among its elements \cite{renes2004symmetric,zhu2010sic,huangjun2012quantum,fuchs2017sic,appleby2018constructing}. Specifically, the elements of a SIC-POVM satisfy $\mathrm{rank}\left( \Pi _i \right) =1$, and
\begin{align}\label{eq:sicpovm_conditions}
    \begin{aligned}
        \langle \!\langle \Pi _i|\Pi _j\rangle \!\rangle=\frac{d\delta _{ij}+1}{d^2\left( d+1 \right)}.
    \end{aligned}
\end{align}
Two typical single-qubit SIC-POVMs are illustrated in Fig.~\ref{fig:2sets}.

\begin{figure}[htbp]
\centering
\includegraphics[width=0.48\textwidth]{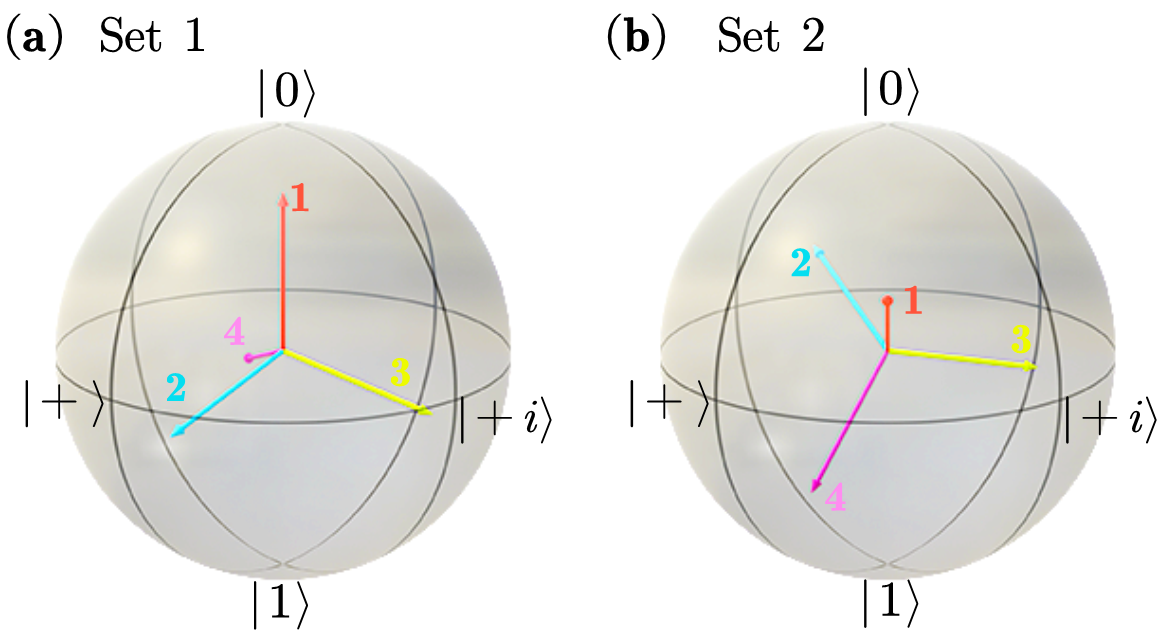}
\caption{\justifying{Two commonly used SIC-POVM sets. (a) Set 1: $\{ \Pi _i=|\varphi _i\rangle \langle \varphi _i|\}_{i=1}^{4}$, where $|\varphi _1\rangle = \frac{1}{\sqrt{2}}|0\rangle $ and $|\varphi _j\rangle =\frac{1}{\sqrt{6}}\left( |0\rangle +\sqrt{2}e^{\mathrm{i}\frac{2\pi}{3}\left( j-2 \right)}\ket{1} \right), \ j\in \{2,3,4\}$ are subnormalized state vectors. This set is widely used like in Refs.~\cite{stricker2022experimental, garcia2021learning, garcia2023experimentally}.
(b) Set 2: The Weyl–Heisenberg group covariant SIC-POVM \cite{renes2004symmetric,zhu2010sic,fuchs2017sic}, which is generated from a fiducial state represented by the red vector in the Bloch sphere. This state can be obtained by applying the unitary $R_z(\frac{3\pi}{4})R_x(\arccos(\frac{1}{3}))$ to $ \frac{1}{\sqrt{2}}\ket{0}$. Additionally, permuting the SIC-POVM elements (e.g., exchanging vectors 3 and 4) essentially results in the same SIC-POVM, merely with different labels.}}
\label{fig:2sets}
\end{figure}

\subsection{Circuit Model to Realize POVMs}\label{sec:circ_povm}
The implementations of POVMs \cite{yordanov2019implementation,singh2022implementation} primarily rely on Neumark's theorem \cite{peres1990neumark}. This theorem states that POVMs for a quantum system $\mathcal{H}_S$ can be realized by embedding $\mathcal{H}_S$ into a higher-dimensional system via an auxiliary system $\mathcal{H}_A$, applying a unitary transformation $U_{\mathrm{P}}$, and then performing PVMs on the combined system or solely on $\mathcal{H}_A$ \cite{chen2007ancilla}.
Performing joint PVMs on the combined system requires fewer auxiliary dimensions than performing PVMs only on $\mathcal{H}_A$. This strategy is depicted in the orange part of Fig.~\ref{fig:overall}, that is 
\begin{align}\label{eq:povm_process}
    &U_{\mathrm{P}}\left( \mathcal{H}_A\otimes \mathcal{H}_S \right)\ + \ \mathrm{PVM} \ \mathrm{on} \ \mathcal{H}_A\otimes \mathcal{H}_S \nonumber\\ &\rightarrow \left\{ \Pi _i \right\}_{i=1}^n \,\mathrm{on} \ \mathcal{H}_S.
\end{align}

When applying shadow estimation, preserving the post-collapse quantum states is unnecessary. Therefore, adopting the measurement strategy outlined in Eq.~\eqref{eq:povm_process} is preferable. The key point of Eq.~\eqref{eq:povm_process} is selecting the appropriate auxiliary state and $U_{\mathrm{P}}$. We first set the initial state of $\mathcal{H}_{A}$ to $|0\rangle$, then construct $U_{\mathrm{P}}$ based on the relationship between $U_{\mathrm{P}}$ and rank-1 POVMs (App.~\ref{ap:POVM_Unitary}) using a straightforward process \cite{tabia2012experimental}. For a rank-1 POVM $\left\{ \Pi _{i}=|\varphi _{i}\rangle \langle \varphi _{i}| \right\} _{i=1}^{n}$ with $n$ (subnormalized) elements in a $d$-dimensional system, an $n \times d$ matrix can be constructed as
\begin{align}\label{eq:V}
    \begin{aligned}
      V=\left[ \begin{matrix} |\varphi _1\rangle& \cdots& |\varphi _n\rangle\\ \end{matrix} \right] ^{\dagger} .
    \end{aligned}
\end{align}
Then, we can find an $n \times \left( n-d \right)$ matrix $W$ such that
\begin{align}\label{eq:povm_elements_matrix}
    \begin{aligned}
        U_{\mathrm{P}}=\left[ \begin{matrix}
	V&		W\\
        \end{matrix} \right], 
    \end{aligned}
\end{align}
where $W$ can generally be generated using the Gram-Schmidt orthogonalization to ensure the unitarity of the entire matrix $U_{\mathrm{P}}$.

For $U_{\mathrm{P}}$, various libraries \cite{gadi_aleksandrowicz_2019_2562111,iten2019introduction,krol2022efficient} and methods \cite{shende2005synthesis,mottonen12006decompositions,rakyta2022approaching} are available to decompose it into a quantum circuit using a basic gate set that includes CNOT gate, where the CNOT count is our optimization target. The minimum number of CNOT gates required to implement an arbitrary unitary matrix on an $N$-qubit system is given by $\lceil \left( 4^N-3N-1 \right)/4 \rceil$ \cite{shende2004minimal}. Specifically, for a single-qubit SIC-POVM, corresponding to some $U_{\mathrm{P}}$ with $N=2$, three CNOT gates are typically needed. However, we find that the CNOT count can be optimized by implementing operations that do not affect the measurement.

\subsection{Local Equivalence Class of 2-qubit Gates}\label{sec:lec}
In this section, we introduce the local equivalence class and its determination method for 2-qubit gates, which would be useful for analyzing and reducing the CNOT count of $U_{\mathrm{P}}$, as applied in Sec.~\ref{sec:reduce}. 

Two operators $U_1,U_2\in U(4)$ are considered to be in the same local equivalence class if they can be compiled with the same minimum number of CNOT gates, denoted as $U_1\sim U_2$. There are generally two methods to determine these classes.

The first method involves obtaining the canonical class vector by performing the Cartan decomposition \cite{tucci2005introduction}. After ignoring the global phase, this decomposition provides a canonical form of $U$ for any $U \in U(4)$, that is
\begin{equation}\label{eq:kak_decom}        
\includegraphics{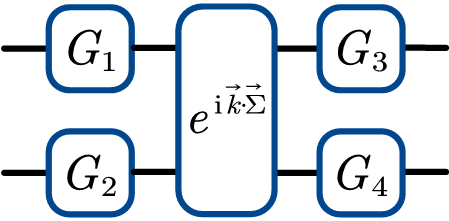},
\end{equation}
where $G_i$ is some single-qubit gate, $\vec{k}=\left[k_1\,\,k_2\,\,k_3 \right] $, and $\vec{\Sigma}=\left[\sigma _{1}^{\otimes 2}\,\,\sigma _{2}^{\otimes 2}\,\,\sigma _{3}^{\otimes 2} \right]$ with $\sigma_i$ denoting the Pauli operators $X$, $Y$, $Z$ respectively. If $\vec{k}^{(1)} = \vec{k}^{(2)}$, then $U_1\sim U_2$. The set of all canonical class vectors forms a so-called Weyl chamber, with the regime $\frac{\pi}{4}\geqslant k_1\geqslant k_2\geqslant |k_3|\geqslant 0$. Specifically, when $k_1=\frac{\pi}{4}$, one has $k_3\geqslant 0$. In the Weyl chamber, each point uniquely corresponds to a local equivalence class.

The second approach involves computing an invariant of $U$ \cite{shende2004recognizing}:
\begin{align}\label{eq:char_poly}
\begin{aligned}
\chi \left[ \gamma \left( U \right) \right] =  \det(xI-\gamma \left( U \right)),
\end{aligned}
\end{align}
where $\gamma\left( U \right)=U\sigma_{2}^{\otimes 2}U^{T}\sigma_{2}^{\otimes 2}$. For $U_1,U_2 \in SU(4)$, where $SU(4)$ represents unitary matrices with a determinant of 1 in $4$-dimensional space, if $\chi \left[ \gamma \left( U_1 \right) \right] =\chi \left[ \gamma \left( U_2 \right) \right]$, then $U_1\sim U_2$. Since Eq.~\eqref{eq:char_poly} requires $U \in SU(4)$, we can adjust the global phase to ensure all gates have a determinant of 1 when analyzing operators derived from circuit representations, such as those in Sec.~\ref{sec:reduce}. For instance, $e^{\mathrm{i}\frac{\pi}{4}}\mathrm{CNOT}$ can replace the original CNOT.

Using the above two methods, we are prepared to determine whether a specific 2-qubit gate can be compiled with fewer CNOT gates. For $U \in U(4)$, it can be decomposed with 2 CNOT gates if and only if $k_3=0$ \cite{vidal2004universal}, implying that the canonical class vector must lie on the $k_3=0$ plane of the Weyl chamber. For the case where $U$ can be compiled using a single CNOT gate, the canonical class vector should be located at $\vec{k}=\left[ \frac{\pi}{4} \ 0 \ 0 \right]$, corresponding to the CNOT gate itself, and must satisfy
\begin{align}\label{eq:chi_simpl}
\begin{aligned}
\begin{cases}
	\mathrm{tr}\left( \gamma \left( U\right) \right) =0,\\
	\mathrm{tr}\left( \gamma ^2\left( U\right) \right) =-4\det(U).\\
\end{cases}
\end{aligned}
\end{align}
The details about Eq.~\eqref{eq:chi_simpl} can be found in App.~\ref{ap:determinant}.

\section{Circuit compilation and optimization for IC-POVM}\label{sec:reduce}
In this section, we mainly focus on the compilation of the unitary $U_{\mathrm{P}}$ of arbitrary qubit-POVM described in Sec.~\ref{sec:circ_povm} with $n=4$ and $d=2$. We analyze the parameters of $U_{\mathrm{P}}$ and relate their effect to a circuit form. We find that the minimum CNOT count of single-qubit 4-element rank-1 POVMs (including minimal IC-POVMs) and SIC-POVMs are at most 2 and 1, respectively. Finally, we provide a more structured circuit model specifically for SIC-POVMs.

\subsection{Extract parameters of $U_{\mathrm{P}}$ as circuit elements}\label{sec:extract}
For a specific unitary, the minimum required number of CNOT gates is fixed. However, for a POVM, there is a set of $U_{\mathrm{P}}$ that can implement the same POVM. Such degrees of freedom in $U_{\mathrm{P}}$ that are not related to the POVM allow for a possible reduction in the CNOT count. Hereafter we first count and classify the parameters of $U_{\mathrm{P}}$.

\begin{figure}[h]
\centering
\includegraphics[width=0.48\textwidth]{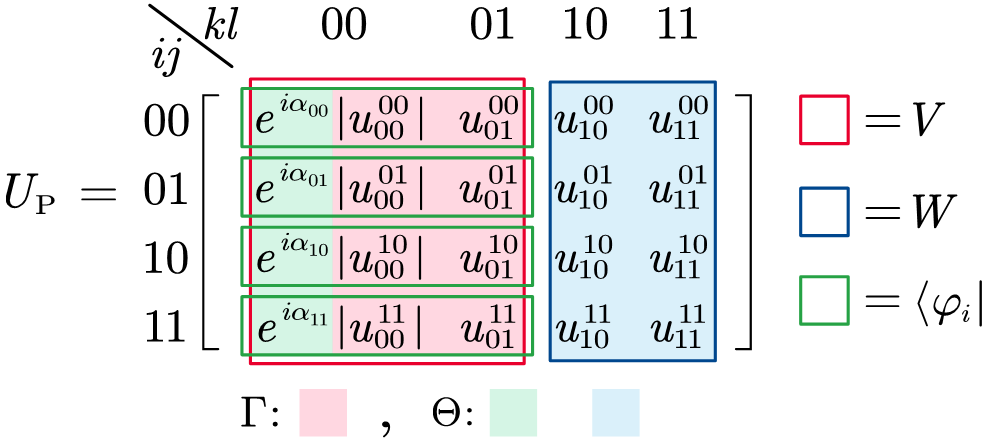}
\caption{\justifying{A detailed description of different parts of the $U_{\mathrm{P}}$ matrix. The elements of $U_{\mathrm{P}}$ are denoted as $u_{kl}^{ij}$. To emphasize the phase of each element in the first column, we write $u_{00}^{ij}=e^{\mathrm{i}\alpha _{ij}}|u_{00}^{ij}|$. The $V$ and $W$ matrices from Eq.~\eqref{eq:povm_elements_matrix} are enclosed within the red and blue line boxes, respectively. The $\bra{\varphi_i}$, a row element of $V$, is highlighted by the green line box. The matrix elements corresponding to the POVM are shaded in pink, which relates to $\Gamma$. And the non-corresponding elements are shaded in green and blue, which relate to $\Theta$.}}
\label{fig:Up_detail}
\end{figure}

The unitary $U_{\mathrm{P}}$ is some 2-qubit gate, which is determined by 16 parameters, as shown in Fig.~\ref{fig:Up_detail}. As mentioned in Sec.~\ref{sec:preli}, the density matrix of $|\varphi_i\rangle$ (defined in Eq.~\eqref{eq:V}) is essential for POVMs, corresponding to the shaded pink area in Fig.~\ref{fig:Up_detail} and determined by 8 parameters, denoted as $\Gamma$. In contrast, the $W$ matrix and the global phase of each $|\varphi_i\rangle$ in Eq.~\eqref{eq:povm_elements_matrix}, which do not affect the POVM measurement results, correspond to the shaded blue and green areas. These matrix elements are also determined by 8 parameters, denoted as $\Theta$.
Consequently, we categorize all 16 parameters of $U_{\mathrm{P}}$ into two classes based on their relevance with the POVM, that is $\Gamma$ and $\Theta$, denoted as $U_{\mathrm{P}}(\Gamma,\Theta)$. In other words, once the POVM is given, $\Gamma$ is determined, while $\Theta$ can be adjusted freely. We leave the detailed description of $\Gamma$ and $\Theta$ in App.~\ref{ap:parameterize}. 
In the following, we focus on their representations in the circuit form. 

Here, we present a circuit for constructing a general $U_{\mathrm{P}}(\Gamma,\Theta)$ as described in Eq.~\eqref{eq:povm_syn}. This circuit can adjust the free parameters $\Theta$ by adding additional gates around a specific one $U_{\mathrm{P}}(\Gamma,\Theta_0)$, where $U_{\mathrm{P}}(\Gamma,\Theta_0)$ serves as the starting point for $\Theta$ variation and is generally chosen to a fixed reference $U_{\mathrm{P}}$.
\begin{equation}\label{eq:povm_syn} 
\includegraphics{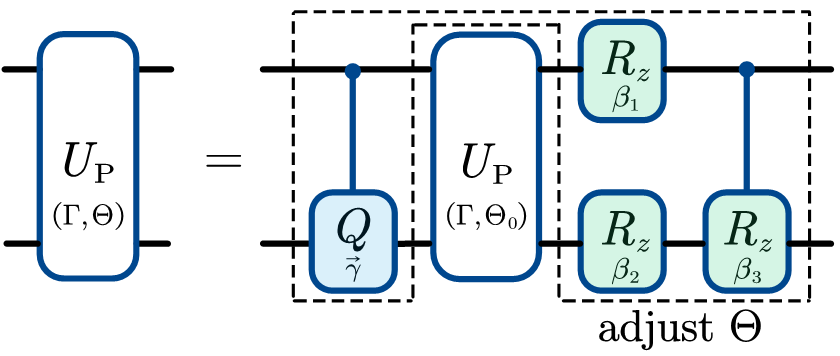}.
\end{equation}

As mentioned previously, $\Theta$ consists of 8 parameters. In this context, the controlled unitary gate $C_Q$ (highlighted in blue) can change $W$ which consists of 4 parameters. The $R_z(\beta)$ and $C_{R_z}(\beta)$ gates (highlighted in green) control the phases of each $\ket{\varphi_i}$, where $R_z(\beta)$ is defined as $R_z(\beta)=\mathrm{diag}\left( \left[ \begin{matrix}
	1&		e^{\mathrm{i}\beta}\\
\end{matrix} \right] \right)$, with $\beta$ as a phase parameter. Additionally, a global phase is not depicted in the diagram, so there are only 3 of $\beta$ in these phase gates. As a result, we can then express the change of $\Theta$ from $\Theta_0$ as 
\begin{equation}\label{eq:Theta}
\varDelta_{\Theta_0 \rightarrow \Theta} = \left[ \begin{matrix}
    \gamma _0 &		\gamma _1 &		\gamma _2 &		\gamma _3 &		\beta _1 &		\beta _2 &		\beta _3\\
\end{matrix} \right],
\end{equation}
where $\Theta_0$ is a fixed parameter set in $U_{\mathrm{P}}(\Gamma,\Theta_0)$, $\gamma_j$ represents the parameters in $Q=e^{\mathrm{i}\gamma_0}R_x(\gamma_1)R_z(\gamma_2)R_x(\gamma_3)$, and $\beta_i$ corresponds to those in $R_z$ and $C_{R_z}$.

Furthermore, let us consider an important class of IC-POVMs, say SIC-POVMs. Due to its high symmetry, we can further simplify the previous circuit and construct a more in-depth parameterized circuit model. For SIC-POVMs, we denote the corresponding unitary $U_{\mathrm{P}}$ as $U_{\mathrm{SIC}}$. Given the constraints of Eq.~\eqref{eq:sicpovm_conditions}, the number of parameters of $\Gamma$ is reduced from 8 to 3 continuous parameters and 1 discrete one. The 3 continuous parameters correspond to a single-qubit gate $U_S$, while the discrete parameter $c$ is retained in $U_{\mathrm{SIC}}(c,\Theta_0
)$, representing a reference $U_{\mathrm{SIC}}$ with a discrete parameter $c \in \{0,1\}$, as detailed in App.~\ref{ap:para_psi}. By integrating these parameters with Eq.~\eqref{eq:povm_syn}, the decomposition of $U_{\mathrm{SIC}}$ for a SIC-POVM can be further expressed as
\begin{equation}\label{eq:sicpovm_decomp} 
\includegraphics{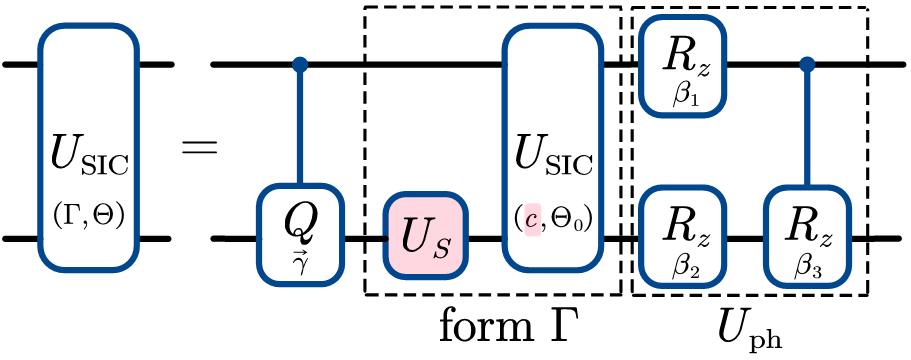}.
\end{equation}
The parameters in $\Gamma$ are highlighted in pink, where the parameter $c$ determines the row order of $U_{\mathrm{SIC}}(c,\Theta_0)$. The element labels in Fig.~\ref{fig:2sets} correspond directly to row indices of $U_{\mathrm{SIC}}(c,\Theta_0)$: `1234' for $c=0$ and `1243' for $c=1$. For simplicity, we collectively refer to the phase gates $R_z$ and $C_{R_z}$ at the circuit's end as $U_{\mathrm{ph}}$, as shown in Eq.~\eqref{eq:sicpovm_decomp}. In the following section, we take Set 1 introduced in Fig.~\ref{fig:2sets} as the reference SIC-POVM, with its corresponding unitary denoted as $U_{\mathrm{SIC-1}}(c,\Theta_1)$ as defined in Eq.~\eqref{eq:uk}.

\subsection{Reduce CNOT Count of Minimal IC-POVMs}\label{sec:reduceic}
In Sec.~\ref{sec:extract}, we classify the parameters of $U_{\mathrm{P}}$ and connect their effects to the circuit form. This framework enables us to reduce the CNOT count by adjusting the circuit gates rather than directly modifying the matrix elements of $U_{\mathrm{P}}$. To facilitate these adjustments, we outline some fundamental circuit identities and results. The first one is the decomposition of single-qubit unitaries \cite{shende2004minimal}, where any $U \in SU(2)$ can be decomposed as
\begin{equation}\label{eq:single_rrr} 
\includegraphics{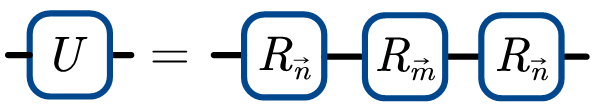},
\end{equation}
with $\vec{n},\vec{m} \in \mathbb{R}^{3}$, and $\vec{n} \bot \vec{m}$. In different tasks, we will choose different forms for convenience. And for controlled-unitary gate, a practical decomposition is given by
\begin{equation}\label{eq:controlled_u_decomp} 
\includegraphics{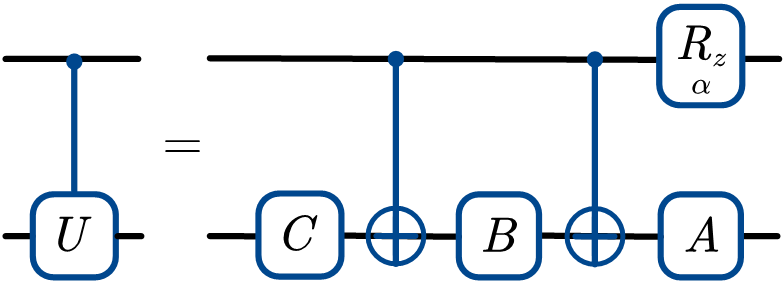},
\end{equation}
where $U=e^{\mathrm{i}\alpha}AXBXC$, $ABC=I$ \cite{barenco1995elementary}. Note that we also ignore the global phase here. Next is the basic circuit identity for controlled unitary gates, for which we have
\begin{equation}\label{eq:control_u_cross} 
\includegraphics{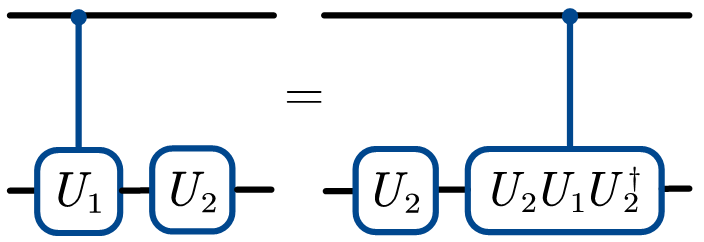}.
\end{equation}
This identity can be derived using block matrix multiplication. In particular, the $R_z$ gate on the control qubit commutes with the CNOT gate, and the $R_x$ gate commutes with the CNOT gate when applied to the controlled qubit. Lastly, we recall 
Proposition \uppercase\expandafter{\romannumeral5}.2 in Ref.~\cite{shende2004minimal}, which we restate here for convenience, that is, 
for any $U \in SU(4)$, one can find $\xi _1,\xi _2,\xi _3$, such that
\begin{align}\label{eq:equivl}
\begin{aligned}
&\chi \left[ \gamma \left( UC_{X}\left( I \otimes R_z\left( \xi _1 \right) \right) C_{X} \right) \right] \\ 
=&\chi \left[ \gamma \left( C_{X}\left( R_x\left( \xi _2 \right) \otimes R_z\left( \xi _3 \right) \right) C_{X} \right) \right].
\end{aligned}
\end{align}
Here, $C_X$ refers to the CNOT gate, and $\chi \left[ \gamma \left( U \right) \right]$ (defined in Eq.~\eqref{eq:char_poly}) can be used to determine if two unitary gates can be compiled under the same CNOT count. By observing the structure of the left-hand side of Eq.~\eqref{eq:equivl}, we can decompose $C_Q$ to match the required form. Since the parameters embedded in $C_Q$ can be freely adjusted, we can always identify a $U_{\mathrm{P}}$ that requires only 2 CNOT gates by applying Eq.~\eqref{eq:equivl}, which is shown as follows.

\begin{prop}\label{prop:povm_2cnot}
    For any $U_{\mathrm{P}}$ of single-qubit 4-element rank-1 POVMs, one can always
    find a suitable $\Theta$ by optimizing the circuit using Eq.~\eqref{eq:povm_syn}, such that the final circuit can be compiled with only 2 CNOT gates.
\end{prop}
\begin{proof}
Firstly, we can always decompose the $C_{Q}$ in Eq.~\eqref{eq:povm_syn} using Eq.~\eqref{eq:controlled_u_decomp}. By observing the construction of Eq.~\eqref{eq:equivl}, we can further decompose $B$ using Eq.~\eqref{eq:single_rrr}, where we choose $R_x$ and $R_z$ to replace $R_{\vec{n}}$ and $R_{\vec{m}}$, respectively. Then, leveraging the circuit identities mentioned in Eq.~\eqref{eq:control_u_cross}, we can move $R_x$ gates across the CNOT gate and merge them with $A$ and $C$. This process ultimately yields
\begin{equation}\label{eq:povm_prof} 
\includegraphics{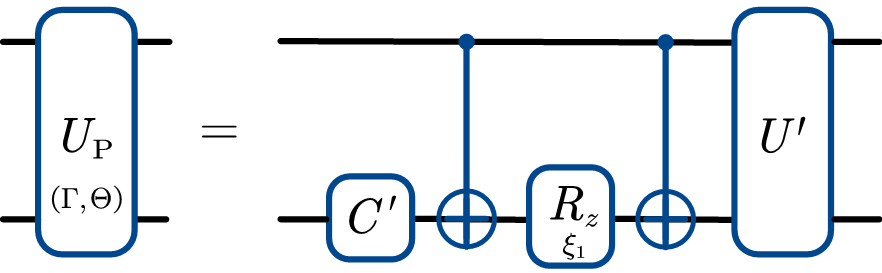},
\end{equation}
where
\begin{equation}
\includegraphics{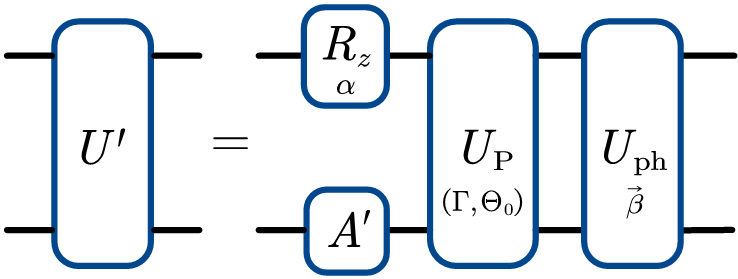}.
\end{equation}
The single-qubit gate $C^{\prime}$ does not affect local equivalence. Given the freedom to choose $\Theta$, a $\xi _1$ can always be found to satisfy Eq.~\eqref{eq:equivl}. Thus, with this chosen $\Theta$, denoted as $\Theta^{*}$, a new $U_{\mathrm{P}}$ can be obtained, which can be implemented using 2 CNOT gates, that is
\begin{equation}
\includegraphics{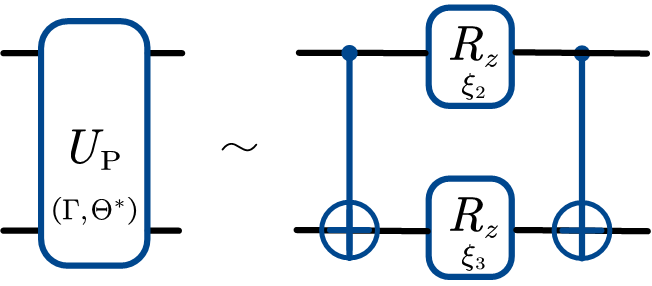}.
\end{equation}
\end{proof}

The role of $C_{Q}$ is crucial in this process, as we adjust the parameters $\gamma_i$ embedded in $C_Q$ to apply the Eq.~\eqref{eq:equivl}. In Fig.~\ref{fig:povm_cors} (a) and (b), we numerically find the solution of 2 CNOT gates by adjusting $C_Q$. To illustrate this process, we perform the Cartan decomposition to obtain the canonical class vectors (the first approach in Sec.~\ref{sec:lec}) and observe their trajectory as $\gamma_i$ changes.
We select $U_{\mathrm{SIC-1}}(c=0,\Theta_1)$ (with $c=0$ in Eq.~\eqref{eq:uk} of App.~\ref{ap:parameterize}) as the initial $U_{\mathrm{P}}$, denoted as $U_{\mathrm{P}}(\Gamma_1,\Theta_1)$, for our numerical experiments. We set $\beta_1,\beta_2,\beta_3=0$ in $U_{\mathrm{ph}}$ and utilize only the degrees of freedom in $Q$ for optimization.

In Fig.~\ref{fig:povm_cors}(a), by changing $C_Q$, the possible canonical class vectors form a cuboid within the Weyl chamber, intersecting the $k_3 = 0$ plane. It implies that while an infinite number of $U_{\mathrm{P}}$ can implement the same POVM, the probability of randomly generating one that requires only 2 CNOT gates is almost zero. In Fig.~\ref{fig:povm_cors}(b), we only adjust the parameter $\gamma_2$ in $C_Q$. The red point in the graph represents the initial $U_{\mathrm{P}}(\Gamma_1,\Theta_1)$ that we constructed, which requires 3 CNOT gates for compilation. By adjusting $\gamma_2$, we generate a continuous trajectory that intersects with $k_3=0$. By employing the dichotomy method, we can locate a $U_{\mathrm{P}}(\Gamma_1,\Theta^{*})$ that can be compiled with 2 CNOT gates.

\subsection{Further Reduce CNOT Count of SIC-POVM}\label{sec:sic_reduce}
For SIC-POVM, owing to the more refined structure presented in Eq.~\eqref{eq:sicpovm_decomp}, it can be demonstrated that there is always a $U_{\mathrm{SIC}}$ which can invariably be constructed using 1 CNOT gate by constructing Eq.~\eqref{eq:povm_syn} and adjusting $\Theta$.

\begin{figure*}[!ht]
\centering
\includegraphics[width=0.8\textwidth]{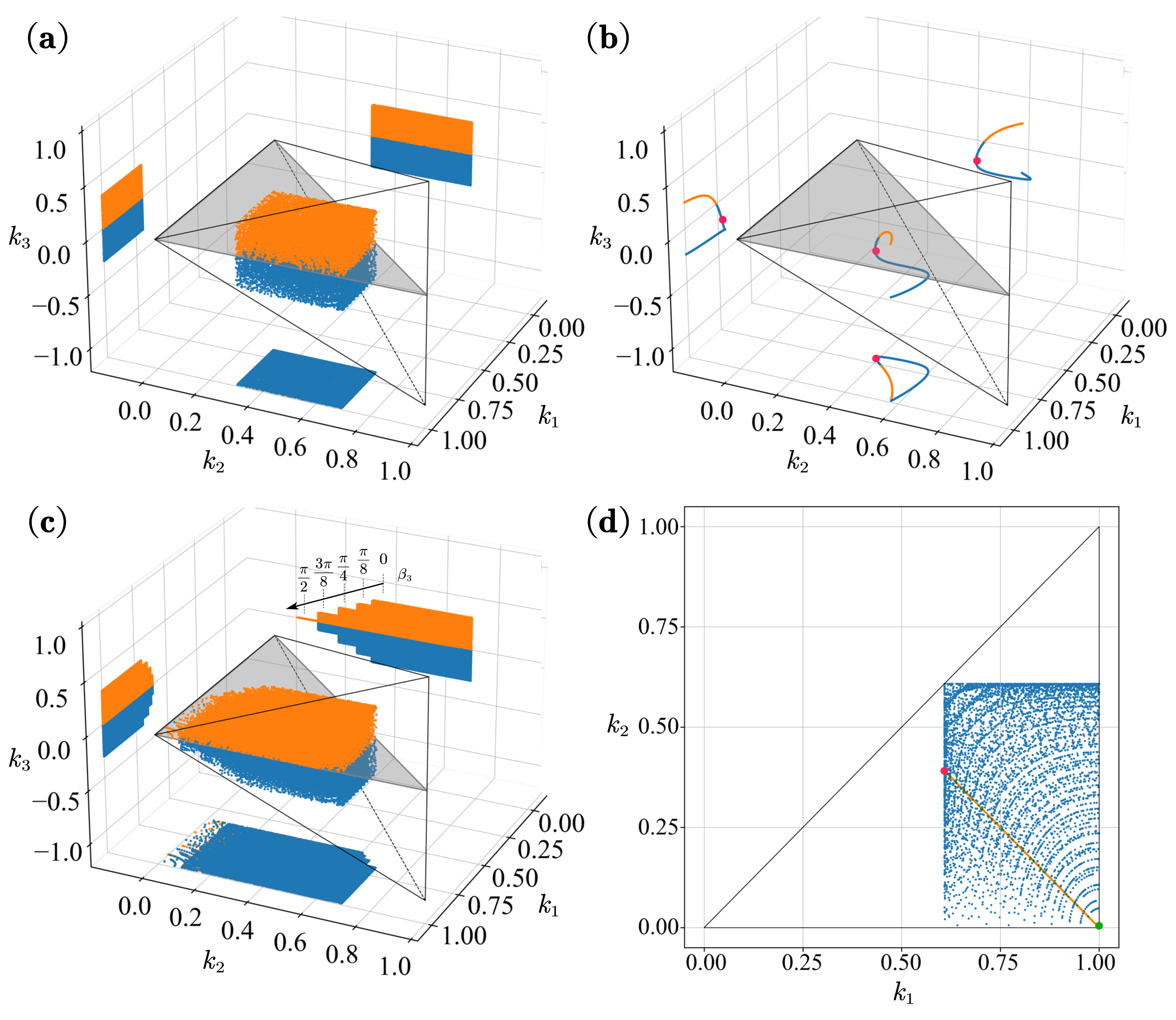}
\caption{\justifying{Searching for the solution of $U_{\mathrm{P}}(\Gamma,\Theta^{*})$ and $U_{\mathrm{SIC}}(c,\Theta_m^*)$ by varying $\Theta$ in Eq.~\eqref{eq:povm_syn}. The initial $U_{\mathrm{P}}$ and the reference $U_{\mathrm{SIC}}(c,\Theta_0)$ are both selected as $U_{\mathrm{SIC-1}}(c=0,\Theta_1)$ (with $c=0$ in Eq.~\eqref{eq:uk} of Appendix).
In (a)-(c), orange points represent the canonical vector with $k_3 \geqslant 0$, while blue points represent those with $k_3 < 0$. The values are scaled and normalized by $\frac{\pi}{4}$. (a) By changing the parameters in $C_Q$, the possible canonical class vectors form a cuboid within the Weyl chamber, intersecting the plane $k_3 = 0$. (b) A possible trajectory when varying only $\gamma_2$ in $C_Q$, with the red point indicating the initial $U_{\mathrm{P}}$. (c) The canonical vector trajectory formed by varying $C_Q$ under different values of $\beta_3$. As $\beta_3$ approaches $\frac{\pi}{2}$, the trajectory cuboid aligns more closely with the $k_3=0$ plane. (d) When $\beta_3=\frac{\pi}{2}$, the canonical class vectors lie entirely on the $k_3=0$ plane. The orange line shows a trajectory while adjusting $\gamma_1$ in $C_Q$. And the green point represents the canonical class vector corresponding to the CNOT gate.}}
\label{fig:povm_cors}
\end{figure*}

\begin{prop}\label{prop:sicpovm_1cnot}
    For any $U_{\mathrm{SIC}}$ of single-qubit SIC-POVMs, one can always
    find a suitable $\Theta$ by optimizing the circuit using Eq.~\eqref{eq:povm_syn}, such that the final circuit can be compiled with only 1 CNOT gate.
\end{prop}
\begin{proof}
A given $U_{\mathrm{SIC}}(\Gamma,\Theta)$ can always be decomposed in the form of Eq.~\eqref{eq:sicpovm_decomp}. The essential parameter $\Gamma$ of the SIC-POVM depends on $U_S$ and $U_{\mathrm{SIC}}(c,\Theta_0)$, where $U_S$ represents the single-qubit gate. Using Eq.~\eqref{eq:povm_syn}, one can modify the free parameter $\Theta$ by introducing additional  $C_{Q_2}$ and $U_{\mathrm{ph}2}$, as shown in the first line of Eq.~\eqref{eq:sic_prof}. 
\begin{equation}\label{eq:sic_prof} 
\includegraphics{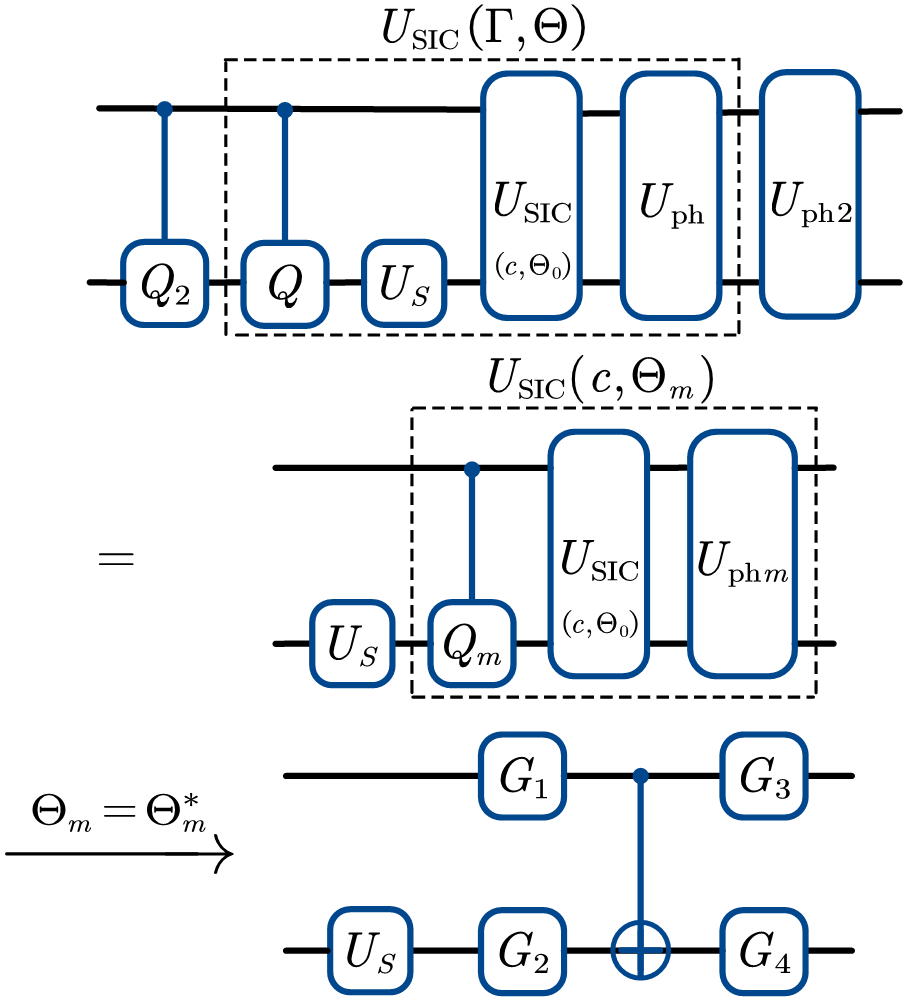}.
\end{equation}

Since $U_S$ can always be exchanged with the controlled-unitary operations as explained in Eq.~\eqref{eq:control_u_cross}, one can move it to the beginning of the circuit. And we further merge the controlled operations as $C_{Q_m}$, as shown in the second line of Eq.~\eqref{eq:sic_prof}. After $U_S$ is moved outside, we ignore it and denote all the parameters in the merged circuit as $\Theta_m$, referring to the merged circuit as $U_{\mathrm{SIC}}(c,\Theta_m)$.

Finally, we choose $U_{\mathrm{SIC-1}}(c,\Theta_1)$ in Eq.~\eqref{eq:uk} as the $U_{\mathrm{SIC}}(c,\Theta_0)$. By solving Eq.~\eqref{eq:chi_simpl} while adjusting $\Theta_m$, we find that there exists a solution of 1 CNOT for any $c$, that is
\begin{align}\label{eq:thetamstar}
\begin{aligned}
\varDelta_{\Theta_1 \rightarrow \Theta_m^*}(c)=\left[ \begin{matrix}
	0&		\pi&		0&		0&		0&		0&		\left( -1 \right) ^c\frac{\pi}{2}\\
\end{matrix} \right],
\end{aligned}
\end{align}
with the meaning of each parameter defined in Eq.~\eqref{eq:Theta}. Then, the final unitary is $U_{\mathrm{SIC-1}}(c,\Theta_m^*)$. After obtaining $\Theta_m^*$, the merged circuit can be decomposed as in the third line of Eq.~\eqref{eq:sic_prof}, where $G_{i}$ is a determined single-qubit gate. This completes the proof.
\end{proof}

In Fig.~\ref{fig:povm_cors} (c) and (d),
we again use $U_{\mathrm{SIC-1}}(c=0,\Theta_1)$ to demonstrate the 1 CNOT solution for the SIC-POVM. This is done by tracing the trajectory of the canonical class vectors as $\Theta$ varies, using the dichotomy method, which is similar to the approach in Fig.~\ref{fig:povm_cors} (a) and (b). We vary the $\beta_3$ values (noting that $\beta_1$ and $\beta_2$ do not affect the canonical class vector) and proceed to adjust the parameter $\gamma_i$. As illustrated in Fig.~\ref{fig:povm_cors}(c), when $\beta_3$ approaches $\frac{\pi}{2}$, the overall trajectory increasingly aligns with the $k_3=0$ plane. When $\beta_3=\frac{\pi}{2}$, the canonical class vectors lie entirely on the plane $k_3=0$. Consequently, regardless of variations in $\gamma_i$, compiling the corresponding $U_{\mathrm{SIC}}$ requires no more than 2 CNOT gates. In Fig.~\ref{fig:povm_cors}(d), we set $\beta_3=\frac{\pi}{2}$ to observe this special case. By adjusting $\gamma_1$ in $C_Q$, we obtained the orange trajectory, which contains the green point $\vec{k}=\left[ \begin{matrix}
	\frac{\pi}{4}&		0&		0\\
\end{matrix} \right] $, representing the canonical class vector of the CNOT gate.

\begin{figure*}[ht]
\centering
\includegraphics[width=\textwidth]{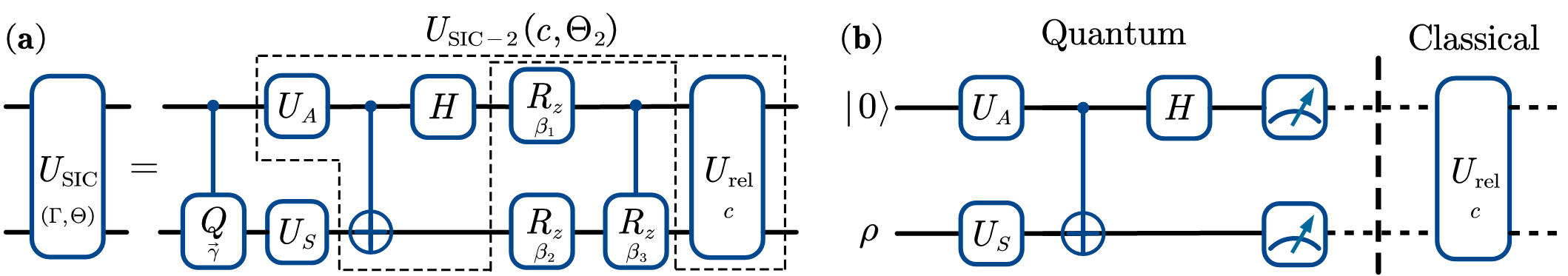}
\caption{\justifying{General and practical circuits for $U_{\mathrm{SIC}}$. (a) The general circuit is derived using $U_{\mathrm{SIC-2}}(c,\Theta_2)$ as the reference $U_{\mathrm{SIC}}$, and explicitly compiling it with the Bell measurement circuit and the relabeling operator $U_{\mathrm{rel}}$. The relabeling operator $U_{\mathrm{rel}}$ corresponds directly to the parameter $c$. This general circuit can compile any valid $U_{\mathrm{SIC}}$ by appropriately choosing $\Gamma$ and $\Theta$, with its unitary representation being identical to the original $U_{\mathrm{SIC}}$. (b) The practical circuit is obtained by implementing $U_{\mathrm{rel}}$ in the classical post-processing stage and employing Eq.~\eqref{eq:sic_prof} to eliminate the $U_{\mathrm{ph}}$ and $C_Q$. This circuit efficiently implements a SIC-POVM with just 1 CNOT. The components that vary for different SIC-POVMs are $U_S$ and $c$, which can be determined using Algo.~\ref{algo:algo1}.}}
\label{fig:prac_circ}
\end{figure*}

\subsection{Practical Circuit for SIC-POVMs}\label{sec:prac}
In the previous Sec.~\ref{sec:sic_reduce}, we obtained $\Theta_m^*$ in Eq.~\eqref{eq:thetamstar} and the corresponding 1-CNOT solution using the numerical dichotomy method. However, further steps are required to fully compile an arbitrary SIC-POVM. Specifically, we need to explicitly compile $U_{\mathrm{SIC}}(c,\Theta_m^*)$ in the third line of Eq.~\eqref{eq:sic_prof} and determine $U_S$ of $\Gamma$.

While $U_{\mathrm{SIC}}(c,\Theta_m^*)$ can be directly compiled into a 1-CNOT circuit using libraries like Qiskit \cite{gadi_aleksandrowicz_2019_2562111}, there is a more structured strategy for the explicit compilation using Bell measurements. In the literature, Ref.~\cite{jiang2020optimal,galvis2023single} implements SIC-POVM Set 2 in Fig.~\ref{fig:2sets} using a Bell measurement-based circuit. We denote the unitary of this
circuit as $U_{\mathrm{SIC-2}}(\Gamma_2,\Theta_2)$, where $\Theta_2$ is already a 1-CNOT solution. Hereafter, we show the connection of that solution to our method and then give a general construction of arbitrary SIC-POVM.

First, by extracting a suitable $U_S$ from $U_{\mathrm{SIC-2}}(\Gamma_2,\Theta_2)$, the remaining part becomes $U_{\mathrm{SIC-1}}(c=1,\Theta_2)$. Then, we find that the difference between $\Theta_2$ and $\Theta_m^*$ shows
\begin{align}\label{eq:solutiondiff}
\begin{aligned}
\varDelta_{\Theta_2 \rightarrow \Theta_m^*}(c=1)=\left[ \begin{matrix}
	\frac{2}{3}\pi&		0&		0&		0&		0&		0&		0\\
\end{matrix} \right],
\end{aligned}
\end{align}
where the only difference is the first parameter, $\gamma_0$, representing the phase of $Q$. Indeed, this phase can be realized by a single-qubit $R_{z}(\frac{2}{3}\pi)$ gate on the auxiliary qubit and can be ignored as it acts on the initial $\ket{0}$. As a result, the Bell measurement-based circuit is a special case within our framework. Given the fixed structure of the Bell measurement circuit, which simplifies explicit compilation, we choose $U_{\mathrm{SIC-2}}(\Gamma_2,\Theta_2)$ (also denoted as $U_{\mathrm{SIC-2}}(c=1,\Theta_2)$) as the reference $U_{\mathrm{SIC}}$, instead of $U_{\mathrm{SIC-1}}(c,\Theta_m^*)$ we used in Sec.~\ref{sec:sic_reduce}.

The Bell measurement-based circuit is directly applicable only when $c=1$, so further adjustment is needed to compile $U_{\mathrm{SIC}}$ when $c=0$. To maintain a fixed circuit structure while varying $c$, we introduce an additional operator $U_{\mathrm{rel}}$, which is a CNOT gate for $c=0$ or an identity gate for $c=1$. Details of the relabeling operators are provided in App.~\ref{ap:relabel_op}. This modification leads to a general fixed-structure circuit applicable to any $U_{\mathrm{SIC}}$ as shown in Fig.~\ref{fig:prac_circ}(a). In this circuit, $U_{A}$ initializes the state $|0\rangle$ into a specific state $|\psi\rangle$, with its conjugate $|\psi^*\rangle$ serving as the fiducial state \cite{feng2022stabilizer,saraceno2017phase} for Set 2 in Fig.~\ref{fig:2sets}
(detailed in App.~\ref{ap:sic_bell}). The parameters in $U_S$ and $c$ are determined by the given $U_{\mathrm{SIC}}$. 
A simple process for determining these parameters is outlined in Algo.~\ref{algo:algo1} with further details in App.~\ref{ap:algo1}.
\begin{algorithm}[H]
    \caption{Parameters for Practical Circuit}\label{algo:algo1}
    \begin{algorithmic}[1]
    \Require
    Target $U_{\mathrm{SIC}}$, denoted as $U_{\mathrm{SIC}}^{(1)}$; $U_{\mathrm{SIC-2}}(c=1,\Theta_2)$, denoted as $U_{\mathrm{SIC}}^{(2)}$
    \Ensure
    $U_S$, $c$
    \State Extract $V^{(i)}=\left[ \begin{matrix} |\varphi _1^{(i)}\rangle& |\varphi _2^{(i)}\rangle& |\varphi _3^{(i)}\rangle& |\varphi _4^{(i)}\rangle\ \end{matrix} \right]^{\dagger}$ from $U_{\mathrm{SIC}}^{(i)}$, with $i \in \left\{ 1, 2 \right\}$, as described in Eq.~\eqref{eq:povm_elements_matrix}
    \State Normalize: $|\phi _j^{(i)}\rangle=\frac{1}{\sqrt{2}}|\varphi _j^{(i)}\rangle$, with $j \in \left\{ 1,2,3,4\right\}$
    \State Compute $U^{(i)}=|\phi _{1}^{(i)}\rangle \langle 0|+|\phi _{1}^{(i)\bot}\rangle \langle 1|$
    \State Compute and rewrite: 
    
    $\langle \phi _{2}^{(i)}|U^{(i)}=ae^{\mathrm{i}\alpha _{i1}}\langle 0|+\sqrt{1-a^2}e^{\mathrm{i}\alpha _{i2}}\langle 1|$, $a \in \mathbb{R}$.
    \State Get $U_S=U^{(2)}U_{r}(U^{(1)})^{\dagger}$, where 
    
    $U_r=|0\rangle \langle 0|+e^{\mathrm{i}\left( \alpha _{12}-\alpha _{22}-\alpha _{11}+\alpha _{21} \right)}|1\rangle \langle 1|$
    \State Compute $U_{\mathrm{pr}}=V_{\mathrm{down}}^{(1)}(V_{\mathrm{down}}^{(2)}U_{S})^{-1}$, where 
    
    $V_{\mathrm{down}}^{(i)}=\left[ \begin{matrix} |\varphi _3^{(i)}\rangle& |\varphi _4^{(i)}\rangle\\ \end{matrix} \right]^{\dagger}$
    \If {$U_{\mathrm{pr}}$ is diagonal}
    \State Get $c=1$ (i.e., $U_{\mathrm{rel}}=I$)
    \Else 
    \State Get $c=0$ (i.e., $U_{\mathrm{rel}}=\mathrm{CNOT}$)
    \EndIf
    \end{algorithmic}
\end{algorithm}

In fact, the relabeling operation in Fig.~\ref{fig:prac_circ}(a) can be implemented through classical post-processing \cite{liu2022classically}. Based on Eq.~\eqref{eq:povm_syn}, $C_Q$, $R_z$, and $C_{R_z}$ gates can be eliminated. As a result, we finally get the simplified circuit as shown in Fig.~\ref{fig:prac_circ}(b).

In real quantum devices, such 1-CNOT implementation in Fig.~\ref{fig:prac_circ}(b) offers two advantages. First, it reduces the CNOT count to just one, thereby significantly minimizing the impact of CNOT gate noise. 
This method demonstrates superior efficiency compared to other methods, which typically necessitate three CNOT gates \cite{garcia2021learning,oszmaniec2019simulating,garcia2023experimentally}. Second, this implementation offers a general and explicit circuit structure that can be adapted for any single-qubit SIC-POVM with minimal circuit modifications, i.e., requiring only the adjustment of $U_S$. 
Such flexibility makes it more feasible to integrate POVM-related algorithms into real quantum circuits, enhancing the adjustability and practicality.

\begin{figure*}[ht]
\centering
\includegraphics[width=0.8\textwidth]{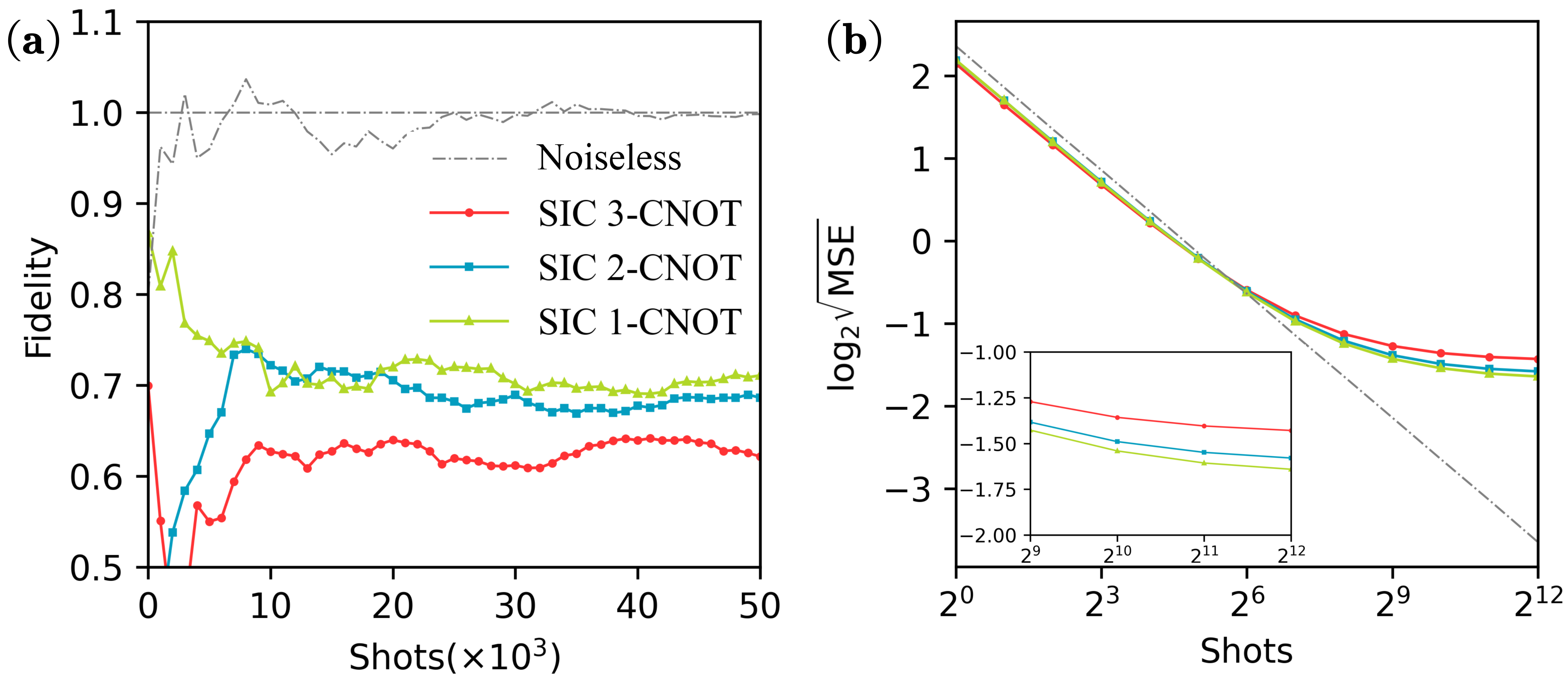}
\caption{\justifying{
The estimation fidelity and mean-square error (MSE) of a 6-qubit GHZ state conducted with three SIC-POVMs compilation methods in the presence of T1/T2 thermal relaxation noise. (a) The estimation fidelity of different compilation methods, which leads to varying counts of CNOT gates, i.e., 18, 12, and 6, respectively. The gray dash-dotted line represents the benchmark for a noiseless scenario for comparison purposes. (b) The MSE relative to the ideal fidelity with different compilation methods. The gray dash-dotted line indicates the variance convergence rate in an ideal setting. As the variance decreases rapidly, the MSEs will converge towards the bias, demonstrating that reducing the CNOT count results in a final outcome with reduced bias.}
}
\label{fig:noise_sic}
\end{figure*}

\section{Applications in Shadow Tomography}\label{sec:appli}
In this section, we showcase the advantages of our optimized compilation method: the noise-resilient performance and compatibility with the POVMs optimization techniques. These aspects are demonstrated using the shadow estimation method with SIC-POVMs \cite{garcia2021learning,acharya2021shadow,nguyen2022optimizing,stricker2022experimental}. We begin with a concise introduction to shadow estimation within the POVM framework.

Shadow estimation involves two key aspects: generating shadows with sufficient information and utilizing these shadows to reconstruct quantum parameters. For the first aspect, an IC-POVM, as described in Eq.~\eqref{eq:icpovm_conditions}, is required to make the measurement channel $\mathcal{M}$ invertible. For an $n$-element IC-POVM, the $\mathcal{M}$ channel can be expressed as
\begin{align}\label{eq:Mchannel}
\begin{aligned}
\mathcal{M} |\rho\rangle \!\rangle =\sum_{i=1}^{n}{|\Pi _i\rangle \!\rangle \langle \!\langle \Pi _i|\rho\rangle \!\rangle}.
\end{aligned}
\end{align}

Similar to the randomized measurements scheme, we treat $|\Pi_{\boldsymbol{b}}\rangle \!\rangle$ as a snapshot and $|\hat{\rho}_{\boldsymbol{b}}\rangle \!\rangle=\mathcal{M}^{-1}{|\Pi_{\boldsymbol{b}}\rangle \!\rangle}$ as the corresponding classical shadow. Here, $\boldsymbol{b}$ denotes the bitstring obtained from measurements performed across the entire quantum system, which is distinct from the measurement result $b$ of the single-qubit IC-POVM.
This bitstring $\boldsymbol{b}$ directly identifies the measurement operator $\Pi _i$ as $\Pi_{\boldsymbol{b}}$. To obtain the classical shadow, it is necessary to calculate the inverse map $\mathcal{M}^{-1}$. This can be achieved either by directly computing the inverse matrix or by employing projective 2-designs for simplification. For the single-qubit case, the inverse map of SIC-POVMs can be represented as
\begin{align}\label{eq:Mchannel_sic}
\begin{aligned}
\mathcal{M} ^{-1}=3I-2|\Phi ^+\rangle \langle \Phi ^+|,
\end{aligned}
\end{align}
with $|\Phi ^+\rangle=\frac{1}{\sqrt{2}}\left( |00\rangle +|11\rangle \right)$. One can verify that it is equivalent to the form of $\mathcal{M} ^{-1}\left( \Pi _b \right) =3\Pi _b-\mathrm{tr}\left( \Pi _b \right) I$. For single-qubit SIC-POVMs, there are only four possible measurement outcomes. This implies that we only need to store the measurement result bitstrings and the four potential values of $\mathcal{M}^{-1}{|\Pi_b\rangle \!\rangle}$. Consequently, classical shadows can be obtained by referring to a pre-computed table, as shown in the blue part of Fig.~\ref{fig:overall}. 

For an $N$-qubit system, finding and implementing a global SIC-POVM presents significant challenges. Alternatively, we can perform single-qubit SIC-POVMs on each qubit individually. The overall inverse map is then the tensor product of the inverse maps for each qubit, given by
\begin{align}\label{eq:Mchannel_sic}
\begin{aligned}
|\hat{\rho}_{\boldsymbol{b}}\rangle \!\rangle =\bigotimes_{i=1}^N{\mathcal{M}_{i}^{-1}}|\Pi _{b_{i}}\rangle \!\rangle ,
\end{aligned}
\end{align}
where
$b_{i}$ denotes the measurement bitstring of the single-qubit SIC-POVM of the $i$th qubit.

\subsection{Noise-Resilient Performance in Optimized Circuits}
One key benefit of optimized compilation is its ability to minimize the impact of noise from CNOT gates by reducing their count. To illustrate this advantage, we examine the task of estimating the fidelity of a 6-qubit GHZ state, as depicted in Fig.~\ref{fig:noise_sic}. The $N$-qubit GHZ state can be represented as
\begin{align}
\begin{aligned}
|\mathrm{GHZ}\left( N \right) \rangle =\frac{1}{\sqrt{2}}(|0\rangle ^{\otimes N}+|1\rangle ^{\otimes N}).
\end{aligned}
\end{align}

Given a state $\rho$ and a target pure state $\Psi$, the fidelity can be calculated using the formula $F=\langle \!\langle \Psi |\rho \rangle \!\rangle $. In the context of shadow estimation, the fidelity between the reconstructed state from $M$ shots and the pure state $\Psi$ can be further calculated as
\begin{align}\label{eq:fidelity}
\begin{aligned}
\hat{F}=\frac{1}{M}\sum_{m=1}^{M}{\langle \!\langle  \Psi |\hat{\rho}_{\boldsymbol{b}^{(m)}} \rangle \!\rangle},
\end{aligned}
\end{align}
where $\boldsymbol{b}^{(m)}$ denotes the result bitstring $\boldsymbol{b}$ of the $m$th shot. 

\begin{figure*}[ht]
\centering
\includegraphics[width=0.8\textwidth]{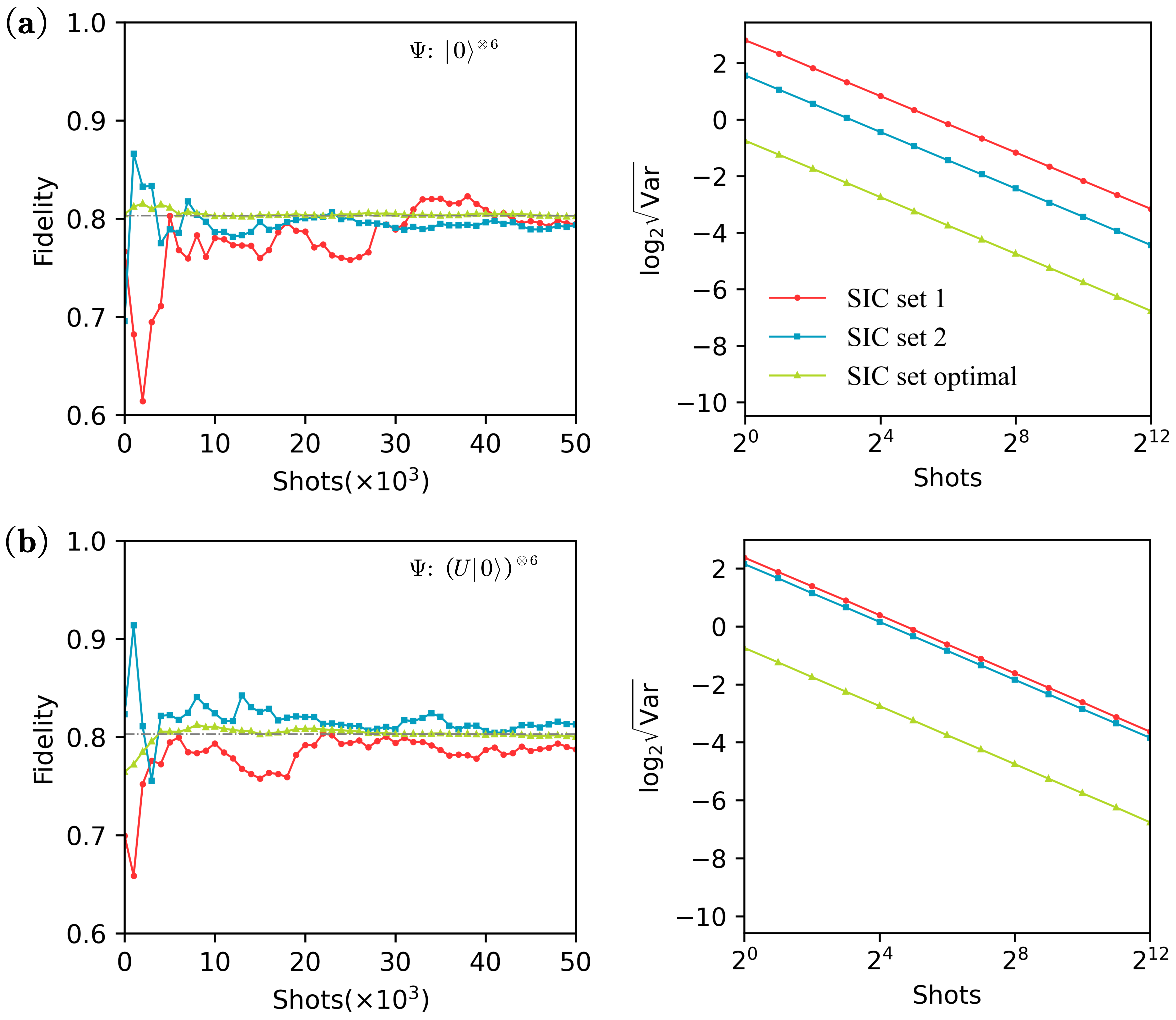}
\caption{\justifying{
The estimation fidelity and variance of 6-qubit mixed states, $\rho =p\frac{I}{2^N}+\left( 1-p \right) \Psi$ with $p=0.2$, using different SIC-POVM sets. The expected fidelity equals to $0.80315$. We set $\Psi = \ket{0}^{\otimes 6}$ in (a) and $\Psi = U\ket{0}^{\otimes 6}$ where $U = e^{\mathrm{i}\sum_{j=0}^3{\gamma_j \sigma_j}}$ with $\vec{\gamma} = [0 \ \frac{\pi}{3} \ \frac{\pi}{3} \ \frac{\pi}{3}]$ in (b). In (a) and (b), the optimized SIC-POVM demonstrates superior variance performance compared with other SIC-POVM sets.
}}
\label{fig:optimal_sic}
\end{figure*}

We utilize the same set of SIC-POVM of Fig.~\ref{fig:povm_cors}. To closely mimic a real quantum computing environment, we adopt a noise model known as T1/T2 thermal relaxation \cite{rost2020simulation}, with the parameters for this noise model following the documentation provided by Qiskit. The thermal relaxation time constant $T_1$ and the dephasing time constant $T_2$ are sampled from normal distributions with a mean of 50 ms, ensuring that $T_2 \leqslant T_1$ to reflect realistic quantum computing conditions. The instruction times for various operations within a quantum computing environment are specified as follows: single-qubit gates are set to 100 ns, CNOT gates to 300 ns, and both reset and measurement operations are allocated 1000 ns each. Noise is applied to every gate, including those used in both the state preparation process and the measurement circuit.

In Fig.~\ref{fig:noise_sic}(a), the gray dash-dot line represents a comparative experiment in an ideal environment, which demonstrates good convergence to the expected fidelity. The other three lines depict scenarios within a noisy environment. Due to the noise present during the state preparation stage, none of these scenarios will achieve the ideal fidelity. However, by reducing the CNOT count, the final bias can exhibit better performance. Given that it acts as a biased estimator of the ideal fidelity, we employ the mean-square error (MSE) to characterize performance, as depicted in Fig.~\ref{fig:noise_sic}(b). Here, it is observed that the variance diminishes as the number of shots increases, leading to a convergence of the MSEs towards the bias. This observation provides a more confident conclusion regarding the performance improvement achieved by reducing the CNOT count.

\subsection{Flexibility in Compiling Various SIC-POVMs}
Another advantage of the compilation scheme is its ability to compile any SIC-POVM using 1 CNOT gate, and it offers 2-CNOT solutions for any single-qubit minimal IC-POVM. This flexibility makes it highly adaptable to various algorithms that focus on optimizing the POVM itself, such as those discussed in Ref.~\cite{garcia2021learning,nguyen2022optimizing}. To illustrate this point, we introduce a toy task aimed at estimating the fidelity of a mixed state using different SIC-POVM settings. The mixed states $\rho$ can be derived from a pure state $\Psi$ through the application of the depolarizing channel, that is
\begin{align}\label{eq:dep}
\begin{aligned}
\rho =p\frac{I}{2^N}+\left( 1-p \right) \Psi .
\end{aligned}
\end{align}

To simplify the optimization process of SIC-POVM, we use the tensor product of identical single-qubit state for the pure state, denoted as $\Psi =\Psi _{1}^{\otimes N}$. In this task, we focus on the compilation of different sets. Therefore, the measurement circuits are considered noise-free, allowing us to calculate the expected fidelity directly from Eq.~\eqref{eq:dep} as $F=\frac{1}{2^N}p+(1-p)$. Here, we set $p=0.2$, and the expected fidelity calculates to $0.803125$.

From Ref.~\cite{nguyen2022optimizing}, we know that for a rank-1 observable, the optimal SIC-POVM for estimation includes one element that is orthogonal to the observable.
In the task of pure state fidelity estimation, where the observable is the pure state itself, the optimized SIC-POVM can be selected by ensuring it includes a state $|\varphi_i\rangle$ that is orthogonal to the pure state. This orthogonality criterion makes the SIC-POVM the optimal choice for this task. For comparison, we also use the SIC-POVMs from Set 1 and Set 2 introduced in Fig.~\ref{fig:2sets}.

In Fig.~\ref{fig:optimal_sic}(a), we employ $\Psi=|0\rangle^{\otimes 6}$ as the initial pure state. All sets yield unbiased estimators, though their convergence speeds vary. The optimized SIC-POVM consistently shows superior variance performance. In Fig.~\ref{fig:optimal_sic}(b), we employ $\Psi=(U|0\rangle)^{\otimes 6}$ as the initial pure state, where $U=e^{\mathrm{i}\sum_{j=0}^3{\gamma _j\sigma _j}}$ and $\vec{\gamma}$ is set as $[0 \ \frac{\pi}{3} \ \frac{\pi}{3} \ \frac{\pi}{3}]$. The performance of Set 1 and Set 2 varies with changes in the initial state, while the optimized SIC-POVM consistently achieves optimal variance performance. Across both examples, the optimized SIC-POVM demonstrates faster convergence and highlights the compatibility of our compilation method.

\section{Conclusion and Outlook}

In this work, we tackled the challenge of minimizing CNOT gate counts in the implementation circuits for single-qubit minimal IC-POVMs and SIC-POVMs. By adjusting parameters that do not affect POVM outcomes, we demonstrated that any single-qubit minimal IC-POVM can be implemented with at most 2 CNOT gates, while SIC-POVMs require only 1. Inspired by Bell measurements, we developed a concise and general SIC-POVM compilation circuit along with an efficient algorithm for parameter determination. Moreover, we applied this optimized circuit to shadow estimation tasks, demonstrating its significant noise resilience and versatility in compiling various SIC-POVMs.
Our work enhances the practicality of POVM-based shadow estimation for real-world quantum platforms. Based on our compilation framework, there are several potential and intriguing directions to be explored in the future. 


In Sec.~\ref{sec:reduceic}, we use a dichotomous approach to identify the 2-CNOT $U_{\mathrm{P}}$ as shown in Fig.~\ref{fig:povm_cors}. While practical and effective, this method allows for various possible $U_{\mathrm{P}}$ choices, leaving room for further optimization based on specific purposes. Future studies could focus on optimizing 2-CNOT $U_{\mathrm{P}}$ to achieve either a fixed-structure circuit or one that minimizes the use of single-qubit gates.
Moreover, note that we only adjust $C_Q$ to find the 2-CNOT $U_{\mathrm{P}}$ in Sec.~\ref{sec:reduceic}. By adjusting both $C_Q$ and $U_{\mathrm{ph}}$,
we expect that minimal IC-POVMs could be further classified into two types: one can be implemented with a 1-CNOT circuit, similar to SIC-POVMs; and the other that can not. Such a classification problem is also a promising point for future research.



In Sec.~\ref{sec:appli}, we employed the tensor product of single-qubit IC-POVMs for measurements. This approach is practical for quantum devices and sufficient for reconstructing information in many-qubit systems. While it is effective for local properties, further exploring global IC-POVMs could enhance performance for global properties. Future research could focus on optimizing the CNOT gate count of $2N$-qubit $U_{\mathrm{P}}$ in global IC-POVMs, potentially utilizing higher-dimensional Cartan decomposition techniques \cite{khaneja2001cartan,mansky2023near}. Additionally, the relationship between SIC-POVMs and Bell measurements, as discussed in App.~\ref{ap:sic_bell}, could extend to higher dimensions. This suggests a potential direction for identifying higher-dimensional Bell bases and corresponding compilation methods.

Additionally, our work reveals that parameters such as phase factors, often overlooked in many tasks, can significantly influence the unitary representation of the circuit and thus the CNOT count. This insight can be leveraged to improve various tasks that involve compiling 2-qubit unitaries with a focus solely on measurement results, such as collective measurements based on qubit circuits \cite{conlon2023approaching,conlon2023discriminating}. The reduction in CNOT and single-qubit gate counts also facilitates the incorporation of error mitigation techniques \cite{endo2018practical,glos2022adaptive}. Future research could also focus on developing a unified framework that combines our compilation scheme with error mitigation approaches, aiming for more robust and efficient quantum computing solutions.


\section{Acknowledgement}
We acknowledge useful discussions with Dayue Qin and Qingyue Zhang.
This work is supported by the National Natural Science Foundation of China (NSFC) (Grant No.~12205048), Innovation Program for Quantum Science and Technology (Grant No.~2021ZD0302000), and the start-up funding of Fudan University. 


\begin{thebibliography}{70}%
\makeatletter
\providecommand \@ifxundefined [1]{%
 \@ifx{#1\undefined}
}%
\providecommand \@ifnum [1]{%
 \ifnum #1\expandafter \@firstoftwo
 \else \expandafter \@secondoftwo
 \fi
}%
\providecommand \@ifx [1]{%
 \ifx #1\expandafter \@firstoftwo
 \else \expandafter \@secondoftwo
 \fi
}%
\providecommand \natexlab [1]{#1}%
\providecommand \enquote  [1]{``#1''}%
\providecommand \bibnamefont  [1]{#1}%
\providecommand \bibfnamefont [1]{#1}%
\providecommand \citenamefont [1]{#1}%
\providecommand \href@noop [0]{\@secondoftwo}%
\providecommand \href [0]{\begingroup \@sanitize@url \@href}%
\providecommand \@href[1]{\@@startlink{#1}\@@href}%
\providecommand \@@href[1]{\endgroup#1\@@endlink}%
\providecommand \@sanitize@url [0]{\catcode `\\12\catcode `\$12\catcode
  `\&12\catcode `\#12\catcode `\^12\catcode `\_12\catcode `\%12\relax}%
\providecommand \@@startlink[1]{}%
\providecommand \@@endlink[0]{}%
\providecommand \url  [0]{\begingroup\@sanitize@url \@url }%
\providecommand \@url [1]{\endgroup\@href {#1}{\urlprefix }}%
\providecommand \urlprefix  [0]{URL }%
\providecommand \Eprint [0]{\href }%
\providecommand \doibase [0]{https://doi.org/}%
\providecommand \selectlanguage [0]{\@gobble}%
\providecommand \bibinfo  [0]{\@secondoftwo}%
\providecommand \bibfield  [0]{\@secondoftwo}%
\providecommand \translation [1]{[#1]}%
\providecommand \BibitemOpen [0]{}%
\providecommand \bibitemStop [0]{}%
\providecommand \bibitemNoStop [0]{.\EOS\space}%
\providecommand \EOS [0]{\spacefactor3000\relax}%
\providecommand \BibitemShut  [1]{\csname bibitem#1\endcsname}%
\let\auto@bib@innerbib\@empty
\bibitem [{\citenamefont {Preskill}(2018)}]{preskill2018quantum}%
  \BibitemOpen
  \bibfield  {author} {\bibinfo {author} {\bibfnamefont {J.}~\bibnamefont
  {Preskill}},\ }\bibfield  {title} {\bibinfo {title} {Quantum computing in the
  nisq era and beyond},\ }\href
  {https://quantum-journal.org/papers/q-2018-08-06-79/} {\bibfield  {journal}
  {\bibinfo  {journal} {Quantum}\ }\textbf {\bibinfo {volume} {2}},\ \bibinfo
  {pages} {79} (\bibinfo {year} {2018})}\BibitemShut {NoStop}%
\bibitem [{\citenamefont {Bharti}\ \emph {et~al.}(2022)\citenamefont {Bharti},
  \citenamefont {Cervera-Lierta}, \citenamefont {Kyaw}, \citenamefont {Haug},
  \citenamefont {Alperin-Lea}, \citenamefont {Anand}, \citenamefont {Degroote},
  \citenamefont {Heimonen}, \citenamefont {Kottmann}, \citenamefont {Menke}
  \emph {et~al.}}]{bharti2022noisy}%
  \BibitemOpen
  \bibfield  {author} {\bibinfo {author} {\bibfnamefont {K.}~\bibnamefont
  {Bharti}}, \bibinfo {author} {\bibfnamefont {A.}~\bibnamefont
  {Cervera-Lierta}}, \bibinfo {author} {\bibfnamefont {T.~H.}\ \bibnamefont
  {Kyaw}}, \bibinfo {author} {\bibfnamefont {T.}~\bibnamefont {Haug}}, \bibinfo
  {author} {\bibfnamefont {S.}~\bibnamefont {Alperin-Lea}}, \bibinfo {author}
  {\bibfnamefont {A.}~\bibnamefont {Anand}}, \bibinfo {author} {\bibfnamefont
  {M.}~\bibnamefont {Degroote}}, \bibinfo {author} {\bibfnamefont
  {H.}~\bibnamefont {Heimonen}}, \bibinfo {author} {\bibfnamefont {J.~S.}\
  \bibnamefont {Kottmann}}, \bibinfo {author} {\bibfnamefont {T.}~\bibnamefont
  {Menke}}, \emph {et~al.},\ }\bibfield  {title} {\bibinfo {title} {Noisy
  intermediate-scale quantum algorithms},\ }\href
  {https://journals.aps.org/rmp/abstract/10.1103/RevModPhys.94.015004}
  {\bibfield  {journal} {\bibinfo  {journal} {Reviews of Modern Physics}\
  }\textbf {\bibinfo {volume} {94}},\ \bibinfo {pages} {015004} (\bibinfo
  {year} {2022})}\BibitemShut {NoStop}%
\bibitem [{\citenamefont {Eisert}\ \emph {et~al.}(2020)\citenamefont {Eisert},
  \citenamefont {Hangleiter}, \citenamefont {Walk}, \citenamefont {Roth},
  \citenamefont {Markham}, \citenamefont {Parekh}, \citenamefont {Chabaud},\
  and\ \citenamefont {Kashefi}}]{eisert2020quantum}%
  \BibitemOpen
  \bibfield  {author} {\bibinfo {author} {\bibfnamefont {J.}~\bibnamefont
  {Eisert}}, \bibinfo {author} {\bibfnamefont {D.}~\bibnamefont {Hangleiter}},
  \bibinfo {author} {\bibfnamefont {N.}~\bibnamefont {Walk}}, \bibinfo {author}
  {\bibfnamefont {I.}~\bibnamefont {Roth}}, \bibinfo {author} {\bibfnamefont
  {D.}~\bibnamefont {Markham}}, \bibinfo {author} {\bibfnamefont
  {R.}~\bibnamefont {Parekh}}, \bibinfo {author} {\bibfnamefont
  {U.}~\bibnamefont {Chabaud}},\ and\ \bibinfo {author} {\bibfnamefont
  {E.}~\bibnamefont {Kashefi}},\ }\bibfield  {title} {\bibinfo {title} {Quantum
  certification and benchmarking},\ }\href
  {https://www.nature.com/articles/s42254-020-0186-4} {\bibfield  {journal}
  {\bibinfo  {journal} {Nature Reviews Physics}\ }\textbf {\bibinfo {volume}
  {2}},\ \bibinfo {pages} {382} (\bibinfo {year} {2020})}\BibitemShut {NoStop}%
\bibitem [{\citenamefont {Kliesch}\ and\ \citenamefont
  {Roth}(2021)}]{kliesch2021theory}%
  \BibitemOpen
  \bibfield  {author} {\bibinfo {author} {\bibfnamefont {M.}~\bibnamefont
  {Kliesch}}\ and\ \bibinfo {author} {\bibfnamefont {I.}~\bibnamefont {Roth}},\
  }\bibfield  {title} {\bibinfo {title} {Theory of quantum system
  certification},\ }\href
  {https://journals.aps.org/prxquantum/abstract/10.1103/PRXQuantum.2.010201}
  {\bibfield  {journal} {\bibinfo  {journal} {PRX quantum}\ }\textbf {\bibinfo
  {volume} {2}},\ \bibinfo {pages} {010201} (\bibinfo {year}
  {2021})}\BibitemShut {NoStop}%
\bibitem [{\citenamefont {Huang}\ \emph {et~al.}(2020)\citenamefont {Huang},
  \citenamefont {Kueng},\ and\ \citenamefont {Preskill}}]{huang2020predicting}%
  \BibitemOpen
  \bibfield  {author} {\bibinfo {author} {\bibfnamefont {H.-Y.}\ \bibnamefont
  {Huang}}, \bibinfo {author} {\bibfnamefont {R.}~\bibnamefont {Kueng}},\ and\
  \bibinfo {author} {\bibfnamefont {J.}~\bibnamefont {Preskill}},\ }\bibfield
  {title} {\bibinfo {title} {Predicting many properties of a quantum system
  from very few measurements},\ }\href
  {https://www.nature.com/articles/s41567-020-0932-7} {\bibfield  {journal}
  {\bibinfo  {journal} {Nature Physics}\ }\textbf {\bibinfo {volume} {16}},\
  \bibinfo {pages} {1050} (\bibinfo {year} {2020})}\BibitemShut {NoStop}%
\bibitem [{\citenamefont {Elben}\ \emph {et~al.}(2023)\citenamefont {Elben},
  \citenamefont {Flammia}, \citenamefont {Huang}, \citenamefont {Kueng},
  \citenamefont {Preskill}, \citenamefont {Vermersch},\ and\ \citenamefont
  {Zoller}}]{elben2023randomized}%
  \BibitemOpen
  \bibfield  {author} {\bibinfo {author} {\bibfnamefont {A.}~\bibnamefont
  {Elben}}, \bibinfo {author} {\bibfnamefont {S.~T.}\ \bibnamefont {Flammia}},
  \bibinfo {author} {\bibfnamefont {H.-Y.}\ \bibnamefont {Huang}}, \bibinfo
  {author} {\bibfnamefont {R.}~\bibnamefont {Kueng}}, \bibinfo {author}
  {\bibfnamefont {J.}~\bibnamefont {Preskill}}, \bibinfo {author}
  {\bibfnamefont {B.}~\bibnamefont {Vermersch}},\ and\ \bibinfo {author}
  {\bibfnamefont {P.}~\bibnamefont {Zoller}},\ }\bibfield  {title} {\bibinfo
  {title} {The randomized measurement toolbox},\ }\href
  {https://www.nature.com/articles/s42254-022-00535-2} {\bibfield  {journal}
  {\bibinfo  {journal} {Nature Reviews Physics}\ }\textbf {\bibinfo {volume}
  {5}},\ \bibinfo {pages} {9} (\bibinfo {year} {2023})}\BibitemShut {NoStop}%
\bibitem [{\citenamefont {Hu}\ \emph {et~al.}(2023)\citenamefont {Hu},
  \citenamefont {Choi},\ and\ \citenamefont {You}}]{hu2023classical}%
  \BibitemOpen
  \bibfield  {author} {\bibinfo {author} {\bibfnamefont {H.-Y.}\ \bibnamefont
  {Hu}}, \bibinfo {author} {\bibfnamefont {S.}~\bibnamefont {Choi}},\ and\
  \bibinfo {author} {\bibfnamefont {Y.-Z.}\ \bibnamefont {You}},\ }\bibfield
  {title} {\bibinfo {title} {Classical shadow tomography with locally scrambled
  quantum dynamics},\ }\href {https://doi.org/10.1103/PhysRevResearch.5.023027}
  {\bibfield  {journal} {\bibinfo  {journal} {Physical Review Research}\
  }\textbf {\bibinfo {volume} {5}},\ \bibinfo {pages} {023027} (\bibinfo {year}
  {2023})}\BibitemShut {NoStop}%
\bibitem [{\citenamefont {Akhtar}\ \emph {et~al.}(2023)\citenamefont {Akhtar},
  \citenamefont {Hu},\ and\ \citenamefont {You}}]{akhtar2023scalable}%
  \BibitemOpen
  \bibfield  {author} {\bibinfo {author} {\bibfnamefont {A.~A.}\ \bibnamefont
  {Akhtar}}, \bibinfo {author} {\bibfnamefont {H.-Y.}\ \bibnamefont {Hu}},\
  and\ \bibinfo {author} {\bibfnamefont {Y.-Z.}\ \bibnamefont {You}},\
  }\bibfield  {title} {\bibinfo {title} {Scalable and flexible classical shadow
  tomography with tensor networks},\ }\href
  {https://quantum-journal.org/papers/q-2023-06-01-1026/} {\bibfield  {journal}
  {\bibinfo  {journal} {Quantum}\ }\textbf {\bibinfo {volume} {7}},\ \bibinfo
  {pages} {1026} (\bibinfo {year} {2023})}\BibitemShut {NoStop}%
\bibitem [{\citenamefont {Bertoni}\ \emph {et~al.}(2024)\citenamefont
  {Bertoni}, \citenamefont {Haferkamp}, \citenamefont {Hinsche}, \citenamefont
  {Ioannou}, \citenamefont {Eisert},\ and\ \citenamefont
  {Pashayan}}]{bertoni2024shallow}%
  \BibitemOpen
  \bibfield  {author} {\bibinfo {author} {\bibfnamefont {C.}~\bibnamefont
  {Bertoni}}, \bibinfo {author} {\bibfnamefont {J.}~\bibnamefont {Haferkamp}},
  \bibinfo {author} {\bibfnamefont {M.}~\bibnamefont {Hinsche}}, \bibinfo
  {author} {\bibfnamefont {M.}~\bibnamefont {Ioannou}}, \bibinfo {author}
  {\bibfnamefont {J.}~\bibnamefont {Eisert}},\ and\ \bibinfo {author}
  {\bibfnamefont {H.}~\bibnamefont {Pashayan}},\ }\bibfield  {title} {\bibinfo
  {title} {Shallow shadows: Expectation estimation using low-depth random
  clifford circuits},\ }\href
  {https://journals.aps.org/prl/abstract/10.1103/PhysRevLett.133.020602}
  {\bibfield  {journal} {\bibinfo  {journal} {Physical Review Letters}\
  }\textbf {\bibinfo {volume} {133}},\ \bibinfo {pages} {020602} (\bibinfo
  {year} {2024})}\BibitemShut {NoStop}%
\bibitem [{\citenamefont {Zhang}\ \emph {et~al.}(2024)\citenamefont {Zhang},
  \citenamefont {Liu},\ and\ \citenamefont {Zhou}}]{zhang2024minimal}%
  \BibitemOpen
  \bibfield  {author} {\bibinfo {author} {\bibfnamefont {Q.}~\bibnamefont
  {Zhang}}, \bibinfo {author} {\bibfnamefont {Q.}~\bibnamefont {Liu}},\ and\
  \bibinfo {author} {\bibfnamefont {Y.}~\bibnamefont {Zhou}},\ }\bibfield
  {title} {\bibinfo {title} {Minimal-clifford shadow estimation by mutually
  unbiased bases},\ }\href
  {https://journals.aps.org/prapplied/abstract/10.1103/PhysRevApplied.21.064001}
  {\bibfield  {journal} {\bibinfo  {journal} {Physical Review Applied}\
  }\textbf {\bibinfo {volume} {21}},\ \bibinfo {pages} {064001} (\bibinfo
  {year} {2024})}\BibitemShut {NoStop}%
\bibitem [{\citenamefont {Imai}\ \emph {et~al.}(2024)\citenamefont {Imai},
  \citenamefont {T{\'o}th},\ and\ \citenamefont
  {G{\"u}hne}}]{imai2024collective}%
  \BibitemOpen
  \bibfield  {author} {\bibinfo {author} {\bibfnamefont {S.}~\bibnamefont
  {Imai}}, \bibinfo {author} {\bibfnamefont {G.}~\bibnamefont {T{\'o}th}},\
  and\ \bibinfo {author} {\bibfnamefont {O.}~\bibnamefont {G{\"u}hne}},\
  }\bibfield  {title} {\bibinfo {title} {Collective randomized measurements in
  quantum information processing},\ }\href
  {https://journals.aps.org/prl/abstract/10.1103/PhysRevLett.133.060203}
  {\bibfield  {journal} {\bibinfo  {journal} {Physical Review Letters}\
  }\textbf {\bibinfo {volume} {133}},\ \bibinfo {pages} {060203} (\bibinfo
  {year} {2024})}\BibitemShut {NoStop}%
\bibitem [{\citenamefont {Sack}\ \emph {et~al.}(2022)\citenamefont {Sack},
  \citenamefont {Medina}, \citenamefont {Michailidis}, \citenamefont {Kueng},\
  and\ \citenamefont {Serbyn}}]{sack2022avoiding}%
  \BibitemOpen
  \bibfield  {author} {\bibinfo {author} {\bibfnamefont {S.~H.}\ \bibnamefont
  {Sack}}, \bibinfo {author} {\bibfnamefont {R.~A.}\ \bibnamefont {Medina}},
  \bibinfo {author} {\bibfnamefont {A.~A.}\ \bibnamefont {Michailidis}},
  \bibinfo {author} {\bibfnamefont {R.}~\bibnamefont {Kueng}},\ and\ \bibinfo
  {author} {\bibfnamefont {M.}~\bibnamefont {Serbyn}},\ }\bibfield  {title}
  {\bibinfo {title} {Avoiding barren plateaus using classical shadows},\ }\href
  {https://journals.aps.org/prxquantum/abstract/10.1103/PRXQuantum.3.020365}
  {\bibfield  {journal} {\bibinfo  {journal} {PRX Quantum}\ }\textbf {\bibinfo
  {volume} {3}},\ \bibinfo {pages} {020365} (\bibinfo {year}
  {2022})}\BibitemShut {NoStop}%
\bibitem [{\citenamefont {Boyd}\ and\ \citenamefont
  {Koczor}(2022)}]{boyd2022training}%
  \BibitemOpen
  \bibfield  {author} {\bibinfo {author} {\bibfnamefont {G.}~\bibnamefont
  {Boyd}}\ and\ \bibinfo {author} {\bibfnamefont {B.}~\bibnamefont {Koczor}},\
  }\bibfield  {title} {\bibinfo {title} {Training variational quantum circuits
  with covar: covariance root finding with classical shadows},\ }\href
  {https://journals.aps.org/prx/abstract/10.1103/PhysRevX.12.041022} {\bibfield
   {journal} {\bibinfo  {journal} {Physical Review X}\ }\textbf {\bibinfo
  {volume} {12}},\ \bibinfo {pages} {041022} (\bibinfo {year}
  {2022})}\BibitemShut {NoStop}%
\bibitem [{\citenamefont {Hu}\ \emph {et~al.}(2022)\citenamefont {Hu},
  \citenamefont {LaRose}, \citenamefont {You}, \citenamefont {Rieffel},\ and\
  \citenamefont {Wang}}]{Hu2022Logical}%
  \BibitemOpen
  \bibfield  {author} {\bibinfo {author} {\bibfnamefont {H.-Y.}\ \bibnamefont
  {Hu}}, \bibinfo {author} {\bibfnamefont {R.}~\bibnamefont {LaRose}}, \bibinfo
  {author} {\bibfnamefont {Y.-Z.}\ \bibnamefont {You}}, \bibinfo {author}
  {\bibfnamefont {E.}~\bibnamefont {Rieffel}},\ and\ \bibinfo {author}
  {\bibfnamefont {Z.}~\bibnamefont {Wang}},\ }\bibfield  {title} {\bibinfo
  {title} {Logical shadow tomography: Efficient estimation of error-mitigated
  observables},\ }\bibfield  {journal} {\bibinfo  {journal} {arXiv preprint
  arXiv:2203.07263}\ }\href {https://doi.org/10.48550/arXiv.2203.07263}
  {10.48550/arXiv.2203.07263} (\bibinfo {year} {2022})\BibitemShut {NoStop}%
\bibitem [{\citenamefont {Seif}\ \emph {et~al.}(2023)\citenamefont {Seif},
  \citenamefont {Cian}, \citenamefont {Zhou}, \citenamefont {Chen},\ and\
  \citenamefont {Jiang}}]{seif2023shadow}%
  \BibitemOpen
  \bibfield  {author} {\bibinfo {author} {\bibfnamefont {A.}~\bibnamefont
  {Seif}}, \bibinfo {author} {\bibfnamefont {Z.-P.}\ \bibnamefont {Cian}},
  \bibinfo {author} {\bibfnamefont {S.}~\bibnamefont {Zhou}}, \bibinfo {author}
  {\bibfnamefont {S.}~\bibnamefont {Chen}},\ and\ \bibinfo {author}
  {\bibfnamefont {L.}~\bibnamefont {Jiang}},\ }\bibfield  {title} {\bibinfo
  {title} {Shadow distillation: Quantum error mitigation with classical shadows
  for near-term quantum processors},\ }\href
  {https://doi.org/10.1103/PRXQuantum.4.010303} {\bibfield  {journal} {\bibinfo
   {journal} {PRX Quantum}\ }\textbf {\bibinfo {volume} {4}},\ \bibinfo {pages}
  {010303} (\bibinfo {year} {2023})}\BibitemShut {NoStop}%
\bibitem [{\citenamefont {Peng}\ \emph {et~al.}(2024)\citenamefont {Peng},
  \citenamefont {Liu}, \citenamefont {Liu}, \citenamefont {Zhang},
  \citenamefont {Zhou},\ and\ \citenamefont {Lu}}]{peng2024experimental}%
  \BibitemOpen
  \bibfield  {author} {\bibinfo {author} {\bibfnamefont {X.-J.}\ \bibnamefont
  {Peng}}, \bibinfo {author} {\bibfnamefont {Q.}~\bibnamefont {Liu}}, \bibinfo
  {author} {\bibfnamefont {L.}~\bibnamefont {Liu}}, \bibinfo {author}
  {\bibfnamefont {T.}~\bibnamefont {Zhang}}, \bibinfo {author} {\bibfnamefont
  {Y.}~\bibnamefont {Zhou}},\ and\ \bibinfo {author} {\bibfnamefont
  {H.}~\bibnamefont {Lu}},\ }\bibfield  {title} {\bibinfo {title} {Experimental
  hybrid shadow tomography and distillation},\ }\href
  {https://arxiv.org/abs/2404.11850} {\bibfield  {journal} {\bibinfo  {journal}
  {arXiv preprint arXiv:2404.11850}\ } (\bibinfo {year} {2024})}\BibitemShut
  {NoStop}%
\bibitem [{\citenamefont {Zhou}\ and\ \citenamefont
  {Liu}(2024)}]{zhou2024hybrid}%
  \BibitemOpen
  \bibfield  {author} {\bibinfo {author} {\bibfnamefont {Y.}~\bibnamefont
  {Zhou}}\ and\ \bibinfo {author} {\bibfnamefont {Z.}~\bibnamefont {Liu}},\
  }\bibfield  {title} {\bibinfo {title} {A hybrid framework for estimating
  nonlinear functions of quantum states},\ }\href
  {https://www.nature.com/articles/s41534-024-00846-5} {\bibfield  {journal}
  {\bibinfo  {journal} {npj Quantum Information}\ }\textbf {\bibinfo {volume}
  {10}},\ \bibinfo {pages} {62} (\bibinfo {year} {2024})}\BibitemShut {NoStop}%
\bibitem [{\citenamefont {Liu}\ \emph {et~al.}(2024)\citenamefont {Liu},
  \citenamefont {Li}, \citenamefont {Yuan}, \citenamefont {Zhu},\ and\
  \citenamefont {Zhou}}]{liu2024auxiliary}%
  \BibitemOpen
  \bibfield  {author} {\bibinfo {author} {\bibfnamefont {Q.}~\bibnamefont
  {Liu}}, \bibinfo {author} {\bibfnamefont {Z.}~\bibnamefont {Li}}, \bibinfo
  {author} {\bibfnamefont {X.}~\bibnamefont {Yuan}}, \bibinfo {author}
  {\bibfnamefont {H.}~\bibnamefont {Zhu}},\ and\ \bibinfo {author}
  {\bibfnamefont {Y.}~\bibnamefont {Zhou}},\ }\bibfield  {title} {\bibinfo
  {title} {Auxiliary-free replica shadow estimation},\ }\href
  {https://arxiv.org/abs/2407.20865} {\bibfield  {journal} {\bibinfo  {journal}
  {arXiv preprint arXiv:2407.20865}\ } (\bibinfo {year} {2024})}\BibitemShut
  {NoStop}%
\bibitem [{\citenamefont {Elben}\ \emph {et~al.}(2020)\citenamefont {Elben},
  \citenamefont {Kueng}, \citenamefont {Huang}, \citenamefont {van Bijnen},
  \citenamefont {Kokail}, \citenamefont {Dalmonte}, \citenamefont {Calabrese},
  \citenamefont {Kraus}, \citenamefont {Preskill}, \citenamefont {Zoller} \emph
  {et~al.}}]{elben2020mixed}%
  \BibitemOpen
  \bibfield  {author} {\bibinfo {author} {\bibfnamefont {A.}~\bibnamefont
  {Elben}}, \bibinfo {author} {\bibfnamefont {R.}~\bibnamefont {Kueng}},
  \bibinfo {author} {\bibfnamefont {H.-Y.}\ \bibnamefont {Huang}}, \bibinfo
  {author} {\bibfnamefont {R.}~\bibnamefont {van Bijnen}}, \bibinfo {author}
  {\bibfnamefont {C.}~\bibnamefont {Kokail}}, \bibinfo {author} {\bibfnamefont
  {M.}~\bibnamefont {Dalmonte}}, \bibinfo {author} {\bibfnamefont
  {P.}~\bibnamefont {Calabrese}}, \bibinfo {author} {\bibfnamefont
  {B.}~\bibnamefont {Kraus}}, \bibinfo {author} {\bibfnamefont
  {J.}~\bibnamefont {Preskill}}, \bibinfo {author} {\bibfnamefont
  {P.}~\bibnamefont {Zoller}}, \emph {et~al.},\ }\bibfield  {title} {\bibinfo
  {title} {Mixed-state entanglement from local randomized measurements},\
  }\href {https://journals.aps.org/prl/abstract/10.1103/PhysRevLett.125.200501}
  {\bibfield  {journal} {\bibinfo  {journal} {Physical Review Letters}\
  }\textbf {\bibinfo {volume} {125}},\ \bibinfo {pages} {200501} (\bibinfo
  {year} {2020})}\BibitemShut {NoStop}%
\bibitem [{\citenamefont {Rath}\ \emph {et~al.}(2021)\citenamefont {Rath},
  \citenamefont {Branciard}, \citenamefont {Minguzzi},\ and\ \citenamefont
  {Vermersch}}]{rath2021quantum}%
  \BibitemOpen
  \bibfield  {author} {\bibinfo {author} {\bibfnamefont {A.}~\bibnamefont
  {Rath}}, \bibinfo {author} {\bibfnamefont {C.}~\bibnamefont {Branciard}},
  \bibinfo {author} {\bibfnamefont {A.}~\bibnamefont {Minguzzi}},\ and\
  \bibinfo {author} {\bibfnamefont {B.}~\bibnamefont {Vermersch}},\ }\bibfield
  {title} {\bibinfo {title} {Quantum fisher information from randomized
  measurements},\ }\href
  {https://journals.aps.org/prl/abstract/10.1103/PhysRevLett.127.260501}
  {\bibfield  {journal} {\bibinfo  {journal} {Physical Review Letters}\
  }\textbf {\bibinfo {volume} {127}},\ \bibinfo {pages} {260501} (\bibinfo
  {year} {2021})}\BibitemShut {NoStop}%
\bibitem [{\citenamefont {Liu}\ \emph {et~al.}(2022{\natexlab{a}})\citenamefont
  {Liu}, \citenamefont {Tang}, \citenamefont {Dai}, \citenamefont {Liu},
  \citenamefont {Chen},\ and\ \citenamefont {Ma}}]{liu2022detecting}%
  \BibitemOpen
  \bibfield  {author} {\bibinfo {author} {\bibfnamefont {Z.}~\bibnamefont
  {Liu}}, \bibinfo {author} {\bibfnamefont {Y.}~\bibnamefont {Tang}}, \bibinfo
  {author} {\bibfnamefont {H.}~\bibnamefont {Dai}}, \bibinfo {author}
  {\bibfnamefont {P.}~\bibnamefont {Liu}}, \bibinfo {author} {\bibfnamefont
  {S.}~\bibnamefont {Chen}},\ and\ \bibinfo {author} {\bibfnamefont
  {X.}~\bibnamefont {Ma}},\ }\bibfield  {title} {\bibinfo {title} {Detecting
  entanglement in quantum many-body systems via permutation moments},\ }\href
  {https://journals.aps.org/prl/abstract/10.1103/PhysRevLett.129.260501}
  {\bibfield  {journal} {\bibinfo  {journal} {Physical Review Letters}\
  }\textbf {\bibinfo {volume} {129}},\ \bibinfo {pages} {260501} (\bibinfo
  {year} {2022}{\natexlab{a}})}\BibitemShut {NoStop}%
\bibitem [{\citenamefont {Garcia}\ \emph {et~al.}(2021)\citenamefont {Garcia},
  \citenamefont {Zhou},\ and\ \citenamefont {Jaffe}}]{garcia2021quantum}%
  \BibitemOpen
  \bibfield  {author} {\bibinfo {author} {\bibfnamefont {R.~J.}\ \bibnamefont
  {Garcia}}, \bibinfo {author} {\bibfnamefont {Y.}~\bibnamefont {Zhou}},\ and\
  \bibinfo {author} {\bibfnamefont {A.}~\bibnamefont {Jaffe}},\ }\bibfield
  {title} {\bibinfo {title} {Quantum scrambling with classical shadows},\
  }\href
  {https://journals.aps.org/prresearch/abstract/10.1103/PhysRevResearch.3.033155}
  {\bibfield  {journal} {\bibinfo  {journal} {Physical Review Research}\
  }\textbf {\bibinfo {volume} {3}},\ \bibinfo {pages} {033155} (\bibinfo {year}
  {2021})}\BibitemShut {NoStop}%
\bibitem [{\citenamefont {McGinley}\ \emph {et~al.}(2022)\citenamefont
  {McGinley}, \citenamefont {Leontica}, \citenamefont {Garratt}, \citenamefont
  {Jovanovic},\ and\ \citenamefont {Simon}}]{mcginley2022quantifying}%
  \BibitemOpen
  \bibfield  {author} {\bibinfo {author} {\bibfnamefont {M.}~\bibnamefont
  {McGinley}}, \bibinfo {author} {\bibfnamefont {S.}~\bibnamefont {Leontica}},
  \bibinfo {author} {\bibfnamefont {S.~J.}\ \bibnamefont {Garratt}}, \bibinfo
  {author} {\bibfnamefont {J.}~\bibnamefont {Jovanovic}},\ and\ \bibinfo
  {author} {\bibfnamefont {S.~H.}\ \bibnamefont {Simon}},\ }\bibfield  {title}
  {\bibinfo {title} {Quantifying information scrambling via classical shadow
  tomography on programmable quantum simulators},\ }\href
  {https://journals.aps.org/pra/abstract/10.1103/PhysRevA.106.012441}
  {\bibfield  {journal} {\bibinfo  {journal} {Physical Review A}\ }\textbf
  {\bibinfo {volume} {106}},\ \bibinfo {pages} {012441} (\bibinfo {year}
  {2022})}\BibitemShut {NoStop}%
\bibitem [{\citenamefont {Wallman}(2018)}]{wallman2018randomized}%
  \BibitemOpen
  \bibfield  {author} {\bibinfo {author} {\bibfnamefont {J.~J.}\ \bibnamefont
  {Wallman}},\ }\bibfield  {title} {\bibinfo {title} {Randomized benchmarking
  with gate-dependent noise},\ }\href
  {https://quantum-journal.org/papers/q-2018-01-29-47/} {\bibfield  {journal}
  {\bibinfo  {journal} {Quantum}\ }\textbf {\bibinfo {volume} {2}},\ \bibinfo
  {pages} {47} (\bibinfo {year} {2018})}\BibitemShut {NoStop}%
\bibitem [{\citenamefont {Chen}\ \emph {et~al.}(2021)\citenamefont {Chen},
  \citenamefont {Yu}, \citenamefont {Zeng},\ and\ \citenamefont
  {Flammia}}]{chen2021robust}%
  \BibitemOpen
  \bibfield  {author} {\bibinfo {author} {\bibfnamefont {S.}~\bibnamefont
  {Chen}}, \bibinfo {author} {\bibfnamefont {W.}~\bibnamefont {Yu}}, \bibinfo
  {author} {\bibfnamefont {P.}~\bibnamefont {Zeng}},\ and\ \bibinfo {author}
  {\bibfnamefont {S.~T.}\ \bibnamefont {Flammia}},\ }\bibfield  {title}
  {\bibinfo {title} {Robust shadow estimation},\ }\href
  {https://journals.aps.org/prxquantum/abstract/10.1103/PRXQuantum.2.030348}
  {\bibfield  {journal} {\bibinfo  {journal} {PRX Quantum}\ }\textbf {\bibinfo
  {volume} {2}},\ \bibinfo {pages} {030348} (\bibinfo {year}
  {2021})}\BibitemShut {NoStop}%
\bibitem [{\citenamefont {Brieger}\ \emph {et~al.}(2023)\citenamefont
  {Brieger}, \citenamefont {Heinrich}, \citenamefont {Roth},\ and\
  \citenamefont {Kliesch}}]{brieger2023stability}%
  \BibitemOpen
  \bibfield  {author} {\bibinfo {author} {\bibfnamefont {R.}~\bibnamefont
  {Brieger}}, \bibinfo {author} {\bibfnamefont {M.}~\bibnamefont {Heinrich}},
  \bibinfo {author} {\bibfnamefont {I.}~\bibnamefont {Roth}},\ and\ \bibinfo
  {author} {\bibfnamefont {M.}~\bibnamefont {Kliesch}},\ }\bibfield  {title}
  {\bibinfo {title} {Stability of classical shadows under gate-dependent
  noise},\ }\href {https://arxiv.org/abs/2310.19947} {\bibfield  {journal}
  {\bibinfo  {journal} {arXiv preprint arXiv:2310.19947}\ } (\bibinfo {year}
  {2023})}\BibitemShut {NoStop}%
\bibitem [{\citenamefont {Acharya}\ \emph {et~al.}(2021)\citenamefont
  {Acharya}, \citenamefont {Saha},\ and\ \citenamefont
  {Sengupta}}]{acharya2021shadow}%
  \BibitemOpen
  \bibfield  {author} {\bibinfo {author} {\bibfnamefont {A.}~\bibnamefont
  {Acharya}}, \bibinfo {author} {\bibfnamefont {S.}~\bibnamefont {Saha}},\ and\
  \bibinfo {author} {\bibfnamefont {A.~M.}\ \bibnamefont {Sengupta}},\
  }\bibfield  {title} {\bibinfo {title} {Shadow tomography based on
  informationally complete positive operator-valued measure},\ }\href
  {https://journals.aps.org/pra/abstract/10.1103/PhysRevA.104.052418}
  {\bibfield  {journal} {\bibinfo  {journal} {Physical Review A}\ }\textbf
  {\bibinfo {volume} {104}},\ \bibinfo {pages} {052418} (\bibinfo {year}
  {2021})}\BibitemShut {NoStop}%
\bibitem [{\citenamefont {Nguyen}\ \emph {et~al.}(2022)\citenamefont {Nguyen},
  \citenamefont {B{\"o}nsel}, \citenamefont {Steinberg},\ and\ \citenamefont
  {G{\"u}hne}}]{nguyen2022optimizing}%
  \BibitemOpen
  \bibfield  {author} {\bibinfo {author} {\bibfnamefont {H.~C.}\ \bibnamefont
  {Nguyen}}, \bibinfo {author} {\bibfnamefont {J.~L.}\ \bibnamefont
  {B{\"o}nsel}}, \bibinfo {author} {\bibfnamefont {J.}~\bibnamefont
  {Steinberg}},\ and\ \bibinfo {author} {\bibfnamefont {O.}~\bibnamefont
  {G{\"u}hne}},\ }\bibfield  {title} {\bibinfo {title} {Optimizing shadow
  tomography with generalized measurements},\ }\href
  {https://journals.aps.org/prl/abstract/10.1103/PhysRevLett.129.220502}
  {\bibfield  {journal} {\bibinfo  {journal} {Physical Review Letters}\
  }\textbf {\bibinfo {volume} {129}},\ \bibinfo {pages} {220502} (\bibinfo
  {year} {2022})}\BibitemShut {NoStop}%
\bibitem [{\citenamefont {Innocenti}\ \emph {et~al.}(2023)\citenamefont
  {Innocenti}, \citenamefont {Lorenzo}, \citenamefont {Palmisano},
  \citenamefont {Albarelli}, \citenamefont {Ferraro}, \citenamefont
  {Paternostro},\ and\ \citenamefont {Palma}}]{innocenti2023shadow}%
  \BibitemOpen
  \bibfield  {author} {\bibinfo {author} {\bibfnamefont {L.}~\bibnamefont
  {Innocenti}}, \bibinfo {author} {\bibfnamefont {S.}~\bibnamefont {Lorenzo}},
  \bibinfo {author} {\bibfnamefont {I.}~\bibnamefont {Palmisano}}, \bibinfo
  {author} {\bibfnamefont {F.}~\bibnamefont {Albarelli}}, \bibinfo {author}
  {\bibfnamefont {A.}~\bibnamefont {Ferraro}}, \bibinfo {author} {\bibfnamefont
  {M.}~\bibnamefont {Paternostro}},\ and\ \bibinfo {author} {\bibfnamefont
  {G.~M.}\ \bibnamefont {Palma}},\ }\bibfield  {title} {\bibinfo {title}
  {Shadow tomography on general measurement frames},\ }\href
  {https://journals.aps.org/prxquantum/abstract/10.1103/PRXQuantum.4.040328}
  {\bibfield  {journal} {\bibinfo  {journal} {PRX Quantum}\ }\textbf {\bibinfo
  {volume} {4}},\ \bibinfo {pages} {040328} (\bibinfo {year}
  {2023})}\BibitemShut {NoStop}%
\bibitem [{\citenamefont {Garc{\'\i}a-P{\'e}rez}\ \emph
  {et~al.}(2021)\citenamefont {Garc{\'\i}a-P{\'e}rez}, \citenamefont {Rossi},
  \citenamefont {Sokolov}, \citenamefont {Tacchino}, \citenamefont
  {Barkoutsos}, \citenamefont {Mazzola}, \citenamefont {Tavernelli},\ and\
  \citenamefont {Maniscalco}}]{garcia2021learning}%
  \BibitemOpen
  \bibfield  {author} {\bibinfo {author} {\bibfnamefont {G.}~\bibnamefont
  {Garc{\'\i}a-P{\'e}rez}}, \bibinfo {author} {\bibfnamefont {M.~A.}\
  \bibnamefont {Rossi}}, \bibinfo {author} {\bibfnamefont {B.}~\bibnamefont
  {Sokolov}}, \bibinfo {author} {\bibfnamefont {F.}~\bibnamefont {Tacchino}},
  \bibinfo {author} {\bibfnamefont {P.~K.}\ \bibnamefont {Barkoutsos}},
  \bibinfo {author} {\bibfnamefont {G.}~\bibnamefont {Mazzola}}, \bibinfo
  {author} {\bibfnamefont {I.}~\bibnamefont {Tavernelli}},\ and\ \bibinfo
  {author} {\bibfnamefont {S.}~\bibnamefont {Maniscalco}},\ }\bibfield  {title}
  {\bibinfo {title} {Learning to measure: Adaptive informationally complete
  generalized measurements for quantum algorithms},\ }\href
  {https://journals.aps.org/prxquantum/abstract/10.1103/PRXQuantum.2.040342}
  {\bibfield  {journal} {\bibinfo  {journal} {Prx quantum}\ }\textbf {\bibinfo
  {volume} {2}},\ \bibinfo {pages} {040342} (\bibinfo {year}
  {2021})}\BibitemShut {NoStop}%
\bibitem [{\citenamefont {Stricker}\ \emph {et~al.}(2022)\citenamefont
  {Stricker}, \citenamefont {Meth}, \citenamefont {Postler}, \citenamefont
  {Edmunds}, \citenamefont {Ferrie}, \citenamefont {Blatt}, \citenamefont
  {Schindler}, \citenamefont {Monz}, \citenamefont {Kueng},\ and\ \citenamefont
  {Ringbauer}}]{stricker2022experimental}%
  \BibitemOpen
  \bibfield  {author} {\bibinfo {author} {\bibfnamefont {R.}~\bibnamefont
  {Stricker}}, \bibinfo {author} {\bibfnamefont {M.}~\bibnamefont {Meth}},
  \bibinfo {author} {\bibfnamefont {L.}~\bibnamefont {Postler}}, \bibinfo
  {author} {\bibfnamefont {C.}~\bibnamefont {Edmunds}}, \bibinfo {author}
  {\bibfnamefont {C.}~\bibnamefont {Ferrie}}, \bibinfo {author} {\bibfnamefont
  {R.}~\bibnamefont {Blatt}}, \bibinfo {author} {\bibfnamefont
  {P.}~\bibnamefont {Schindler}}, \bibinfo {author} {\bibfnamefont
  {T.}~\bibnamefont {Monz}}, \bibinfo {author} {\bibfnamefont {R.}~\bibnamefont
  {Kueng}},\ and\ \bibinfo {author} {\bibfnamefont {M.}~\bibnamefont
  {Ringbauer}},\ }\bibfield  {title} {\bibinfo {title} {Experimental
  single-setting quantum state tomography},\ }\href
  {https://journals.aps.org/prxquantum/abstract/10.1103/PRXQuantum.3.040310}
  {\bibfield  {journal} {\bibinfo  {journal} {PRX Quantum}\ }\textbf {\bibinfo
  {volume} {3}},\ \bibinfo {pages} {040310} (\bibinfo {year}
  {2022})}\BibitemShut {NoStop}%
\bibitem [{\citenamefont {Tannu}\ and\ \citenamefont
  {Qureshi}(2019)}]{tannu2019not}%
  \BibitemOpen
  \bibfield  {author} {\bibinfo {author} {\bibfnamefont {S.~S.}\ \bibnamefont
  {Tannu}}\ and\ \bibinfo {author} {\bibfnamefont {M.~K.}\ \bibnamefont
  {Qureshi}},\ }\bibfield  {title} {\bibinfo {title} {Not all qubits are
  created equal: A case for variability-aware policies for nisq-era quantum
  computers},\ }in\ \href {https://dl.acm.org/doi/abs/10.1145/3297858.3304007}
  {\emph {\bibinfo {booktitle} {Proceedings of the twenty-fourth international
  conference on architectural support for programming languages and operating
  systems}}}\ (\bibinfo {year} {2019})\ pp.\ \bibinfo {pages}
  {987--999}\BibitemShut {NoStop}%
\bibitem [{\citenamefont {Fischer}\ \emph {et~al.}(2022)\citenamefont
  {Fischer}, \citenamefont {Miller}, \citenamefont {Tacchino}, \citenamefont
  {Barkoutsos}, \citenamefont {Egger},\ and\ \citenamefont
  {Tavernelli}}]{fischer2022ancilla}%
  \BibitemOpen
  \bibfield  {author} {\bibinfo {author} {\bibfnamefont {L.~E.}\ \bibnamefont
  {Fischer}}, \bibinfo {author} {\bibfnamefont {D.}~\bibnamefont {Miller}},
  \bibinfo {author} {\bibfnamefont {F.}~\bibnamefont {Tacchino}}, \bibinfo
  {author} {\bibfnamefont {P.~K.}\ \bibnamefont {Barkoutsos}}, \bibinfo
  {author} {\bibfnamefont {D.~J.}\ \bibnamefont {Egger}},\ and\ \bibinfo
  {author} {\bibfnamefont {I.}~\bibnamefont {Tavernelli}},\ }\bibfield  {title}
  {\bibinfo {title} {Ancilla-free implementation of generalized measurements
  for qubits embedded in a qudit space},\ }\href
  {https://journals.aps.org/prresearch/abstract/10.1103/PhysRevResearch.4.033027}
  {\bibfield  {journal} {\bibinfo  {journal} {Physical review research}\
  }\textbf {\bibinfo {volume} {4}},\ \bibinfo {pages} {033027} (\bibinfo {year}
  {2022})}\BibitemShut {NoStop}%
\bibitem [{\citenamefont {Zhu}\ and\ \citenamefont
  {Englert}(2011)}]{zhu2011quantum}%
  \BibitemOpen
  \bibfield  {author} {\bibinfo {author} {\bibfnamefont {H.}~\bibnamefont
  {Zhu}}\ and\ \bibinfo {author} {\bibfnamefont {B.-G.}\ \bibnamefont
  {Englert}},\ }\bibfield  {title} {\bibinfo {title} {Quantum state tomography
  with fully symmetric measurements and product measurements},\ }\href
  {https://journals.aps.org/pra/abstract/10.1103/PhysRevA.84.022327} {\bibfield
   {journal} {\bibinfo  {journal} {Physical Review A—Atomic, Molecular, and
  Optical Physics}\ }\textbf {\bibinfo {volume} {84}},\ \bibinfo {pages}
  {022327} (\bibinfo {year} {2011})}\BibitemShut {NoStop}%
\bibitem [{\citenamefont {Renes}\ \emph {et~al.}(2004)\citenamefont {Renes},
  \citenamefont {Blume-Kohout}, \citenamefont {Scott},\ and\ \citenamefont
  {Caves}}]{renes2004symmetric}%
  \BibitemOpen
  \bibfield  {author} {\bibinfo {author} {\bibfnamefont {J.~M.}\ \bibnamefont
  {Renes}}, \bibinfo {author} {\bibfnamefont {R.}~\bibnamefont {Blume-Kohout}},
  \bibinfo {author} {\bibfnamefont {A.~J.}\ \bibnamefont {Scott}},\ and\
  \bibinfo {author} {\bibfnamefont {C.~M.}\ \bibnamefont {Caves}},\ }\bibfield
  {title} {\bibinfo {title} {Symmetric informationally complete quantum
  measurements},\ }\href
  {https://pubs.aip.org/aip/jmp/article-abstract/45/6/2171/231563/Symmetric-informationally-complete-quantum}
  {\bibfield  {journal} {\bibinfo  {journal} {Journal of Mathematical Physics}\
  }\textbf {\bibinfo {volume} {45}},\ \bibinfo {pages} {2171} (\bibinfo {year}
  {2004})}\BibitemShut {NoStop}%
\bibitem [{\citenamefont {Zhu}(2010)}]{zhu2010sic}%
  \BibitemOpen
  \bibfield  {author} {\bibinfo {author} {\bibfnamefont {H.}~\bibnamefont
  {Zhu}},\ }\bibfield  {title} {\bibinfo {title} {Sic povms and clifford groups
  in prime dimensions},\ }\href
  {https://iopscience.iop.org/article/10.1088/1751-8113/43/30/305305/meta}
  {\bibfield  {journal} {\bibinfo  {journal} {Journal of Physics A:
  Mathematical and Theoretical}\ }\textbf {\bibinfo {volume} {43}},\ \bibinfo
  {pages} {305305} (\bibinfo {year} {2010})}\BibitemShut {NoStop}%
\bibitem [{\citenamefont {Huangjun}(2012)}]{huangjun2012quantum}%
  \BibitemOpen
  \bibfield  {author} {\bibinfo {author} {\bibfnamefont {Z.}~\bibnamefont
  {Huangjun}},\ }\emph {\bibinfo {title} {Quantum state estimation and
  symmetric informationally complete POMs}},\ \href
  {https://core.ac.uk/download/pdf/48656576.pdf} {Ph.D. thesis},\ \bibinfo
  {school} {National University of Singapore} (\bibinfo {year}
  {2012})\BibitemShut {NoStop}%
\bibitem [{\citenamefont {Fuchs}\ \emph {et~al.}(2017)\citenamefont {Fuchs},
  \citenamefont {Hoang},\ and\ \citenamefont {Stacey}}]{fuchs2017sic}%
  \BibitemOpen
  \bibfield  {author} {\bibinfo {author} {\bibfnamefont {C.~A.}\ \bibnamefont
  {Fuchs}}, \bibinfo {author} {\bibfnamefont {M.~C.}\ \bibnamefont {Hoang}},\
  and\ \bibinfo {author} {\bibfnamefont {B.~C.}\ \bibnamefont {Stacey}},\
  }\bibfield  {title} {\bibinfo {title} {The sic question: History and state of
  play},\ }\href {https://www.mdpi.com/2075-1680/6/3/21} {\bibfield  {journal}
  {\bibinfo  {journal} {Axioms}\ }\textbf {\bibinfo {volume} {6}},\ \bibinfo
  {pages} {21} (\bibinfo {year} {2017})}\BibitemShut {NoStop}%
\bibitem [{\citenamefont {Appleby}\ \emph {et~al.}(2018)\citenamefont
  {Appleby}, \citenamefont {Chien}, \citenamefont {Flammia},\ and\
  \citenamefont {Waldron}}]{appleby2018constructing}%
  \BibitemOpen
  \bibfield  {author} {\bibinfo {author} {\bibfnamefont {M.}~\bibnamefont
  {Appleby}}, \bibinfo {author} {\bibfnamefont {T.-Y.}\ \bibnamefont {Chien}},
  \bibinfo {author} {\bibfnamefont {S.}~\bibnamefont {Flammia}},\ and\ \bibinfo
  {author} {\bibfnamefont {S.}~\bibnamefont {Waldron}},\ }\bibfield  {title}
  {\bibinfo {title} {Constructing exact symmetric informationally complete
  measurements from numerical solutions},\ }\href
  {https://iopscience.iop.org/article/10.1088/1751-8121/aab4cd/meta} {\bibfield
   {journal} {\bibinfo  {journal} {Journal of Physics A: Mathematical and
  Theoretical}\ }\textbf {\bibinfo {volume} {51}},\ \bibinfo {pages} {165302}
  (\bibinfo {year} {2018})}\BibitemShut {NoStop}%
\bibitem [{\citenamefont {Garc{\'\i}a-P{\'e}rez}\ \emph
  {et~al.}(2023)\citenamefont {Garc{\'\i}a-P{\'e}rez}, \citenamefont {Kerppo},
  \citenamefont {Rossi},\ and\ \citenamefont
  {Maniscalco}}]{garcia2023experimentally}%
  \BibitemOpen
  \bibfield  {author} {\bibinfo {author} {\bibfnamefont {G.}~\bibnamefont
  {Garc{\'\i}a-P{\'e}rez}}, \bibinfo {author} {\bibfnamefont {O.}~\bibnamefont
  {Kerppo}}, \bibinfo {author} {\bibfnamefont {M.~A.}\ \bibnamefont {Rossi}},\
  and\ \bibinfo {author} {\bibfnamefont {S.}~\bibnamefont {Maniscalco}},\
  }\bibfield  {title} {\bibinfo {title} {Experimentally accessible
  nonseparability criteria for multipartite-entanglement-structure detection},\
  }\href
  {https://journals.aps.org/prresearch/abstract/10.1103/PhysRevResearch.5.013226}
  {\bibfield  {journal} {\bibinfo  {journal} {Physical Review Research}\
  }\textbf {\bibinfo {volume} {5}},\ \bibinfo {pages} {013226} (\bibinfo {year}
  {2023})}\BibitemShut {NoStop}%
\bibitem [{\citenamefont {Yordanov}\ and\ \citenamefont
  {Barnes}(2019)}]{yordanov2019implementation}%
  \BibitemOpen
  \bibfield  {author} {\bibinfo {author} {\bibfnamefont {Y.~S.}\ \bibnamefont
  {Yordanov}}\ and\ \bibinfo {author} {\bibfnamefont {C.~H.}\ \bibnamefont
  {Barnes}},\ }\bibfield  {title} {\bibinfo {title} {Implementation of a
  general single-qubit positive operator-valued measure on a circuit-based
  quantum computer},\ }\href
  {https://journals.aps.org/pra/abstract/10.1103/PhysRevA.100.062317}
  {\bibfield  {journal} {\bibinfo  {journal} {Physical Review A}\ }\textbf
  {\bibinfo {volume} {100}},\ \bibinfo {pages} {062317} (\bibinfo {year}
  {2019})}\BibitemShut {NoStop}%
\bibitem [{\citenamefont {Singh}\ \emph {et~al.}(2022)\citenamefont {Singh},
  \citenamefont {Arvind},\ and\ \citenamefont
  {Goyal}}]{singh2022implementation}%
  \BibitemOpen
  \bibfield  {author} {\bibinfo {author} {\bibfnamefont {J.}~\bibnamefont
  {Singh}}, \bibinfo {author} {\bibnamefont {Arvind}},\ and\ \bibinfo {author}
  {\bibfnamefont {S.~K.}\ \bibnamefont {Goyal}},\ }\bibfield  {title} {\bibinfo
  {title} {Implementation of discrete positive operator valued measures on
  linear optical systems using cosine-sine decomposition},\ }\href
  {https://journals.aps.org/prresearch/abstract/10.1103/PhysRevResearch.4.013007}
  {\bibfield  {journal} {\bibinfo  {journal} {Physical Review Research}\
  }\textbf {\bibinfo {volume} {4}},\ \bibinfo {pages} {013007} (\bibinfo {year}
  {2022})}\BibitemShut {NoStop}%
\bibitem [{\citenamefont {Peres}(1990)}]{peres1990neumark}%
  \BibitemOpen
  \bibfield  {author} {\bibinfo {author} {\bibfnamefont {A.}~\bibnamefont
  {Peres}},\ }\bibfield  {title} {\bibinfo {title} {Neumark's theorem and
  quantum inseparability},\ }\href
  {https://link.springer.com/article/10.1007/bf01883517} {\bibfield  {journal}
  {\bibinfo  {journal} {Foundations of Physics}\ }\textbf {\bibinfo {volume}
  {20}},\ \bibinfo {pages} {1441} (\bibinfo {year} {1990})}\BibitemShut
  {NoStop}%
\bibitem [{\citenamefont {Chen}\ \emph {et~al.}(2007)\citenamefont {Chen},
  \citenamefont {Bergou}, \citenamefont {Zhu},\ and\ \citenamefont
  {Guo}}]{chen2007ancilla}%
  \BibitemOpen
  \bibfield  {author} {\bibinfo {author} {\bibfnamefont {P.-X.}\ \bibnamefont
  {Chen}}, \bibinfo {author} {\bibfnamefont {J.~A.}\ \bibnamefont {Bergou}},
  \bibinfo {author} {\bibfnamefont {S.-Y.}\ \bibnamefont {Zhu}},\ and\ \bibinfo
  {author} {\bibfnamefont {G.-C.}\ \bibnamefont {Guo}},\ }\bibfield  {title}
  {\bibinfo {title} {Ancilla dimensions needed to carry out
  positive-operator-valued measurement},\ }\href
  {https://journals.aps.org/pra/abstract/10.1103/PhysRevA.76.060303} {\bibfield
   {journal} {\bibinfo  {journal} {Physical Review A—Atomic, Molecular, and
  Optical Physics}\ }\textbf {\bibinfo {volume} {76}},\ \bibinfo {pages}
  {060303} (\bibinfo {year} {2007})}\BibitemShut {NoStop}%
\bibitem [{\citenamefont {Tabia}(2012)}]{tabia2012experimental}%
  \BibitemOpen
  \bibfield  {author} {\bibinfo {author} {\bibfnamefont {G.~N.~M.}\
  \bibnamefont {Tabia}},\ }\bibfield  {title} {\bibinfo {title} {Experimental
  scheme for qubit and qutrit symmetric informationally complete positive
  operator-valued measurements using multiport devices},\ }\href
  {https://journals.aps.org/pra/abstract/10.1103/PhysRevA.86.062107} {\bibfield
   {journal} {\bibinfo  {journal} {Physical Review A—Atomic, Molecular, and
  Optical Physics}\ }\textbf {\bibinfo {volume} {86}},\ \bibinfo {pages}
  {062107} (\bibinfo {year} {2012})}\BibitemShut {NoStop}%
\bibitem [{\citenamefont {Aleksandrowicz}\ \emph {et~al.}(2019)\citenamefont
  {Aleksandrowicz}, \citenamefont {Alexander}, \citenamefont {Barkoutsos},
  \citenamefont {Bello}, \citenamefont {Ben-Haim}, \citenamefont {Bucher},\
  and\ \citenamefont {et~al.}}]{gadi_aleksandrowicz_2019_2562111}%
  \BibitemOpen
  \bibfield  {author} {\bibinfo {author} {\bibfnamefont {G.}~\bibnamefont
  {Aleksandrowicz}}, \bibinfo {author} {\bibfnamefont {T.}~\bibnamefont
  {Alexander}}, \bibinfo {author} {\bibfnamefont {P.}~\bibnamefont
  {Barkoutsos}}, \bibinfo {author} {\bibfnamefont {L.}~\bibnamefont {Bello}},
  \bibinfo {author} {\bibfnamefont {Y.}~\bibnamefont {Ben-Haim}}, \bibinfo
  {author} {\bibfnamefont {D.}~\bibnamefont {Bucher}},\ and\ \bibinfo {author}
  {\bibnamefont {et~al.}},\ }\href {https://doi.org/10.5281/zenodo.2562111}
  {\bibinfo {title} {Qiskit: An open-source framework for quantum computing}}
  (\bibinfo {year} {2019})\BibitemShut {NoStop}%
\bibitem [{\citenamefont {Iten}\ \emph {et~al.}(2019)\citenamefont {Iten},
  \citenamefont {Reardon-Smith}, \citenamefont {Malvetti}, \citenamefont
  {Mondada}, \citenamefont {Pauvert}, \citenamefont {Redmond}, \citenamefont
  {Kohli},\ and\ \citenamefont {Colbeck}}]{iten2019introduction}%
  \BibitemOpen
  \bibfield  {author} {\bibinfo {author} {\bibfnamefont {R.}~\bibnamefont
  {Iten}}, \bibinfo {author} {\bibfnamefont {O.}~\bibnamefont {Reardon-Smith}},
  \bibinfo {author} {\bibfnamefont {E.}~\bibnamefont {Malvetti}}, \bibinfo
  {author} {\bibfnamefont {L.}~\bibnamefont {Mondada}}, \bibinfo {author}
  {\bibfnamefont {G.}~\bibnamefont {Pauvert}}, \bibinfo {author} {\bibfnamefont
  {E.}~\bibnamefont {Redmond}}, \bibinfo {author} {\bibfnamefont {R.~S.}\
  \bibnamefont {Kohli}},\ and\ \bibinfo {author} {\bibfnamefont
  {R.}~\bibnamefont {Colbeck}},\ }\bibfield  {title} {\bibinfo {title}
  {Introduction to universalqcompiler},\ }\href
  {https://arxiv.org/abs/1904.01072} {\bibfield  {journal} {\bibinfo  {journal}
  {arXiv preprint arXiv:1904.01072}\ } (\bibinfo {year} {2019})}\BibitemShut
  {NoStop}%
\bibitem [{\citenamefont {Krol}\ \emph {et~al.}(2022)\citenamefont {Krol},
  \citenamefont {Sarkar}, \citenamefont {Ashraf}, \citenamefont {Al-Ars},\ and\
  \citenamefont {Bertels}}]{krol2022efficient}%
  \BibitemOpen
  \bibfield  {author} {\bibinfo {author} {\bibfnamefont {A.~M.}\ \bibnamefont
  {Krol}}, \bibinfo {author} {\bibfnamefont {A.}~\bibnamefont {Sarkar}},
  \bibinfo {author} {\bibfnamefont {I.}~\bibnamefont {Ashraf}}, \bibinfo
  {author} {\bibfnamefont {Z.}~\bibnamefont {Al-Ars}},\ and\ \bibinfo {author}
  {\bibfnamefont {K.}~\bibnamefont {Bertels}},\ }\bibfield  {title} {\bibinfo
  {title} {Efficient decomposition of unitary matrices in quantum circuit
  compilers},\ }\href {https://www.mdpi.com/2076-3417/12/2/759} {\bibfield
  {journal} {\bibinfo  {journal} {Applied Sciences}\ }\textbf {\bibinfo
  {volume} {12}},\ \bibinfo {pages} {759} (\bibinfo {year} {2022})}\BibitemShut
  {NoStop}%
\bibitem [{\citenamefont {Shende}\ \emph {et~al.}(2005)\citenamefont {Shende},
  \citenamefont {Bullock},\ and\ \citenamefont {Markov}}]{shende2005synthesis}%
  \BibitemOpen
  \bibfield  {author} {\bibinfo {author} {\bibfnamefont {V.~V.}\ \bibnamefont
  {Shende}}, \bibinfo {author} {\bibfnamefont {S.~S.}\ \bibnamefont
  {Bullock}},\ and\ \bibinfo {author} {\bibfnamefont {I.~L.}\ \bibnamefont
  {Markov}},\ }\bibfield  {title} {\bibinfo {title} {Synthesis of quantum logic
  circuits},\ }in\ \href {https://dl.acm.org/doi/abs/10.1145/1120725.1120847}
  {\emph {\bibinfo {booktitle} {Proceedings of the 2005 Asia and South Pacific
  Design Automation Conference}}}\ (\bibinfo {year} {2005})\ pp.\ \bibinfo
  {pages} {272--275}\BibitemShut {NoStop}%
\bibitem [{\citenamefont {M{\"o}tt{\"o}nen$^1$}\ and\ \citenamefont
  {Vartiainen}(2006)}]{mottonen12006decompositions}%
  \BibitemOpen
  \bibfield  {author} {\bibinfo {author} {\bibfnamefont {M.}~\bibnamefont
  {M{\"o}tt{\"o}nen$^1$}}\ and\ \bibinfo {author} {\bibfnamefont {J.~J.}\
  \bibnamefont {Vartiainen}},\ }\bibfield  {title} {\bibinfo {title}
  {Decompositions of general quantum gates},\ }\href
  {https://arxiv.org/pdf/quant-ph/0504100} {\bibfield  {journal} {\bibinfo
  {journal} {Trends in quantum computing research}\ ,\ \bibinfo {pages} {149}}
  (\bibinfo {year} {2006})}\BibitemShut {NoStop}%
\bibitem [{\citenamefont {Rakyta}\ and\ \citenamefont
  {Zimbor{\'a}s}(2022)}]{rakyta2022approaching}%
  \BibitemOpen
  \bibfield  {author} {\bibinfo {author} {\bibfnamefont {P.}~\bibnamefont
  {Rakyta}}\ and\ \bibinfo {author} {\bibfnamefont {Z.}~\bibnamefont
  {Zimbor{\'a}s}},\ }\bibfield  {title} {\bibinfo {title} {Approaching the
  theoretical limit in quantum gate decomposition},\ }\href
  {https://quantum-journal.org/papers/q-2022-05-11-710/} {\bibfield  {journal}
  {\bibinfo  {journal} {Quantum}\ }\textbf {\bibinfo {volume} {6}},\ \bibinfo
  {pages} {710} (\bibinfo {year} {2022})}\BibitemShut {NoStop}%
\bibitem [{\citenamefont {Shende}\ \emph
  {et~al.}(2004{\natexlab{a}})\citenamefont {Shende}, \citenamefont {Markov},\
  and\ \citenamefont {Bullock}}]{shende2004minimal}%
  \BibitemOpen
  \bibfield  {author} {\bibinfo {author} {\bibfnamefont {V.~V.}\ \bibnamefont
  {Shende}}, \bibinfo {author} {\bibfnamefont {I.~L.}\ \bibnamefont {Markov}},\
  and\ \bibinfo {author} {\bibfnamefont {S.~S.}\ \bibnamefont {Bullock}},\
  }\bibfield  {title} {\bibinfo {title} {Minimal universal two-qubit
  controlled-not-based circuits},\ }\href
  {https://journals.aps.org/pra/abstract/10.1103/PhysRevA.69.062321} {\bibfield
   {journal} {\bibinfo  {journal} {Physical Review A}\ }\textbf {\bibinfo
  {volume} {69}},\ \bibinfo {pages} {062321} (\bibinfo {year}
  {2004}{\natexlab{a}})}\BibitemShut {NoStop}%
\bibitem [{\citenamefont {Tucci}(2005)}]{tucci2005introduction}%
  \BibitemOpen
  \bibfield  {author} {\bibinfo {author} {\bibfnamefont {R.~R.}\ \bibnamefont
  {Tucci}},\ }\bibfield  {title} {\bibinfo {title} {An introduction to cartan's
  kak decomposition for qc programmers},\ }\href
  {https://arxiv.org/abs/quant-ph/0507171} {\bibfield  {journal} {\bibinfo
  {journal} {arXiv preprint quant-ph/0507171}\ } (\bibinfo {year}
  {2005})}\BibitemShut {NoStop}%
\bibitem [{\citenamefont {Shende}\ \emph
  {et~al.}(2004{\natexlab{b}})\citenamefont {Shende}, \citenamefont {Bullock},\
  and\ \citenamefont {Markov}}]{shende2004recognizing}%
  \BibitemOpen
  \bibfield  {author} {\bibinfo {author} {\bibfnamefont {V.~V.}\ \bibnamefont
  {Shende}}, \bibinfo {author} {\bibfnamefont {S.~S.}\ \bibnamefont
  {Bullock}},\ and\ \bibinfo {author} {\bibfnamefont {I.~L.}\ \bibnamefont
  {Markov}},\ }\bibfield  {title} {\bibinfo {title} {Recognizing small-circuit
  structure in two-qubit operators},\ }\href
  {https://journals.aps.org/pra/abstract/10.1103/PhysRevA.70.012310} {\bibfield
   {journal} {\bibinfo  {journal} {Physical Review A}\ }\textbf {\bibinfo
  {volume} {70}},\ \bibinfo {pages} {012310} (\bibinfo {year}
  {2004}{\natexlab{b}})}\BibitemShut {NoStop}%
\bibitem [{\citenamefont {Vidal}\ and\ \citenamefont
  {Dawson}(2004)}]{vidal2004universal}%
  \BibitemOpen
  \bibfield  {author} {\bibinfo {author} {\bibfnamefont {G.}~\bibnamefont
  {Vidal}}\ and\ \bibinfo {author} {\bibfnamefont {C.~M.}\ \bibnamefont
  {Dawson}},\ }\bibfield  {title} {\bibinfo {title} {Universal quantum circuit
  for two-qubit transformations with three controlled-not gates},\ }\href
  {https://journals.aps.org/pra/abstract/10.1103/PhysRevA.69.010301} {\bibfield
   {journal} {\bibinfo  {journal} {Physical Review A}\ }\textbf {\bibinfo
  {volume} {69}},\ \bibinfo {pages} {010301} (\bibinfo {year}
  {2004})}\BibitemShut {NoStop}%
\bibitem [{\citenamefont {Barenco}\ \emph {et~al.}(1995)\citenamefont
  {Barenco}, \citenamefont {Bennett}, \citenamefont {Cleve}, \citenamefont
  {DiVincenzo}, \citenamefont {Margolus}, \citenamefont {Shor}, \citenamefont
  {Sleator}, \citenamefont {Smolin},\ and\ \citenamefont
  {Weinfurter}}]{barenco1995elementary}%
  \BibitemOpen
  \bibfield  {author} {\bibinfo {author} {\bibfnamefont {A.}~\bibnamefont
  {Barenco}}, \bibinfo {author} {\bibfnamefont {C.~H.}\ \bibnamefont
  {Bennett}}, \bibinfo {author} {\bibfnamefont {R.}~\bibnamefont {Cleve}},
  \bibinfo {author} {\bibfnamefont {D.~P.}\ \bibnamefont {DiVincenzo}},
  \bibinfo {author} {\bibfnamefont {N.}~\bibnamefont {Margolus}}, \bibinfo
  {author} {\bibfnamefont {P.}~\bibnamefont {Shor}}, \bibinfo {author}
  {\bibfnamefont {T.}~\bibnamefont {Sleator}}, \bibinfo {author} {\bibfnamefont
  {J.~A.}\ \bibnamefont {Smolin}},\ and\ \bibinfo {author} {\bibfnamefont
  {H.}~\bibnamefont {Weinfurter}},\ }\bibfield  {title} {\bibinfo {title}
  {Elementary gates for quantum computation},\ }\href
  {https://journals.aps.org/pra/abstract/10.1103/PhysRevA.52.3457} {\bibfield
  {journal} {\bibinfo  {journal} {Physical review A}\ }\textbf {\bibinfo
  {volume} {52}},\ \bibinfo {pages} {3457} (\bibinfo {year}
  {1995})}\BibitemShut {NoStop}%
\bibitem [{\citenamefont {Jiang}\ \emph {et~al.}(2020)\citenamefont {Jiang},
  \citenamefont {Kalev}, \citenamefont {Mruczkiewicz},\ and\ \citenamefont
  {Neven}}]{jiang2020optimal}%
  \BibitemOpen
  \bibfield  {author} {\bibinfo {author} {\bibfnamefont {Z.}~\bibnamefont
  {Jiang}}, \bibinfo {author} {\bibfnamefont {A.}~\bibnamefont {Kalev}},
  \bibinfo {author} {\bibfnamefont {W.}~\bibnamefont {Mruczkiewicz}},\ and\
  \bibinfo {author} {\bibfnamefont {H.}~\bibnamefont {Neven}},\ }\bibfield
  {title} {\bibinfo {title} {Optimal fermion-to-qubit mapping via ternary trees
  with applications to reduced quantum states learning},\ }\href
  {https://quantum-journal.org/papers/q-2020-06-04-276/} {\bibfield  {journal}
  {\bibinfo  {journal} {Quantum}\ }\textbf {\bibinfo {volume} {4}},\ \bibinfo
  {pages} {276} (\bibinfo {year} {2020})}\BibitemShut {NoStop}%
\bibitem [{\citenamefont {Galvis-Florez}\ \emph {et~al.}(2023)\citenamefont
  {Galvis-Florez}, \citenamefont {Reitzner},\ and\ \citenamefont
  {S{\"a}rkk{\"a}}}]{galvis2023single}%
  \BibitemOpen
  \bibfield  {author} {\bibinfo {author} {\bibfnamefont {C.~A.}\ \bibnamefont
  {Galvis-Florez}}, \bibinfo {author} {\bibfnamefont {D.}~\bibnamefont
  {Reitzner}},\ and\ \bibinfo {author} {\bibfnamefont {S.}~\bibnamefont
  {S{\"a}rkk{\"a}}},\ }\bibfield  {title} {\bibinfo {title} {Single qubit state
  estimation on nisq devices with limited resources and sic-povms},\ }in\ \href
  {https://ieeexplore.ieee.org/abstract/document/10313897/} {\emph {\bibinfo
  {booktitle} {2023 IEEE International Conference on Quantum Computing and
  Engineering (QCE)}}},\ Vol.~\bibinfo {volume} {1}\ (\bibinfo {organization}
  {IEEE},\ \bibinfo {year} {2023})\ pp.\ \bibinfo {pages}
  {111--119}\BibitemShut {NoStop}%
\bibitem [{\citenamefont {Feng}\ and\ \citenamefont
  {Luo}(2022)}]{feng2022stabilizer}%
  \BibitemOpen
  \bibfield  {author} {\bibinfo {author} {\bibfnamefont {L.}~\bibnamefont
  {Feng}}\ and\ \bibinfo {author} {\bibfnamefont {S.}~\bibnamefont {Luo}},\
  }\bibfield  {title} {\bibinfo {title} {From stabilizer states to sic-povm
  fiducial states},\ }\href
  {https://link.springer.com/article/10.1134/S004057792212008X} {\bibfield
  {journal} {\bibinfo  {journal} {Theoretical and Mathematical Physics}\
  }\textbf {\bibinfo {volume} {213}},\ \bibinfo {pages} {1747} (\bibinfo {year}
  {2022})}\BibitemShut {NoStop}%
\bibitem [{\citenamefont {Saraceno}\ \emph {et~al.}(2017)\citenamefont
  {Saraceno}, \citenamefont {Ermann},\ and\ \citenamefont
  {Cormick}}]{saraceno2017phase}%
  \BibitemOpen
  \bibfield  {author} {\bibinfo {author} {\bibfnamefont {M.}~\bibnamefont
  {Saraceno}}, \bibinfo {author} {\bibfnamefont {L.}~\bibnamefont {Ermann}},\
  and\ \bibinfo {author} {\bibfnamefont {C.}~\bibnamefont {Cormick}},\
  }\bibfield  {title} {\bibinfo {title} {Phase-space representations of
  symmetric informationally complete positive-operator-valued-measure fiducial
  states},\ }\href
  {https://journals.aps.org/pra/abstract/10.1103/PhysRevA.95.032102} {\bibfield
   {journal} {\bibinfo  {journal} {Physical Review A}\ }\textbf {\bibinfo
  {volume} {95}},\ \bibinfo {pages} {032102} (\bibinfo {year}
  {2017})}\BibitemShut {NoStop}%
\bibitem [{\citenamefont {Liu}\ \emph {et~al.}(2022{\natexlab{b}})\citenamefont
  {Liu}, \citenamefont {Zhang},\ and\ \citenamefont {Ma}}]{liu2022classically}%
  \BibitemOpen
  \bibfield  {author} {\bibinfo {author} {\bibfnamefont {G.}~\bibnamefont
  {Liu}}, \bibinfo {author} {\bibfnamefont {X.}~\bibnamefont {Zhang}},\ and\
  \bibinfo {author} {\bibfnamefont {X.}~\bibnamefont {Ma}},\ }\bibfield
  {title} {\bibinfo {title} {Classically replaceable operations},\ }\href
  {https://quantum-journal.org/papers/q-2022-10-24-845/} {\bibfield  {journal}
  {\bibinfo  {journal} {Quantum}\ }\textbf {\bibinfo {volume} {6}},\ \bibinfo
  {pages} {845} (\bibinfo {year} {2022}{\natexlab{b}})}\BibitemShut {NoStop}%
\bibitem [{\citenamefont {Oszmaniec}\ \emph {et~al.}(2019)\citenamefont
  {Oszmaniec}, \citenamefont {Maciejewski},\ and\ \citenamefont
  {Pucha{\l}a}}]{oszmaniec2019simulating}%
  \BibitemOpen
  \bibfield  {author} {\bibinfo {author} {\bibfnamefont {M.}~\bibnamefont
  {Oszmaniec}}, \bibinfo {author} {\bibfnamefont {F.~B.}\ \bibnamefont
  {Maciejewski}},\ and\ \bibinfo {author} {\bibfnamefont {Z.}~\bibnamefont
  {Pucha{\l}a}},\ }\bibfield  {title} {\bibinfo {title} {Simulating all quantum
  measurements using only projective measurements and postselection},\ }\href
  {https://journals.aps.org/pra/abstract/10.1103/PhysRevA.100.012351}
  {\bibfield  {journal} {\bibinfo  {journal} {Physical Review A}\ }\textbf
  {\bibinfo {volume} {100}},\ \bibinfo {pages} {012351} (\bibinfo {year}
  {2019})}\BibitemShut {NoStop}%
\bibitem [{\citenamefont {Rost}\ \emph {et~al.}(2020)\citenamefont {Rost},
  \citenamefont {Jones}, \citenamefont {Vyushkova}, \citenamefont {Ali},
  \citenamefont {Cullip}, \citenamefont {Vyushkov},\ and\ \citenamefont
  {Nabrzyski}}]{rost2020simulation}%
  \BibitemOpen
  \bibfield  {author} {\bibinfo {author} {\bibfnamefont {B.}~\bibnamefont
  {Rost}}, \bibinfo {author} {\bibfnamefont {B.}~\bibnamefont {Jones}},
  \bibinfo {author} {\bibfnamefont {M.}~\bibnamefont {Vyushkova}}, \bibinfo
  {author} {\bibfnamefont {A.}~\bibnamefont {Ali}}, \bibinfo {author}
  {\bibfnamefont {C.}~\bibnamefont {Cullip}}, \bibinfo {author} {\bibfnamefont
  {A.}~\bibnamefont {Vyushkov}},\ and\ \bibinfo {author} {\bibfnamefont
  {J.}~\bibnamefont {Nabrzyski}},\ }\bibfield  {title} {\bibinfo {title}
  {Simulation of thermal relaxation in spin chemistry systems on a quantum
  computer using inherent qubit decoherence},\ }\href
  {https://arxiv.org/abs/2001.00794} {\bibfield  {journal} {\bibinfo  {journal}
  {arXiv preprint arXiv:2001.00794}\ } (\bibinfo {year} {2020})}\BibitemShut
  {NoStop}%
\bibitem [{\citenamefont {Khaneja}\ and\ \citenamefont
  {Glaser}(2001)}]{khaneja2001cartan}%
  \BibitemOpen
  \bibfield  {author} {\bibinfo {author} {\bibfnamefont {N.}~\bibnamefont
  {Khaneja}}\ and\ \bibinfo {author} {\bibfnamefont {S.~J.}\ \bibnamefont
  {Glaser}},\ }\bibfield  {title} {\bibinfo {title} {Cartan decomposition of su
  (2n) and control of spin systems},\ }\href
  {https://www.sciencedirect.com/science/article/abs/pii/S0301010401003184}
  {\bibfield  {journal} {\bibinfo  {journal} {Chemical Physics}\ }\textbf
  {\bibinfo {volume} {267}},\ \bibinfo {pages} {11} (\bibinfo {year}
  {2001})}\BibitemShut {NoStop}%
\bibitem [{\citenamefont {Mansky}\ \emph {et~al.}(2023)\citenamefont {Mansky},
  \citenamefont {Castillo}, \citenamefont {Puigvert},\ and\ \citenamefont
  {Linnhoff-Popien}}]{mansky2023near}%
  \BibitemOpen
  \bibfield  {author} {\bibinfo {author} {\bibfnamefont {M.~B.}\ \bibnamefont
  {Mansky}}, \bibinfo {author} {\bibfnamefont {S.~L.}\ \bibnamefont
  {Castillo}}, \bibinfo {author} {\bibfnamefont {V.~R.}\ \bibnamefont
  {Puigvert}},\ and\ \bibinfo {author} {\bibfnamefont {C.}~\bibnamefont
  {Linnhoff-Popien}},\ }\bibfield  {title} {\bibinfo {title} {Near-optimal
  quantum circuit construction via cartan decomposition},\ }\href
  {https://journals.aps.org/pra/abstract/10.1103/PhysRevA.108.052607}
  {\bibfield  {journal} {\bibinfo  {journal} {Physical Review A}\ }\textbf
  {\bibinfo {volume} {108}},\ \bibinfo {pages} {052607} (\bibinfo {year}
  {2023})}\BibitemShut {NoStop}%
\bibitem [{\citenamefont {Conlon}\ \emph
  {et~al.}(2023{\natexlab{a}})\citenamefont {Conlon}, \citenamefont {Vogl},
  \citenamefont {Marciniak}, \citenamefont {Pogorelov}, \citenamefont {Yung},
  \citenamefont {Eilenberger}, \citenamefont {Berry}, \citenamefont {Santana},
  \citenamefont {Blatt}, \citenamefont {Monz} \emph
  {et~al.}}]{conlon2023approaching}%
  \BibitemOpen
  \bibfield  {author} {\bibinfo {author} {\bibfnamefont {L.~O.}\ \bibnamefont
  {Conlon}}, \bibinfo {author} {\bibfnamefont {T.}~\bibnamefont {Vogl}},
  \bibinfo {author} {\bibfnamefont {C.~D.}\ \bibnamefont {Marciniak}}, \bibinfo
  {author} {\bibfnamefont {I.}~\bibnamefont {Pogorelov}}, \bibinfo {author}
  {\bibfnamefont {S.~K.}\ \bibnamefont {Yung}}, \bibinfo {author}
  {\bibfnamefont {F.}~\bibnamefont {Eilenberger}}, \bibinfo {author}
  {\bibfnamefont {D.~W.}\ \bibnamefont {Berry}}, \bibinfo {author}
  {\bibfnamefont {F.~S.}\ \bibnamefont {Santana}}, \bibinfo {author}
  {\bibfnamefont {R.}~\bibnamefont {Blatt}}, \bibinfo {author} {\bibfnamefont
  {T.}~\bibnamefont {Monz}}, \emph {et~al.},\ }\bibfield  {title} {\bibinfo
  {title} {Approaching optimal entangling collective measurements on quantum
  computing platforms},\ }\href
  {https://www.nature.com/articles/s41567-022-01875-7} {\bibfield  {journal}
  {\bibinfo  {journal} {Nature Physics}\ }\textbf {\bibinfo {volume} {19}},\
  \bibinfo {pages} {351} (\bibinfo {year} {2023}{\natexlab{a}})}\BibitemShut
  {NoStop}%
\bibitem [{\citenamefont {Conlon}\ \emph
  {et~al.}(2023{\natexlab{b}})\citenamefont {Conlon}, \citenamefont
  {Eilenberger}, \citenamefont {Lam},\ and\ \citenamefont
  {Assad}}]{conlon2023discriminating}%
  \BibitemOpen
  \bibfield  {author} {\bibinfo {author} {\bibfnamefont {L.~O.}\ \bibnamefont
  {Conlon}}, \bibinfo {author} {\bibfnamefont {F.}~\bibnamefont {Eilenberger}},
  \bibinfo {author} {\bibfnamefont {P.~K.}\ \bibnamefont {Lam}},\ and\ \bibinfo
  {author} {\bibfnamefont {S.~M.}\ \bibnamefont {Assad}},\ }\bibfield  {title}
  {\bibinfo {title} {Discriminating mixed qubit states with collective
  measurements},\ }\href {https://www.nature.com/articles/s42005-023-01454-z}
  {\bibfield  {journal} {\bibinfo  {journal} {Communications Physics}\ }\textbf
  {\bibinfo {volume} {6}},\ \bibinfo {pages} {337} (\bibinfo {year}
  {2023}{\natexlab{b}})}\BibitemShut {NoStop}%
\bibitem [{\citenamefont {Endo}\ \emph {et~al.}(2018)\citenamefont {Endo},
  \citenamefont {Benjamin},\ and\ \citenamefont {Li}}]{endo2018practical}%
  \BibitemOpen
  \bibfield  {author} {\bibinfo {author} {\bibfnamefont {S.}~\bibnamefont
  {Endo}}, \bibinfo {author} {\bibfnamefont {S.~C.}\ \bibnamefont {Benjamin}},\
  and\ \bibinfo {author} {\bibfnamefont {Y.}~\bibnamefont {Li}},\ }\bibfield
  {title} {\bibinfo {title} {Practical quantum error mitigation for near-future
  applications},\ }\href
  {https://journals.aps.org/prx/abstract/10.1103/PhysRevX.8.031027} {\bibfield
  {journal} {\bibinfo  {journal} {Physical Review X}\ }\textbf {\bibinfo
  {volume} {8}},\ \bibinfo {pages} {031027} (\bibinfo {year}
  {2018})}\BibitemShut {NoStop}%
\bibitem [{\citenamefont {Glos}\ \emph {et~al.}(2022)\citenamefont {Glos},
  \citenamefont {Nyk{\"a}nen}, \citenamefont {Borrelli}, \citenamefont
  {Maniscalco}, \citenamefont {Rossi}, \citenamefont {Zimbor{\'a}s},\ and\
  \citenamefont {Garc{\'\i}a-P{\'e}rez}}]{glos2022adaptive}%
  \BibitemOpen
  \bibfield  {author} {\bibinfo {author} {\bibfnamefont {A.}~\bibnamefont
  {Glos}}, \bibinfo {author} {\bibfnamefont {A.}~\bibnamefont {Nyk{\"a}nen}},
  \bibinfo {author} {\bibfnamefont {E.-M.}\ \bibnamefont {Borrelli}}, \bibinfo
  {author} {\bibfnamefont {S.}~\bibnamefont {Maniscalco}}, \bibinfo {author}
  {\bibfnamefont {M.~A.}\ \bibnamefont {Rossi}}, \bibinfo {author}
  {\bibfnamefont {Z.}~\bibnamefont {Zimbor{\'a}s}},\ and\ \bibinfo {author}
  {\bibfnamefont {G.}~\bibnamefont {Garc{\'\i}a-P{\'e}rez}},\ }\bibfield
  {title} {\bibinfo {title} {Adaptive povm implementations and measurement
  error mitigation strategies for near-term quantum devices},\ }\href
  {https://arxiv.org/abs/2208.07817} {\bibfield  {journal} {\bibinfo  {journal}
  {arXiv preprint arXiv:2208.07817}\ } (\bibinfo {year} {2022})}\BibitemShut
  {NoStop}%
\bibitem [{\citenamefont {Wood}\ \emph {et~al.}(2011)\citenamefont {Wood},
  \citenamefont {Biamonte},\ and\ \citenamefont {Cory}}]{wood2011tensor}%
  \BibitemOpen
  \bibfield  {author} {\bibinfo {author} {\bibfnamefont {C.~J.}\ \bibnamefont
  {Wood}}, \bibinfo {author} {\bibfnamefont {J.~D.}\ \bibnamefont {Biamonte}},\
  and\ \bibinfo {author} {\bibfnamefont {D.~G.}\ \bibnamefont {Cory}},\
  }\bibfield  {title} {\bibinfo {title} {Tensor networks and graphical calculus
  for open quantum systems},\ }\href {https://arxiv.org/abs/1111.6950}
  {\bibfield  {journal} {\bibinfo  {journal} {arXiv preprint arXiv:1111.6950}\
  } (\bibinfo {year} {2011})}\BibitemShut {NoStop}%
\end{thebibliography}%
%

\clearpage
\onecolumngrid
\newpage

\begin{appendix}
\section*{Appendix Contents}
In this Appendix, we provide detailed explanations, additional discussions, and algorithms to supplement the main text.
In App.~\ref{ap:POVM_Unitary}, we use tensor network representations to connect the POVM operators with the unitary $U_{\mathrm{P}}$ in the dilation framework of Eq.~\eqref{eq:povm_process}. App.~\ref{ap:determinant} offers a derivation of the determinant equations in Eq.~\eqref{eq:chi_simpl}. Moving on to App.~\ref{ap:parameterize}, we give a detailed explanation for the classification and representation of parameters $\Theta$ and $\Gamma$ for both $U_{\mathrm{P}}$ and $U_{\mathrm{SIC}}$, including a description of the relabeling operation $U_{\mathrm{rel}}$, as mentioned in Sec.~\ref{sec:reduce}. While in App.~\ref{ap:sic_bell}, we again use the tensor network representation to discuss the relationship between the Bell measurement and SIC-POVM. Finally, App.~\ref{ap:algo1} elaborates on Algo.~\ref{algo:algo1} in Sec.~\ref{sec:prac} and further introduces Algo.~\ref{algo:algo2}, which determines the remaining parameters necessary for implementing the general circuit shown in Fig.~\ref{fig:prac_circ}(a).

\section{Unitary for Rank-1 POVMs}\label{ap:POVM_Unitary}
The relationship between rank-1 POVMs and quantum circuits is established in the main text through the model in Eq.~\eqref{eq:povm_process}. We first employ a tensor network representation \cite{wood2011tensor} to depict the measurement process as
\begin{equation}\label{eq:ancilla_measure_tensor}        
\includegraphics{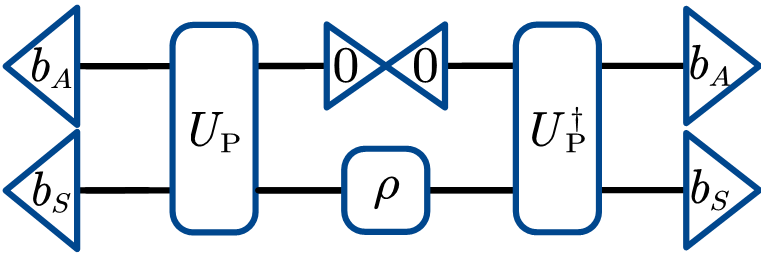},
\end{equation}
where $b_{A}$ and $b_{S}$  represent the possible outcomes of measurement on the auxiliary qubit and $\rho$ respectively on the computational basis, namely the set $\{0,1\}$. It should be noted that the tensor network is read in the same order as the formula but in the reverse order of the circuit form. To delve further into the specific structure of $U_{\mathrm{P}}$, we expand $U_{\mathrm{P}}$ into its individual elements, that is
\begin{equation}        
\includegraphics{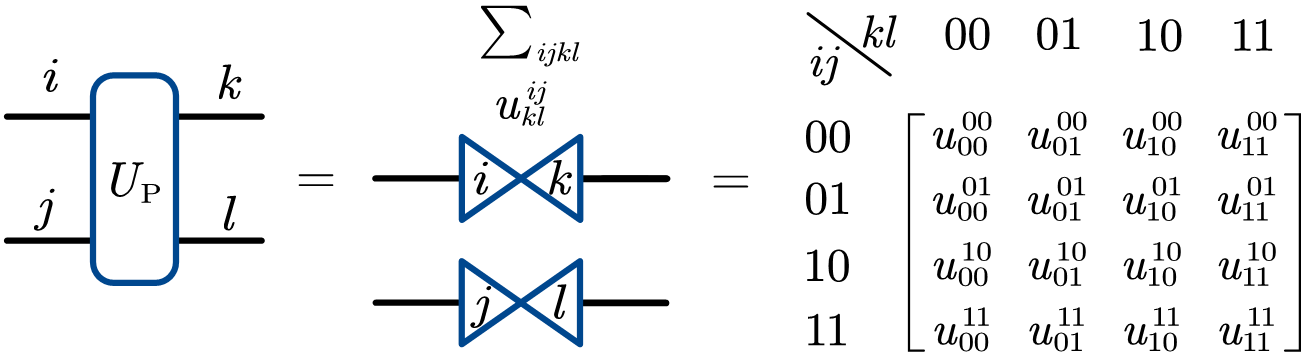}.
\end{equation}
Here, $i,j$ denote the row indices of $U_{\mathrm{P}}$, while $k,l$ pertain to its column indices. Each $u_{kl}^{ij}$ depicts the specific element of $U_{\mathrm{P}}$ at the corresponding row and column indices. Integrating this representation into Eq.~\eqref{eq:ancilla_measure_tensor}, we have
\begin{equation}        
\includegraphics{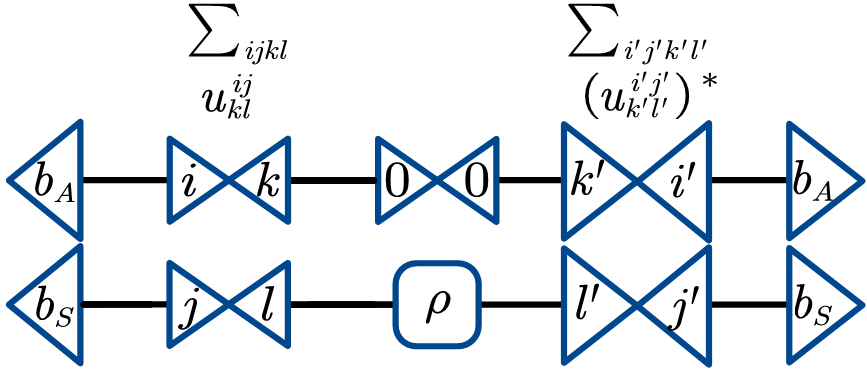}.
\end{equation}

Since $|b_A\rangle$, $|b_S\rangle$, along with $|i\rangle$, $|j\rangle$, $|k\rangle$, $|l\rangle$ are all elements of the computational basis, the inner product equals 1 only when indices match, otherwise, it is 0. Omitting all the items that equal to $0$, we have
\begin{equation}        
\includegraphics{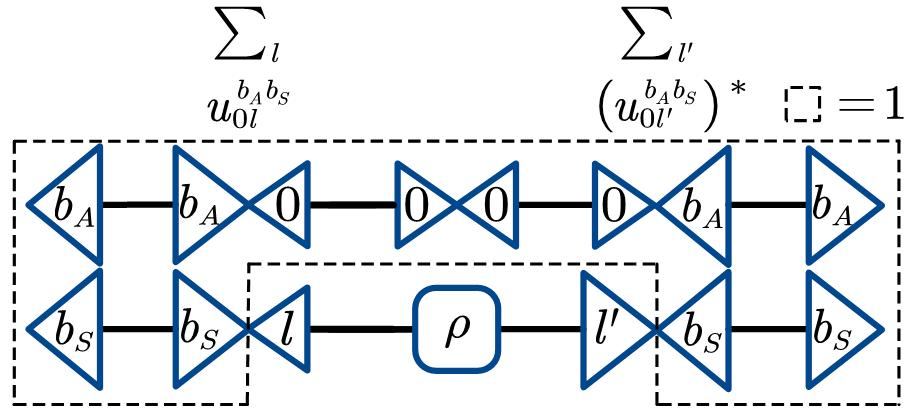}.
\end{equation}
Then abbreviate the expression $\sum\nolimits_l^{}{u_{0l}^{b_Ab_S}}\langle l|$ to $\langle \varphi _{b_Ab_S}|$, thus allowing further simplification of the tensor network, that is
\begin{equation}        
\includegraphics{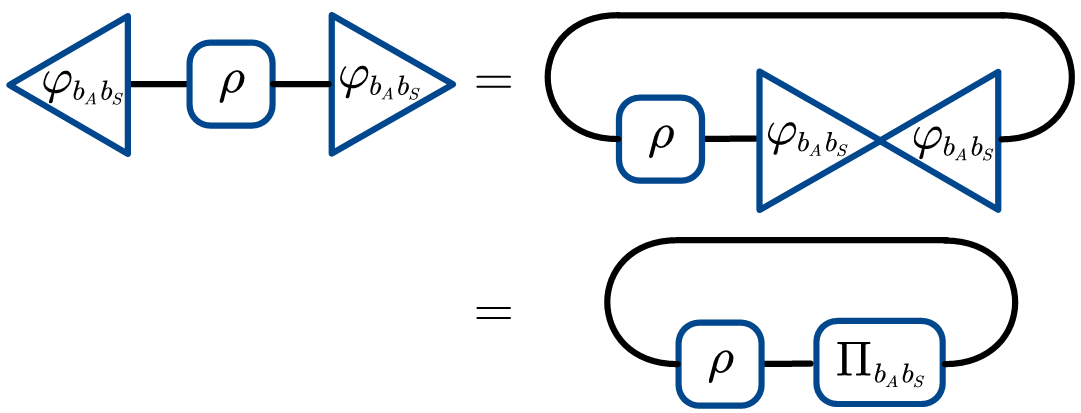}.
\end{equation}

It becomes evident that a POVM is determined solely by the elements where $k=0$ in $U_{\mathrm{P}}$, that is to say, the first two columns of $U_{\mathrm{P}}$. The specific elements thereof are labeled via the measurement outcomes $b_{A}b_{S}$, which conveys that if the measurement result is $b_{A}b_{S}$, it is tantamount to having performed a measurement with the POVM element constituted by the elements in the row $b_{A}b_{S}$ of $U_{\mathrm{P}}$.

\section{1-CNOT Determinant}\label{ap:determinant}
For $U \in U(4)$, the $\chi \left[ \gamma \left( U \right) \right]$ can be expanded as follows
\begin{align}\label{eq:chi}
\begin{aligned}
\chi \left[ \gamma \right] &=\det \left( xI-\gamma \right) 
\\
&=\left( x-\lambda _1 \right) \left( x-\lambda _2 \right) \left( x-\lambda _3 \right) \left( x-\lambda _4 \right) 
\\
&=x^4-\sum_i{\lambda _ix^3}+\sum_{i<j}{\lambda _i\lambda _jx^2}-\sum_{i<j<k}{\lambda _i\lambda _j\lambda _kx}+\prod_i{\lambda _i},
\end{aligned}
\end{align}
where $\gamma$ denotes $\gamma\left( U \right)=U\sigma_{2}^{\otimes 2}U^{T}\sigma_{2}^{\otimes 2}$, and $\lambda_{i}$ are the eigenvalues of $\gamma$. Given that $\gamma$ is a unitary, each $\lambda _i$ should be of the form $\lambda _i=e^{\mathrm{i}\alpha _i}$. Since $\prod_i{\lambda _i}=\det \left( \gamma \right) $, it follows that $\prod_{i\ne j}{\lambda _i}=\det \left( \gamma \right) \lambda _{j}^{-1}=\det \left( \gamma \right) \lambda _{j}^{*}$, leading to $\sum_{i<j<k}{\lambda _i\lambda _j\lambda _k}=\det \left( \gamma \right) \sum_i{\lambda _{i}^{*}}$. Since $\sum_i{\lambda _i}=\mathrm{tr}\left( \gamma \right) $, we have $\sum_i{\lambda _{i}^{2}}=\mathrm{tr}\left( \gamma^2 \right) $, $\sum_i{\lambda _{i}^{*}=\mathrm{tr}^*\left( \gamma \right)}$ and $\sum_{i<j}{\lambda _i\lambda _j}=\frac{1}{2}\left( \sum_i{\lambda _i}\sum_j{\lambda _j}-\sum_i{\lambda _{i}^{2}} \right) =\frac{1}{2}\left( \mathrm{tr}^2\left( \gamma \right) -\mathrm{tr}\left( \gamma ^2 \right) \right) $. And for $U \in SU(4)$, where $\det(U) = 1$, this simplifies to
\begin{align}\label{eq:chi_simp}
\begin{aligned}
\chi \left[ \gamma \right] =x^4-\mathrm{tr}\left( \gamma \right) x^3+\frac{1}{2}\left( \mathrm{tr}^2\left( \gamma \right) -\mathrm{tr}\left( \gamma ^2 \right) \right) x^2-\mathrm{tr}^*\left( \gamma \right) +1.
\end{aligned}
\end{align}

Comparing with the conditions for the number of CNOT gates as mentioned in Proposition \uppercase\expandafter{\romannumeral3}.2 in Ref.~\cite{shende2004recognizing}, we derive the discriminant for the 1 CNOT case where $U \in SU(4)$, that is

\begin{align}\label{eq:chi_simp}
\begin{aligned}
\begin{cases}
	\mathrm{tr}\left( \gamma \right) =0,\\
	\mathrm{tr}\left( \gamma ^2 \right) =-4.\\
\end{cases}
\end{aligned}
\end{align}
For $U \in U(4)$, the second condition modifies to $\mathrm{tr}\left( \gamma ^2 \right) =-4\det(U)$.

\section{Parameterization for $U_{\mathrm{P}}$}\label{ap:parameterize}
In Ref.~\cite{garcia2021learning}, the authors indicate that the measurement of POVM is defined by eight specific parameters. The corresponding elements in $U_{\mathrm{P}}$ are determined through Gram-Schmidt orthogonalization and high-dimensional spherical mapping. However, the orthonormal basis derived from Gram-Schmidt orthogonalization is not uniquely determined. Consequently, variations in the orthonormal basis can lead to different POVMs for the same set of parameters. Moreover, to minimize the CNOT count, it is essential to analyze parameters that do not affect the measurement results of the POVM. This necessitates the introduction of a refined parameterization approach for $U_{\mathrm{P}}$ in this section.

\subsection{Parameters in $\Theta$} 

Firstly, we identify parameters that do not influence the POVM. Due to the property $\mathrm{tr}\left( \rho \Pi _i \right) =\mathrm{tr}\left( \rho |\varphi _i\rangle \langle \varphi _i| \right) =\mathrm{tr}\left( \rho |\varphi _i\rangle e^{\mathrm{i}\alpha}e^{-\mathrm{i}\alpha}\langle \varphi _i| \right) $, it is evident that the phase of each row in $U_{\mathrm{P}}$ does not influence the measurement results. Consequently, we can adjust 4 parameters to ensure that the first column of $U_{\mathrm{P}}$ in Fig.~\ref{fig:Up_detail} is real. We can express this adjustment as 
\begin{equation}\label{eq:ab_phase}
U_{\mathrm{P}}(\Gamma,\Theta)=U_{\mathrm{ph}}U_{\mathrm{P}}(\Gamma,\Theta_{0}),
\end{equation}
where $U_{\mathrm{P}}(\Gamma,\Theta_0)$ serves as a starting point for $\Theta$ variation, and $\Theta_0$ is a set of fixed parameters in $U_{\mathrm{P}}(\Gamma,\Theta_0)$. Here, $U_{\mathrm{ph}}=\mathrm{diag}\left( \left[ \begin{matrix}
	e^{\mathrm{i}\alpha _0}&		e^{\mathrm{i}\alpha _1}&		e^{\mathrm{i}\alpha _2}&		e^{\mathrm{i}\alpha _3}\\
\end{matrix} \right] \right) $ is a diagonal matrix. We can represent these parameters in the form of a circuit using $R_{z}$ and $C_{R_{z}}$ gate:
\begin{equation}
R_{z}\left( \beta \right) =\left[ \begin{matrix}
	1&		0\\
	0&		e^{\mathrm{i}\beta}\\
\end{matrix} \right] ,\ \ C_{R_{z}}\left( \beta \right) =\left[ \begin{matrix}
	1&		0&		0&		0\\
	0&		1&		0&		0\\
	0&		0&		1&		0\\
	0&		0&		0&		e^{\mathrm{i}\beta}\\
\end{matrix} \right] .
\end{equation}
Subsequently, one has $U_{\mathrm{ph}}=e^{\mathrm{i}\beta _0}C_{R_{z}}(\beta _3)(R_{z}(\beta _1)\otimes R_{z}(\beta _2))$ with $\beta _0 = \alpha_0, \beta_1 = \alpha_2 -\alpha_0, \beta_2=\alpha_1-\alpha_0, \beta_3= \alpha_0+\alpha_3-\alpha_1-\alpha_2$.

Based on the relationship between POVM and $U_{\mathrm{P}}$ detailed in App.~\ref{ap:POVM_Unitary}, which indicates that the last two columns are unrelated to the measurement process, we can succinctly represent $U_{\mathrm{P}}(\Gamma,\Theta_0)$ as 
\begin{align}
\begin{aligned}
    U_{\mathrm{P}}(\Gamma,\Theta_0)=\left[ \begin{matrix}
|\psi _1\rangle&	|\psi _2\rangle&	|\psi _3\rangle& |\psi _4\rangle\\
\end{matrix} \right] .
\end{aligned}
\end{align}
Given that $|\psi _1\rangle$ and $|\psi _2\rangle$ form an orthonormal basis spanning a subspace $\mathcal{H}_m$, while $|\psi _3\rangle$ and $|\psi _4\rangle$ span its orthogonal complement $\mathcal{H}_m^{\bot}$, any linear transformation $Q$ applied to $\mathcal{H}_m^{\bot}$ must preserve the orthogonality between the two subspaces. Consequently, $Q$ does not affect the first two columns of $U_{\mathrm{P}}$, which correspond to the POVM elements. To ensure that $U_{\mathrm{P}}$ remains unitary after $Q$ is applied, $Q$ itself must be a unitary transformation parameterized by 4 parameters. We can apply $Q$ to $\mathcal{H}_m^{\bot}$ as 
\begin{align}\label{eq:ab_q}
\begin{aligned}
    U_{\mathrm{P}}(\Gamma,\Theta)&=\left[ \begin{matrix}
|\psi _1\rangle&	|\psi _2\rangle&	|\psi _3\rangle& |\psi _4\rangle\\
\end{matrix} \right]  \left[ \begin{matrix}
	I&		0\\
	0&		Q\\
\end{matrix} \right] \\
&=U_{\mathrm{P}}(\Gamma,\Theta_0) C_{Q}.
\end{aligned}
\end{align}
The overall operation effectively constitutes a controlled unitary gate, denoted as $C_{Q}$. Combining Eq.~\eqref{eq:ab_phase} and Eq.~\eqref{eq:ab_q}, we obtain 
\begin{align}
\begin{aligned}
    U_{\mathrm{P}}(\Gamma,\Theta)=U_{\mathrm{ph}}U_{\mathrm{P}}(\Gamma,\Theta_0) C_{Q},
\end{aligned}
\end{align}
which is able to adjust all eight free parameters that do not influence the measurement outcome of the POVM.

\subsection{Parameters in $\Gamma$}\label{ap:para_psi}
For the remaining 8 parameters that are directly related to the POVM, recalling Eqs.~\eqref{eq:povm_conditions} and \eqref{eq:V} in the main text, we have 
\begin{align}\label{eq:rank1_completeness}
\begin{aligned}
V^{\dagger}V=\sum_i{|\varphi _i\rangle \langle \varphi _i|}=I .
\end{aligned}
\end{align}
Here, $|\varphi _i\rangle$ can be expressed as $|\varphi _i\rangle =a_i\left( \cos \frac{\theta _i}{2}|0\rangle +e^{-\mathrm{i}\phi _i}\sin \frac{\theta _i}{2}|1\rangle \right) $, where $a_i \in \left[0, 1\right]$, $\theta_i \in \left[0, \pi\right]$, and $\phi_i \in \left[0, 2\pi\right)$. Substituting this form into Eq.~\eqref{eq:rank1_completeness}, we can derive and simplify the constraints for $a_i$, $\theta_i$ and $\phi_i$, that is
\begin{align}\label{eq:basic_constrain}
\begin{aligned}
\begin{cases}
	\sum_i{a_{i}^{2}\cos ^2\frac{\theta _i}{2}}=1,\\
	\sum_i{a_{i}^{2}\sin ^2\frac{\theta _i}{2}}=1,\\
	\sum_i{a_{i}^{2}e^{\mathrm{i}\phi _1}\cos \frac{\theta _i}{2}\sin \frac{\theta _i}{2}}=0,\\
	\sum_i{a_{i}^{2}e^{-\mathrm{i}\phi _1}\cos \frac{\theta _i}{2}}\sin \frac{\theta _i}{2}=0.\\
\end{cases}\Rightarrow \ \ \ \begin{cases}
	\sum_i{a_{i}^{2}}=2,\\
	\sum_i{a_{i}^{2}\cos \theta _i}=0,\\
	\sum_i{a_{i}^{2}\sin \theta _i\cos \phi _i}=0,\\
	\sum_i{a_{i}^{2}\sin \theta _i\sin \phi _i}=0.\\
\end{cases}
\end{aligned}
\end{align}
We define a new set of vectors $|\eta _i\rangle =\left[ \begin{matrix}
	a_{i}^{2}\cos \theta _i&		a_{i}^{2}\sin \theta _i\cos \phi _i&		a_{i}^{2}\sin \theta _i\sin \phi _i\\
\end{matrix} \right] ^T$, and the condition in Eq.~\eqref{eq:basic_constrain} can then be reformulated as $\sum_i{|\eta _i\rangle}=0$ and $\sum_i{a_{i}^{2}}=2$. This formulation simplifies the process of constructing a POVM.

\begin{figure}[h!]
\centering
\includegraphics[width=0.55\textwidth]{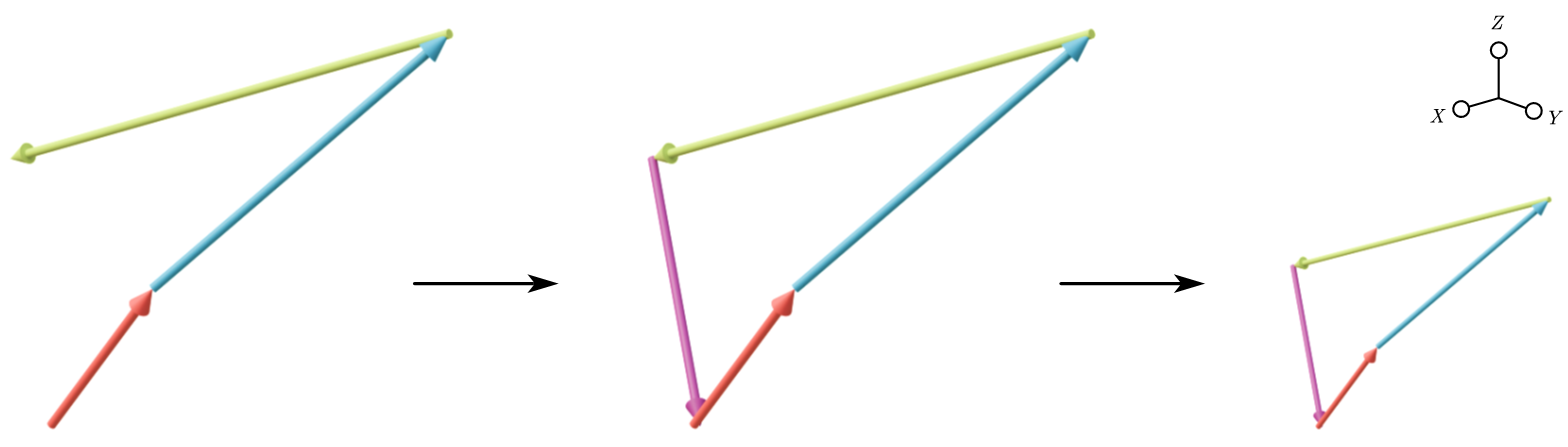}
\caption{\justifying{The construction process of a single-qubit 4-element rank-1 POVM. The vectors depicted are $|\eta_i\rangle$ rather than $|\varphi_i\rangle$. Once the parameters of the first three rows are selected, the parameters of the fourth row are determined accordingly.
The values of $a_i^{2}$ are then scaled by $2\left(\sum_{i=1}^{n-1}{a_{i}^{2}}+\left| |\eta_n\rangle \right\|\right)^{-1}$ to satisfy the condition $\sum_i{a_{i}^{2}}=2$.}}
\label{fig:povm_construction}
\end{figure}

The process of constructing a single-qubit rank-1 POVM starts by selecting parameters for the first $n-1$ vectors, where $n$ is the total number of elements in the POVM. Due to the condition $\sum_i{|\eta_i\rangle}=0$, the final vector $|\eta_n\rangle$ must be $|\eta_n\rangle =-\sum_{i=1}^{n-1}{|\eta_i\rangle}$, leaving it with no free parameters. To ensure that $\sum_i{a_i^2}=2$, we scale this closed polygon by a factor of $2({\sum_{i=1}^{n-1}{a_{i}^{2}}+\left\| |\eta _n\rangle \right\|})^{-1}$, as shown in Fig.~\ref{fig:povm_construction}. For the case where $n=4$, the initial three vectors $|\eta _i\rangle$ provide 9 parameters. The scaling constraint reduces one parameter from $a_i$, resulting in 8 parameters. This result matches our expectations for $\Gamma$, the set of parameters that determines the POVM.

For $U_{\mathrm{SIC}}$ associated with a SIC-POVM, the parameters in $\Gamma$ are further constrained by Eq.~\eqref{eq:sicpovm_conditions}. This constraint can be divided into two parts. First, $\langle \!\langle \Pi _i|\Pi _i\rangle \!\rangle=\frac{1}{d^2}$, leads to $a_i=\frac{1}{\sqrt{2}}$, eliminating all parameters related to $\alpha_i$. Second, $\langle \!\langle \Pi _i|\Pi _j\rangle \!\rangle=\frac{1}{d^2\left( d+1 \right)}$ for $i \ne j$, dictates that each pair of $\ket{\eta_i}$ and $\ket{\eta_j}$ maintain a fixed angle $\varDelta _{ij}=\arccos \left( -\frac{1}{3} \right)$. 

\begin{figure}[H]
\centering
\includegraphics[width=0.55\textwidth]{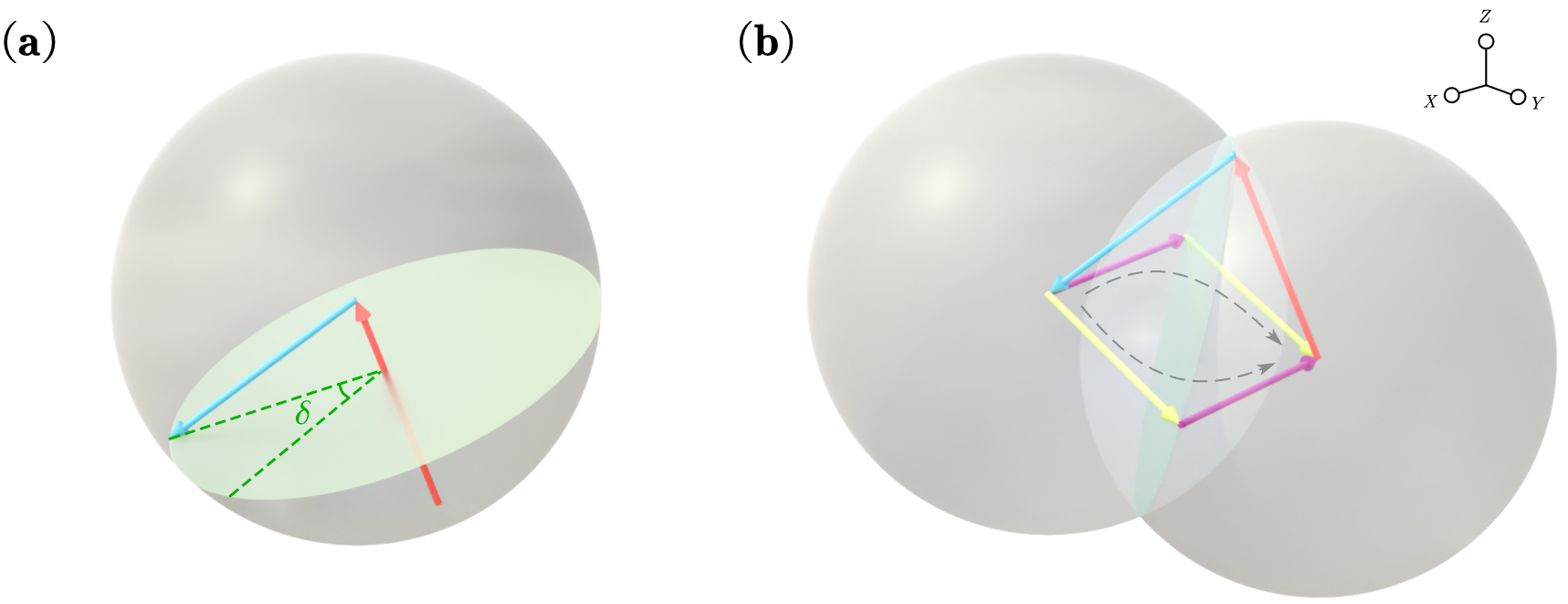}
\caption{\justifying{The construction process of a single-qubit SIC-POVM. Each $|\eta _i\rangle$ has a norm of $a_i^{2}=\frac{1}{2}$. (a) Once the parameters of the first row (red) are selected, the second row (blue) is determined by a single parameter $\delta$.
(b) The last two rows are obtained by ensuring $\sum_i{|\eta _i\rangle}=0$, which guarantees that the tail of $|\eta_3\rangle$ and the head of $|\eta_4\rangle$ lie on the intersecting circles of two spheres with $|\eta_1\rangle$ and $|\eta_2\rangle$ as their radii. The angle $\varDelta _{ij}$ ensures that there are only two possible choices for $|\eta _3\rangle$ and $|\eta _4\rangle$ (yellow and purple or vice versa).}}
\label{fig:sic_construction}
\end{figure}

While selecting $|\eta _1\rangle$
and $|\eta _2\rangle$, the condition indicates that after choosing the parameters $\theta_1$ and $\phi_1$ for $|\eta _1\rangle$, $|\eta_2\rangle$ must maintain a fixed angle $\varDelta_{ij}$ relative to $|\eta_1\rangle$. This restriction confines the tail of $|\eta_2\rangle$ to a specific circle, allowing us to introduce a new angle $\delta$ to replace the original parameters $\theta_2$ and $\phi_2$ for $|\eta_2\rangle$. For $|\eta_3\rangle$, with $\alpha_i$ fixed and the requirement to maintain $\varDelta_{ij}$ with both $|\eta_1\rangle$ and $|\eta_2\rangle$, only two possible choices remain, described by a discrete parameter $c$. Once $|\eta_3\rangle$ is determined, $|\eta_4\rangle$ is implicitly fixed as well. This construction process, illustrated in Fig.~\ref{fig:sic_construction}, is governed by 3 continuous parameters and 1 discrete parameter, consistent with Eq.~\eqref{eq:sicpovm_decomp}.

The continuous parameters $\theta_1$, $\phi_1$, and $\delta$, which define $|\eta_1\rangle$ and $|\eta_2\rangle$, can be adjusted using the single-qubit unitary $U_S = R_z(\phi_1) R_y(\theta_1) R_z(\delta)$. Applying $U_S$ to a reference unitary $U_{\mathrm{SIC}}$ allows for the generation of all other SIC-POVMs. By incorporating the discrete parameter $c$, we construct a reference $U_{\mathrm{SIC}}$ from Set 1 in Fig.~\ref{fig:2sets}, denoted as $U_{\mathrm{SIC-1}}(c, \Theta_1)$:
\begin{align}\label{eq:uk}
\begin{aligned}
U_{\mathrm{SIC-1}}(c,\Theta_1)=\frac{1}{\sqrt{2}}\left[ \begin{matrix}
	1&		0&		1&		0\\
	\frac{1}{\sqrt{3}}&		\frac{\sqrt{2}}{\sqrt{3}}&		-\frac{1}{\sqrt{3}}&		\frac{\sqrt{2}}{\sqrt{3}}\\
	\frac{1}{\sqrt{3}}&		e^{\mathrm{i}\left( -1 \right)^{c}\frac{2 \pi}{3}}\frac{\sqrt{2}}{\sqrt{3}}&		-\frac{1}{\sqrt{3}}&		-e^{\mathrm{i}\left( -1 \right)^{c}\frac{ \pi}{3}}\frac{\sqrt{2}}{\sqrt{3}}\\
	\frac{1}{\sqrt{3}}&		e^{-\mathrm{i}\left( -1 \right)^{c}\frac{2 \pi}{3}}\frac{\sqrt{2}}{\sqrt{3}}&		-\frac{1}{\sqrt{3}}&		-e^{-\mathrm{i}\left( -1 \right)^{c}\frac{ \pi}{3}}\frac{\sqrt{2}}{\sqrt{3}}\\
\end{matrix} \right] ,
\end{aligned}
\end{align}
where $\cos \left( \frac{\varDelta _{ij}}{2} \right) =\frac{1}{\sqrt{3}}$, $ \sin \left( \frac{\varDelta _{ij}}{2} \right) =\frac{\sqrt{2}}{\sqrt{3}}$, $c\in \left\{ 0,1 \right\}$.

\subsection{Relabeling Operators}\label{ap:relabel_op}
Given a 4-element POVM $\{ |\varphi _i\rangle \langle \varphi _i| \}_{i=1}^{4}$, one can construct a $V$ matrix in a specified order, that is
\begin{align}\label{eq:origin_v}
\begin{aligned}
    V=\left[ \begin{matrix}
	|\varphi _1\rangle&		|\varphi _2\rangle&		|\varphi _3\rangle&		|\varphi _4\rangle\\
\end{matrix} \right] ^{\dagger}.
\end{aligned}
\end{align}
However, the construction of the $V$ matrix is not limited to a single ordering. By rearranging the elements, an alternate matrix $V^{\prime}$ can be formulated, such as
\begin{align}\label{eq:relable_v}
\begin{aligned}
    V^{\prime}=\left[ \begin{matrix}
	|\varphi _1\rangle&		|\varphi _2\rangle&		|\varphi _4\rangle&		|\varphi _3\rangle\\
\end{matrix} \right] ^{\dagger}.
\end{aligned}
\end{align}
The matrices $V$ and $V^{\prime}$, though representing the same POVM, result from different parameter sets $\Gamma$. The difference between these matrices is captured by a unitary transformation $U_{\mathrm{rel}}$ that acts on the row space of $V$, yielding $V^{\prime}$ (i.e., $U_{\mathrm{rel}}V = V^{\prime}$). This transformation can be labeled by permutations, such as `1243', which indicates a change from the original order `1234'. Such labels help categorize some different $\Gamma$ into a unified framework.

$U_{\mathrm{rel}}$ can be implemented using CNOT and $X$ gates, resulting in 24 unique $U_{\mathrm{rel}}$. For simplicity, Arabic numerals represent different $X$ circuits, and alphabetic letters represent different CNOT circuits. There are four distinct cases for $X$ circuits:
\begin{equation}        
\includegraphics{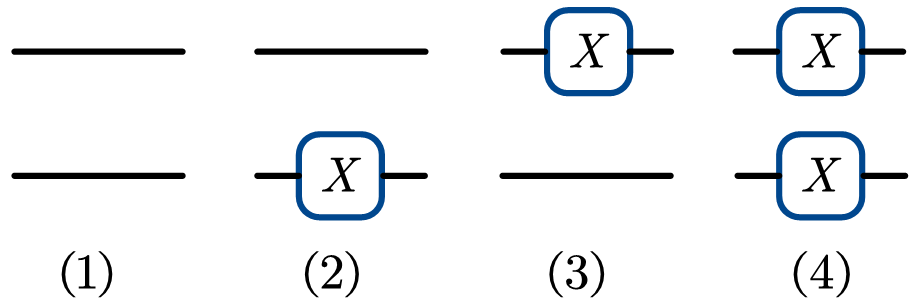}.
\end{equation}
For CNOT circuits, there are a total of $6$ different cases:
\begin{equation}        
\includegraphics{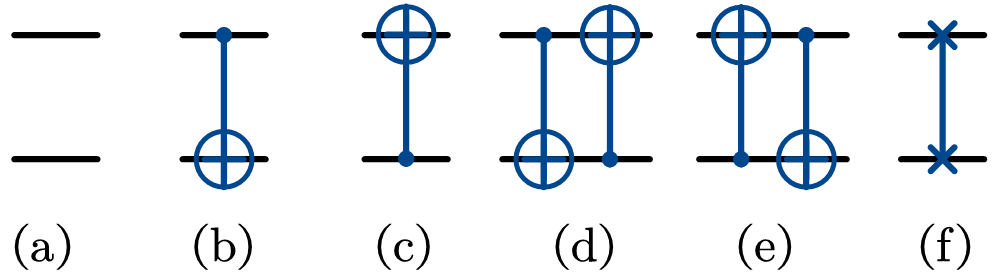},
\end{equation}
where $(\mathrm{f})$ represents a SWAP gate. By combining the $X$ circuits and CNOT circuits, we can generate unique $U_{\mathrm{rel}}$.
For example, the notation `$3\mathrm{e}$' signifies a specific combination, representing a circuit that
\begin{equation}\label{eq:3e}        
\includegraphics{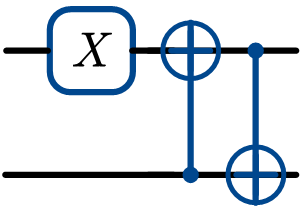}.
\end{equation}
Then, we can outline all 24 operators and their corresponding effects, detailed in Table \ref{table:Relabel Gates}.

\begin{table}[htbp]
\centering
\caption{Relabeling operators}
\label{table:Relabel Gates}
\begin{tabular}{l|cccccc}
\toprule
      & a & b & c & d & e & f \\
\midrule
1 & 1234 & 1243 & 1432 & 1342 & 1423 & 1324 \\
2 & 2143 & 2134 & 2341 & 2431 & 2314 & 2413 \\
3 & 3412 & 3421 & 3214 & 3124 & 3241 & 3142 \\
4 & 4321 & 4312 & 4123 & 4213 & 4132 & 4231 \\
\bottomrule
\end{tabular}
\end{table}

For SIC-POVMs, not all 24 relabeling operators are necessary due to the constraints specified in Eq.~\eqref{eq:sicpovm_conditions}, which ensure identical norms and a fixed angle $\varDelta _{ij}$ among the $|\varphi_i \rangle$ vectors. This allows the first two vectors to be aligned with the desired $|\varphi_1 \rangle$ and $|\varphi_2 \rangle$ using a rotation $U_S$. As a result, the relabeling operations only affect the last two vectors, which can be described by the discrete parameter $c$. Therefore, the 24 operators reduce effectively to two: `$1\mathrm{a}$'(Identity) and `$1\mathrm{b}$'(CNOT).
As introduced in the main text, if we choose the $U_{\mathrm{SIC-2}}(c=1,\Theta_2)$, i.e., the `1243' case shown in Fig.~\ref{fig:sic_two}(b), as the reference SIC-POVM, in order to build the $U_{\mathrm{SIC}}$ with $c=0$ (represented by `1234') as shown in Fig.~\ref{fig:sic_two}(a), we need to implement a CNOT gate as the $U_{\mathrm{rel}}$.

\begin{figure}[H]
\centering
\includegraphics[width=0.55\textwidth]{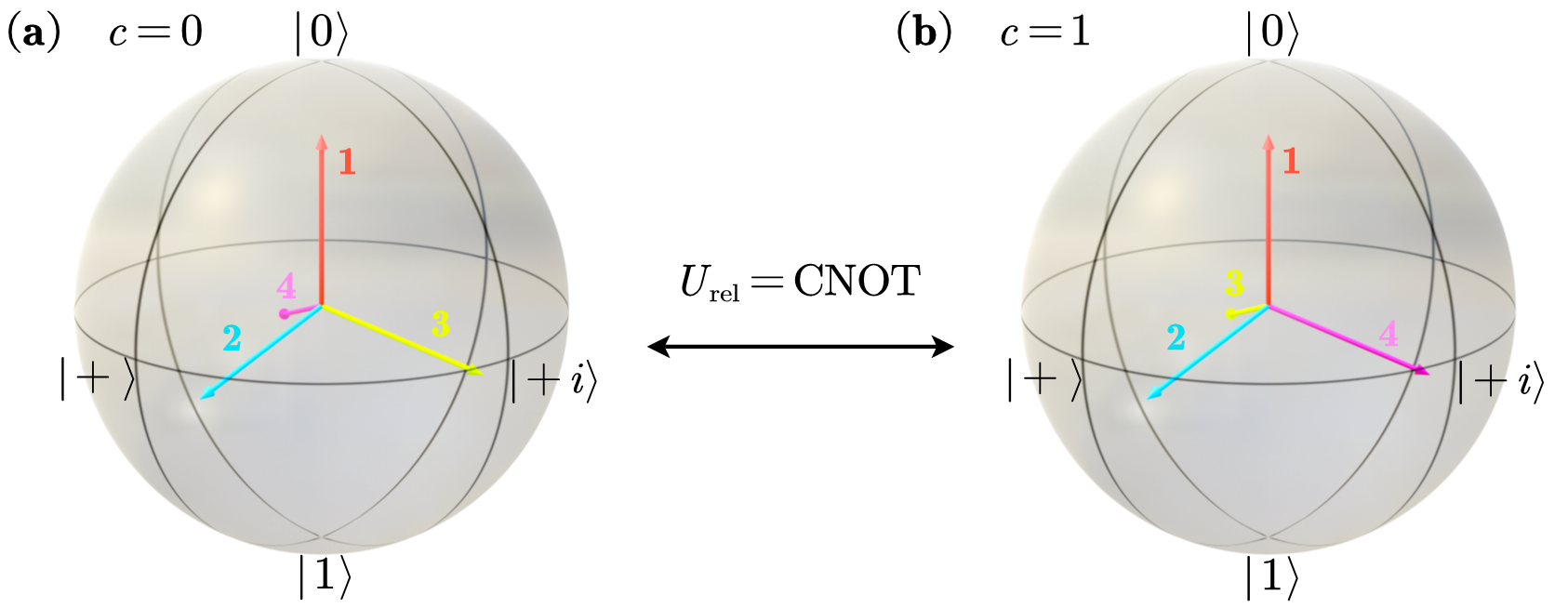}
\caption{\justifying{Two cases of SIC-POVM after adjusting $U_S$. The vectors in the Bloch sphere represent $|\varphi_i\rangle$, corresponding to $|\eta_i\rangle$. By adjusting $U_S$, one can align $|\varphi_1 \rangle$ and $|\varphi_2 \rangle$ (red and blue), consequently narrowing down $|\varphi_3 \rangle$ and $|\varphi_4 \rangle$ (yellow and purple) to two possible scenarios: (a) $c=0$ and (b) $c=1$. These two scenarios can be interconverted using a CNOT gate.}}
\label{fig:sic_two}
\end{figure}

By integrating the three subsections, we clarify the circuit structure presented in Eq.~\eqref{eq:povm_syn} and Eq.~\eqref{eq:sicpovm_decomp}. 
In practice, relabeling operations can be implemented through classical post-processing, eliminating the need to build quantum relabeling circuits such as Eq.~\eqref{eq:3e}. This approach is demonstrated in the circuit shown in Fig.~\ref{fig:prac_circ}(b).

\section{Implementing SIC-POVM by Bell Measurement}\label{ap:sic_bell}

Bell measurement is a joint measurement for two qubits, typically implemented using a quantum circuit consisting of a CNOT gate followed by a Hadamard gate, that is
\begin{equation}\label{eq:bell_mea}        
\includegraphics{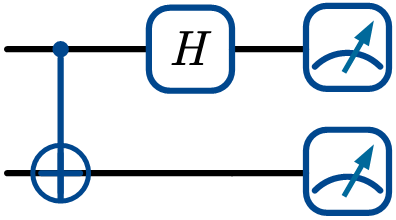}.
\end{equation}

The Bell measurement circuit in Eq.~\eqref{eq:bell_mea} transforms the measurement basis from the computational basis to the Bell basis. The Bell basis consists of four elements, each identified by a measurement outcome bitstring $b$, with different representations as illustrated in Table \ref{table:bell_basis}. Depending on the context, we can choose a preferred representation of the Bell basis.

\begin{table}[h]
\centering
\caption{Different Form of Bell Basis}
\label{table:bell_basis}
\begin{tabular}{l|ccc|c}
\toprule
 Basis & State Vector & Tensor & Tensor (Measure) & Bitstring $b$ \\
\midrule
$|\Phi ^+\rangle$  & $\frac{1}{\sqrt{2}}\left( |00\rangle +|11\rangle \right) $ & \adjustbox{valign=m}{\includegraphics[keepaspectratio]{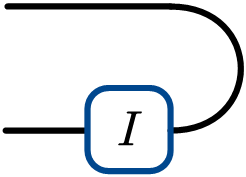}} & \adjustbox{valign=m}{\includegraphics[keepaspectratio]{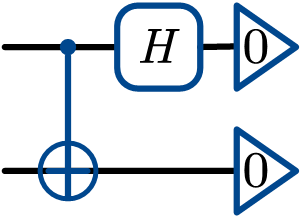}} & 00  \\
$|\Phi ^-\rangle$  & $\frac{1}{\sqrt{2}}\left( |00\rangle -|11\rangle \right) $ &  \adjustbox{valign=m}{\includegraphics[keepaspectratio]{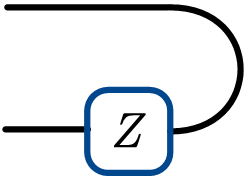}} & \adjustbox{valign=m}{\includegraphics[keepaspectratio]{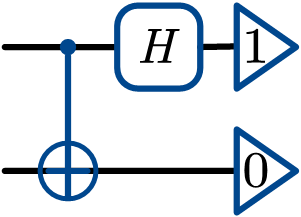}} & 10  \\
$|\Psi ^+\rangle$  & $\frac{1}{\sqrt{2}}\left( |01\rangle +|10\rangle \right) $ &  \adjustbox{valign=m}{\includegraphics[keepaspectratio]{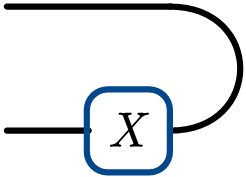}} & \adjustbox{valign=m}{\includegraphics[keepaspectratio]{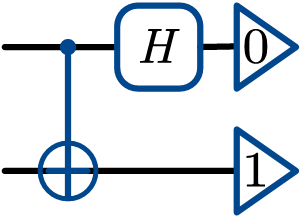}} & 01  \\
$|\Psi ^-\rangle$  & $\frac{1}{\sqrt{2}}\left( |01\rangle -|10\rangle \right) $ &  \adjustbox{valign=m}{\includegraphics[keepaspectratio]{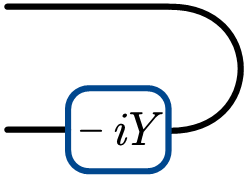}} & \adjustbox{valign=m}{\includegraphics[keepaspectratio]{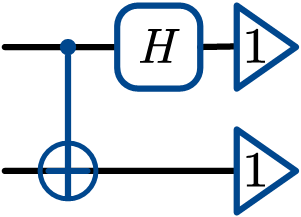}} & 11  \\
\bottomrule
\end{tabular}
\end{table}

When performing a Bell measurement on a system consisting of a target quantum state and an auxiliary qubit, the four possible outcomes can be represented by the tensor networks as
\begin{equation}\label{eq:bell_mea_a}        
\includegraphics{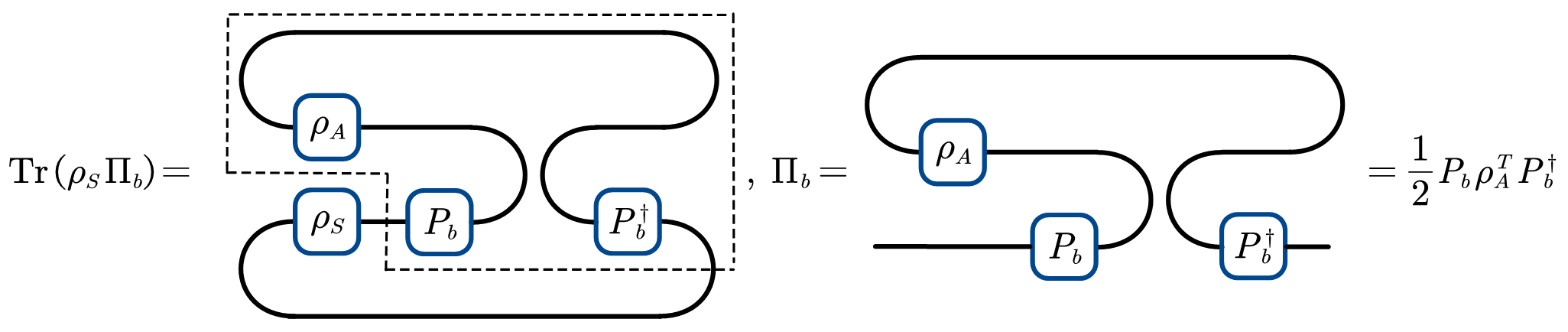}.
\end{equation}
Here, $P_{b} \in \{I,Z,X,-iY\}$, where $b$ corresponds to the measurement bitstring as referenced in Table \ref{table:bell_basis}. In this setup, the auxiliary qubit and the Bell basis together form a POVM operator. The properties of these POVM operators are determined by the auxiliary qubit, with the specific operator selected by $b$.

The Bell basis has a close connection with the Weyl–Heisenberg group covariant SIC-POVM. In a Hilbert space $\mathcal{H}_d = \mathbb{C}^d$, such a SIC-POVM can be generated from a fiducial state $|f\rangle$ \cite{fuchs2017sic,feng2022stabilizer} by applying Weyl displacement operators. These operators are defined by
\begin{align}\label{eq:dis_op}
\begin{aligned}
D_{kl}=\tau ^{kl}X^kZ^l, \ \tau =-e^{\mathrm{i}\frac{\pi}{d}},
\end{aligned}
\end{align}
where $k$ and $l$ are integers modulo $d$, with $X$ as the shift operator and $Z$ as the phase operator, defined as
\begin{align}\label{eq:wely_group}
\begin{aligned}
X = \sum_{j=0}^{n-1}{|j+1\rangle \langle j|}, \ Z = \sum_{j=0}^{n-1}{\omega _{d}^{j}|j\rangle \langle j|}, \ \omega _d= e^{\mathrm{i}\frac{2\pi}{d}}.
\end{aligned}
\end{align}
Here, $|j\rangle$ denotes the $j$th element of orthonormal basis in $\mathcal{H}_d$. Specifically, when $d=2$, these operators correspond to the single-qubit Pauli operators $X$ and $Z$.
Once the fiducial state $|f\rangle$ is determined, the SIC-POVM operators are given by
\begin{align}\label{eq:sic_set}
\begin{aligned}
 \Pi _{kl}=\frac{1}{d}|f_{kl}\rangle \langle f_{kl}|, \ |f_{kl}\rangle =D_{kl}|f\rangle .
\end{aligned}
\end{align}

When considering an auxiliary qubit in a pure state $\rho_{A}=|\psi \rangle \langle \psi |$, its transpose is $\rho_{A}^{T}=|\psi^{*} \rangle \langle \psi^{*} |$. Substituting $\rho_{A}^{T}$ into Eq.~\eqref{eq:bell_mea_a} yields $\frac{1}{2}P_{b}|\psi^{*} \rangle \langle \psi^{*}| P_{b}^{\dagger}$. Comparing this with Eq.~\eqref{eq:sic_set}, it is evident that if $|\psi^{*} \rangle$ is chosen as the fiducial state with an arbitrary phase (e.g., $e^{\mathrm{i}\alpha}|f\rangle$), the Bell measurement-based circuit realizes a SIC-POVM. The primary difference is the phase preceding $Y$, which does not affect the measurement outcome.

\section{Procedure for Determining the Parameters in Circuit}\label{ap:algo1}
In Sec.~\ref{sec:prac}, we introduced a general circuit and a practical circuit to compile an arbitrary $U_{\mathrm{SIC}}$, which requires specific parameters. In this section, we outline some convenient processes for determining these parameters.

\subsection{Parameters for Practical Circuit}\label{app:para_prac}
When compiling $U_{\mathrm{SIC}}$ using practical circuit shown in Fig.~\ref{fig:prac_circ}(b), $U_S$ and $c$ are critical as they directly influence $\Gamma$, which determines the SIC-POVM. We provide Algo.~\ref{algo:algo1} to determine the specific $U_S$ and $c$.

For convenience, we denote the target $U_{\mathrm{SIC}}$ as $U_{\mathrm{SIC}}^{(1)}$ and $U_{\mathrm{SIC-2}}(c=1,\Theta_2)$ as $U_{\mathrm{SIC}}^{(2)}$, where $U_{\mathrm{SIC-2}}(c=1,\Theta_2)$ is the updated reference $U_{\mathrm{SIC}}$ we chosen in Sec.~\ref{sec:prac}. Since $U_S$ and $c$ only relate to $V$ matrix of $U_{\mathrm{SIC}}$ (as defined in Eq.~\eqref{eq:V}), we first extract $V^{(1)}$ and $V^{(2)}$ from $U_{\mathrm{SIC}}^{(1)}$ and $U_{\mathrm{SIC}}^{(2)}$, respectively. Then $U_S$ and $c$ can be derived by solving
\begin{align}\label{eq:para_first}
\begin{aligned}
U_{\mathrm{rel}}U_{\mathrm{ph}}V^{(2)}U_S=V^{(1)},
\end{aligned}
\end{align}
where $U_{\mathrm{rel}}$ is the relabeling gate associated with $c$, and $U_{\mathrm{ph}}$ includes the phase gates with parameters $\beta_1$, $\beta_2$ and $\beta_3$ as described in Eq.~\eqref{eq:sicpovm_decomp}. Though the forms of $U_{\mathrm{rel}}$, $U_{\mathrm{ph}}$, and $U_{S}$ are known, solving this equation directly can be challenging. However, inspired by App.~\ref{ap:para_psi}, as $U_S$ only relate to the first two rows of $V$, we can simplify this problem by splitting $V^{(i)}$ into two parts, that is
\begin{align}
\begin{aligned}
V_{\mathrm{up}}^{\left( i \right)}=\left[ \begin{array}{c}
	\langle \varphi _{1}^{\left( i \right)}|\\
	\langle \varphi _{2}^{\left( i \right)}|\\
\end{array} \right] ,\ V_{\mathrm{down}}^{\left( i \right)}=\left[ \begin{array}{c}
	\langle \varphi _{3}^{\left( i \right)}|\\
	\langle \varphi _{4}^{\left( i \right)}|\\
\end{array} \right] ,
\end{aligned}
\end{align}
where $\langle \varphi _{j}^{(i)}|$, with $i \in \left\{ 1,2\right\}$ and $j\in \left\{ 1,2,3,4 \right\}$, denotes the $j$th row of $V^{(i)}$. Using $V_{\mathrm{up}}^{(i)}$, $U_S$ can be determined from
\begin{align}\label{eq:sub_para}
\begin{aligned}
R_z(\beta_2)V_{\mathrm{up}}^{\left( 2 \right)}U_S=V_{\mathrm{up}}^{\left( 1 \right)}.
\end{aligned}
\end{align}
Here, $\beta_2$ is the only parameter on the left side of $V_{\mathrm{up}}^{(2)}$, affecting the phase of $\langle \varphi_2^{(2)}|$. For simplify, we adopt a normalized vector representation: $\langle \phi _{j}^{(i)}|=\frac{1}{\sqrt{2}}\langle \varphi _{j}^{(i)}|$. We then construct a unitary matrix $U^{(i)}=|\phi _{1}^{(i)}\rangle \langle 0|+|\phi _{1}^{(i)\bot}\rangle \langle 1|e^{\mathrm{i}\theta}$, such that $\langle \phi _{1}^{(i)}|U^{(i)}= \langle 0|$ and $\langle \phi _{2}^{(i)}|U^{(i)}= \langle \phi _{2}^{(i)\prime}|$. Here, we opt $\theta=0$ and let $|\phi _{1}^{(i)\bot}\rangle$ be an arbitrary normalized vector orthogonal to $|\phi _{1}^{(i)}\rangle$. We can further express $\langle \phi _{2}^{(i)\prime}|$ as
\begin{align}\label{eq:rewritephi2}
\begin{aligned}
\langle \phi _{2}^{(i)'}|=ae^{\mathrm{i}\alpha _{i1}}\langle 0|+\sqrt{1-a^2}e^{\mathrm{i}\alpha _{i2}}\langle 1|,
\end{aligned}
\end{align}
where $a$ is a real number that is not of concern, and $\alpha_{ik}$, with $k\in \{ 1,2 \}$, are the phases preceding $\ket{0}$ and $\ket{1}$ of $\langle \phi _{2}^{(i)\prime}|$, respectively. $\beta_2$ can then be straightforwardly determined by $\beta_2=\alpha_{11}-\alpha_{21}$. With the value of $\beta_2$, we have
\begin{align}\label{eq:rotatebeta2}
\begin{aligned}
R_{z}(\beta_2) \left[ \begin{array}{c}
	\langle 0|\\
	\langle \phi _{2}^{(2)'}|\\
\end{array} \right] =\left[ \begin{array}{c}
	\langle 0|\\
	\langle \phi _{2}^{(2)''}|\\
\end{array} \right] ,
\end{aligned}
\end{align}
where $\langle \phi _{2}^{(2)''}|=ae^{\mathrm{i}\alpha _{11}}\langle 0|+\sqrt{1-a^2}e^{\mathrm{i}\left( \alpha _{22}+\alpha _{11}-\alpha _{21} \right)}\langle 1|$. Subsequently, by employing a unitary $U_r=|0\rangle \langle 0|+e^{\mathrm{i}\left( \alpha _{12}-\alpha _{22}-\alpha _{11}+\alpha _{21} \right)}|1\rangle \langle 1|$, we achieve $\langle \phi _{2}^{(2)''}|U_r=\langle \phi _{2}^{(1)'}|$. Combining all the operations, we obtain $R_{z}(\beta)V_{\mathrm{up}}^{(2)}U^{(2)}U_r(U^{(1)})^{\dagger}= V_{\mathrm{up}}^{(1)}$. Comparing with Eq.~\eqref{eq:sub_para}, it is evident that
\begin{align}\label{eq:deter_us}
\begin{aligned}
U_S=U^{(2)} U_{r} (U^{(1)})^{\dagger}.
\end{aligned}
\end{align}

Using $V_{\mathrm{down}}^{(i)}$, we can determined $c$ by
\begin{align}
\begin{aligned}
U_{\mathrm{rels}}U_{\mathrm{phs}}V_{\mathrm{down}}^{(2)}U_S = V_{\mathrm{down}}^{(1)},
\end{aligned}
\end{align}
where $U_{\mathrm{phs}}$ and $U_{\mathrm{rels}}$ denote the phase adjustments and relabeling operator performed within $V_{\mathrm{down}}^{(i)}$. These operators are defined as
\begin{align}\label{eq:subops}
\begin{aligned}
U_{\mathrm{phs}}&=\left[ \begin{matrix}
	e^{\mathrm{i}\beta _1}&		0\\
	0&		e^{\mathrm{i}\left( \beta _1+\beta _2+\beta _3 \right)}\\
\end{matrix} \right], \
U_{\mathrm{rels}}&=I\,\,\mathrm{or}\,\,X,
\end{aligned}
\end{align}
and are related to $U_{\mathrm{ph}}$ and $U_{\mathrm{rel}}$ by $U_{\mathrm{ph}}=R_z(\beta_2)\oplus U_{\mathrm{phs}}$ and $U_{\mathrm{rel}}=I\oplus U_{\mathrm{rels}}$ (yeilds $I$ or CNOT). With $U_S$ determined from Eq.~\eqref{eq:deter_us}, we have 
\begin{align}\label{eq:upr}
\begin{aligned}
U_{\mathrm{pr}}=U_{\mathrm{rels}}U_{\mathrm{phs}}=V_{\mathrm{down}}^{(1)}(V_{\mathrm{down}}^{(2)}U_{S})^{-1}.
\end{aligned}
\end{align}
Here we denote $V_{\mathrm{down}}^{(1)}(V_{\mathrm{down}}^{(2)}U_{S})^{-1}$ as $U_{\mathrm{pr}}$. The value of $c$ is directly related to the form of $U_{\mathrm{rels}}$, while $U_{\mathrm{phs}}$ is always a diagonal matrix and does not affect the matrix form. Therefore, $c$ can be determined by examining the form of $U_{\mathrm{pr}}$. If $U_{\mathrm{pr}}$ is a diagonal matrix, no relabeling operation is needed, which implies that the value of $c$ for target unitary $U_{\mathrm{SIC}}^{(1)}$ remains the same as in $U_{\mathrm{SIC}}^{(2)}$, so $c = 1$. Conversely, if $U_{\mathrm{pr}}$ is not a diagonal matrix, then $c = 0$.

\subsection{Parameters for General Circuit}
To optimize the compilation of a $U_{\mathrm{SIC}}$, one only needs to implement the practical circuit shown in Fig.~\ref{fig:prac_circ}(b), which involves just $U_S$ and $c$. However, if we want to decompose $U_{\mathrm{SIC}}$ using the general circuit illustrated in Fig.~\ref{fig:prac_circ}(a), which has the completely same matrix representation as the target $U_{\mathrm{SIC}}$, additional parameters $\beta_1$, $\beta_2$, $\beta_3$, and $Q$ are required. These parameters can be obtained from the outcomes and intermediate variables of Algo.~\ref{algo:algo1}, as detailed in Algo.~\ref{algo:algo2}.

\begin{algorithm}[H]
    \caption{Parameters for General Circuit}\label{algo:algo2}
    \begin{algorithmic}[1]
    \Require
     Target $U_{\mathrm{SIC}}$, denoted as $U_{\mathrm{SIC}}^{(1)}$; $U_{\mathrm{SIC-2}}(c=1,\Theta_2)$, denoted as $U_{\mathrm{SIC}}^{(2)}$
    \Ensure
    $\beta_1$, $\beta_2$, $\beta_3$, $Q$
    \State Execute Algo.~\ref{algo:algo1}, get outputs: $U_S$, $c$ and intermediate variables: $\left\{\alpha_{11},\alpha_{12},\alpha_{21},\alpha_{22} \right\}$, $U_{\mathrm{pr}}$
    \State Get $\beta_2=\alpha_{11}-\alpha_{21}$
    \State Get $\beta_1=\mathrm{arg}(\sum_i{{u_{\mathrm{pr}}}_{1}^{i}})$, $\beta_3=\mathrm{arg}(\sum_i{{u_{\mathrm{pr}}}_{2}^{i}})-\beta_1-\beta_2$, where ${u_{\mathrm{pr}}}_{j}^{i}$ is the $i$th row and $j$th column entry of $U_{\mathrm{pr}}$
    \State Extract $W^{(i)}$ for $U_{\mathrm{SIC}}^{(i)}$, with $i \in \left\{ 1,2\right\}$, as Eq.~\eqref{eq:povm_elements_matrix}
    \State Update $W^{(2)} \gets C_{R_z}(\beta_3)[R_z(\beta_1)\otimes R_z(\beta_2)]W^{(2)}U_{S}$
    \If {$c=0$}
    \State Update $W^{(2)} \gets C_X W^{(2)}$
    \EndIf
    \State Get $Q=(W^{(2)})^{+}W^{(1)}$
    \end{algorithmic}
\end{algorithm}

To determine the parameters $\beta_1$, $\beta_2$, and $\beta_3$, we rely on the values of $\alpha_{ik}$ in Eq.~\eqref{eq:rewritephi2} and $U_{\mathrm{pr}}$ in Eq.~\eqref{eq:upr}. First, $\beta_2$ controls the phase of the second row of the $V$ matrix and is derived directly from $\beta_2 = \alpha_{11} - \alpha_{21}$, as discussed in Eq.~\eqref{eq:rotatebeta2}.

Next, the parameters $\beta_1$ and $\beta_3$ correspond to the third and fourth rows of $V$, which are represented by $U_{\mathrm{phs}}$ in Eq.~\eqref{eq:subops}. Since $U_{\mathrm{phs}}$ is always a diagonal matrix and $U_{\mathrm{rels}}$ only affects the row order, the first column's non-zero element of $U_{\mathrm{pr}} = U_{\mathrm{rels}}U_{\mathrm{phs}}$ gives $e^{\mathrm{i}\beta_1}$, while the non-zero element in the second column equals $e^{\mathrm{i}(\beta_1 + \beta_2 + \beta_3)}$. Note that each column of $U_{\mathrm{pr}}$ contains only one non-zero element, $\beta_1$ and $\beta_3$ can be calculated as follows:
\begin{align}
\begin{aligned}
\beta_1=\mathrm{arg}(\sum_i{{u_{\mathrm{pr}}}_{1}^{i}}),\ \beta_3=\mathrm{arg}(\sum_i{{u_{\mathrm{pr}}}_{2}^{i}})-\beta_1-\beta_2,
\end{aligned}
\end{align}
where ${u_{\mathrm{pr}}}_{j}^{i}$ is the entry in the $i$th row and $j$th column of $U_{\mathrm{pr}}$. This approach allows us to determine $\beta_1$ and $\beta_3$ without considering the specific type of $U_{\mathrm{rels}}$.

Since $Q$ is related to the $W$ matrix of $U_{\mathrm{SIC}}$, we first extract $W^{(1)}$ and $W^{(2)}$ from $U_{\mathrm{SIC}}^{(1)}$ and $U_{\mathrm{SIC}}^{(2)}$, respectively. Once the other parameters have been determined, the next step is to ensure that $W^{(2)}$ is brought into the same column space as $W^{(1)}$. To achieve this, we first update $W^{(2)}$ by
\begin{align}
\begin{aligned}
W^{(2)} \gets C_{R_z}(\beta_3)[R_z(\beta_1)\otimes R_z(\beta_2)]W^{(2)}U_{S}
\end{aligned}
\end{align}
and apply the corresponding $U_{\mathrm{rel}}$. After this update, with $W^{(2)}$ now aligned to the column space of $W^{(1)}$, the unique matrix $Q$ can be calculated as
\begin{align}
\begin{aligned}
Q=(W^{(2)})^{+}W^{(1)},
\end{aligned}
\end{align}
where $(W^{(2)})^{+}$ denotes the pseudo-inverse of $W^{(2)}$.
\end{appendix}

\end{document}